\documentclass[%
 reprint,
 amsmath,amssymb,
 aps,
longbibliography
]{revtex4-2}

\usepackage{graphicx}
\usepackage{dcolumn}
\usepackage{bm}
\usepackage[colorlinks=true]{hyperref}
\usepackage{multirow}
\usepackage{amsthm}
\usepackage{enumitem} 
\usepackage{cancel} 

\usepackage{xifthen}
\usepackage[caption=false]{subfig}

\begin{document}
\newcommand{\order}[1]{O\left(#1\right)}
\newcommand{\bigtheta}[1]{\Theta\left(#1\right)}
\newcommand{\littleo}[1]{o\left(#1\right)}
\newcommand{\bigomega}[1]{\Omega\left(#1\right)}
\newcommand{\littleomega}[1]{\omega\left(#1\right)}
\newcommand{\add}[1]{\textcolor{blue}{#1}}

\newcommand{\mat}[1]{\mathbf{#1}}
\newcommand{\vect}[1]{\boldsymbol{#1}}
\newcommand{\ones}[1]{\vect{1}_{#1}}
\newcommand{\zeros}[1]{\vect{0}_{#1}}
\newcommand{\eyes}[1]{\mat{I}_{#1}}

\newcommand{\expint}[1]{\operatorname{Ei}\left(#1\right)}

\newcommand{\idx}[2]{\left[#1\right]_{#2}}
\newcommand{\idxp}[2]{\left(#1\right)_{#2}}
\newcommand{\abs}[1]{\left\lvert#1\right\rvert}

\newcommand{\avg}[1]{\left\langle#1\right\rangle}
\newcommand{\cond}[2]{\left.#1\,\middle\vert\,#2\right.}
\newcommand{\condavg}[2]{\left\langle#1\,\middle\vert\,#2\right\rangle}
\newcommand{\prob}[1]{\mathbb{P}\left(#1\right)}
\newcommand{\condprob}[2]{\mathbb{P}\left(#1\,\middle\vert\,#2\right)}
\newcommand{\expect}[1]{\mathbb{E}\left[#1\right]}
\newcommand{\expectwrt}[2]{\mathbb{E}_{#1}\left[#2\right]}
\newcommand{\iversion}[1]{\left[#1\right]}

\newcommand{\bernoulli}[1]{\operatorname{Bernoulli}\left(#1\right)}
\newcommand{\diag}[1]{\operatorname{diag}\left(#1\right)}

\newcommand{\integernonneg}{\mathbb{Z}_{\ge 0}}
\newcommand{\integerpos}{\mathbb{Z}_{>0}}
\newcommand{\realnonneg}{\mathbb{R}_{\ge 0}}
\newcommand{\realnonnegk}[1]{\mathbb{R}^{#1}_{\ge 0}}
\newcommand{\real}{\mathbb{R}}
\newcommand{\mlog}{\operatorname{mlog}}
\newcommand{\naturals}{\mathbb{N}}
\newcommand{\sgn}{\operatorname{sgn}}
\newcommand{\norm}[1]{\left\lVert#1\right\rVert}
\newcommand{\normk}[2]{\left\lVert#1\right\rVert_#2}

\newcommand{\psiaf}{\psi}
\newcommand{\Psiaf}{\mat{\Psi}}
\newcommand{\omegaaf}{\omega}
\newcommand{\Omegaaf}{\mat{\Omega}}
\newcommand{\psiaaf}{\widetilde{\psi}}
\newcommand{\Psiaaf}{\widetilde{\mat{\Psi}}}
\newcommand{\omegaaaf}{\widetilde{\omega}}
\newcommand{\Omegaaaf}{\widetilde{\mat{\Omega}}}
\newcommand{\psicf}{\underline{\psi}}
\newcommand{\Psicf}{\underline{\mat{\Psi}}}
\newcommand{\omegacf}{\underline{\omega}}
\newcommand{\Omegacf}{\underline{\mat{\Omega}}}
\newcommand{\psiacf}{\underline{\widetilde{\psi}}}
\newcommand{\Psiacf}{\underline{\widetilde{\mat{\Psi}}}}
\newcommand{\omegaacf}{\underline{\widetilde{\omega}}}
\newcommand{\Omegaacf}{\underline{\widetilde{\mat{\Omega}}}}
\newcommand{\approxnaive}{\stackrel{\ast}{\approx}}

\newcommand{\degree}{d}
\newcommand{\degreeout}{d^+}
\newcommand{\degreein}{d^-}
\newcommand{\meandegree}{\widehat{\degree}}
\newcommand{\meandegreeout}{\widehat{\degreeout}}
\newcommand{\meandegreein}{\widehat{\degreein}}
\newcommand{\avgdegree}{\avg{\degree}}
\newcommand{\poisson}[1]{\mathrm{Poisson}\left(#1\right)}
\newcommand{\uniform}[2]{\mathrm{Uniform}\left(#1,#2\right)}
\newcommand{\stduniform}{\uniform{0}{1}}

\newcommand{\pperp}{\perp\!\!\!\!\!\!\perp}
\newcommand{\inde}[2]{#1\pperp#2}
\newcommand{\condinde}[3]{\inde{#1}{#2}\,\vert\,#3}
\newcommand{\hcancel}[1]{\setbox0=\hbox{$#1$}\rlap{\raisebox{.6\ht0}{\rule{\wd0}{0.4pt}}}#1}

\newtheorem{lemma}{Lemma}
\newtheorem{theorem}{Theorem}
\newtheorem{corollary}{Corollary}[lemma]
\newenvironment{corollarygeoperc}{\renewcommand\thecorollary{1.1}\corollary}{\endcorollary}

\makeatletter
\newcommand*{\addFileDependency}[1]{
\typeout{(#1)}
\@addtofilelist{#1}
\IfFileExists{#1}{}{\typeout{No file #1.}}
}\makeatother

\newcommand*{\myexternaldocument}[1]{%
\externaldocument{#1}%
\addFileDependency{#1.tex}%
\addFileDependency{#1.aux}%
}

\preprint{APS/123-QED}

\title{Geodesic Length Distribution in Sparse Network Ensembles}

\author{Sahil Loomba}
 \email{sloomba@mit.edu}
 \altaffiliation[also at ]{MIT IDSS}
\author{Nick S. Jones}%
 \email{nick.jones@imperial.ac.uk}
\altaffiliation[also at ]{Imperial I-X and EPSRC Centre for the Mathematics of Precision Healthcare}
\affiliation{%
 Department of Mathematics, Imperial College London
}%
\date{March 3, 2025}

\begin{abstract}
A key task in the study of networked systems is to derive local and global properties that impact connectivity, synchronizability, and robustness; computing shortest paths or geodesics yields measures of network connectivity that can explain such phenomena. We derive an analytic distribution of geodesic lengths on the giant component in the supercritical regime---when the giant component exists---or on small components in the subcritical regime, of any sparse (and possibly directed) network with conditionally independent edges, in the infinite-size limit. We provide specific results for widely used network models like stochastic block models, dot product graphs, random geometric graphs, and sparse graphons. The survival function of the geodesic length distribution possesses a simple closed-form expression which is asymptotically tight for finite lengths, has a natural interpretation of traversing independent geodesics in the network, and delivers novel insight into the aforementioned network families.
\end{abstract}

\maketitle


\section{Introduction}

Networks have increasingly become important abstractions for representing a plethora of real-world complex systems. Often, network structures are not entirely known, and are inferred from partial observations of connections, or of signals generated by the system \cite{newman2018inference, goldenberg2010inference, peixoto2019dynamics, godoylorite2021socialising, hoffmann2020communityunobserved}. Inference must account for confounding by many sources of uncertainty, especially measurement errors, which has resulted in a surge of statistical inference methods that yield a statistical model for the network \cite{peixoto2018reconstructing}. Moreover, many real-world networks, such as social networks of entire societies, tend to contain a large number of nodes. Estimating properties of very large networks can be computationally prohibitive, especially if one is interested in higher-order properties such as the size of the giant component, or the average geodesic length between any two nodes in the networks. Both of these have impact on overall network connectivity, network robustness, percolation properties and system synchronizability \cite{cohen2000percolation, callaway2000percolation, newman2001random, nishikawa2003sync}, which have been studied deeply for many physical, biological and social systems. The principal motivation for our work is to analytically estimate such network properties by having access either to just an expectation of network realizations, or to a statistical network model, without either simulating networks or obtaining an error-free observation of one. In particular, we focus on geodesic properties of a network, and establish a distribution of geodesic lengths on the giant component in the supercritical regime---when the giant component exists---or on small components in the subcritical regime---when the giant component does not exist.

\paragraph*{Problem setting} We consider any network model wherein edges between any two nodes are independent of one another, perhaps after conditioning on some auxiliary node variables. Many popular statistical network models are instances of this kind---like the stochastic block model (SBM), random dot product graph (RDPG), random geometric graph (RGG) or more generally any inhomogeneous random graph. Those explicitly encoding higher-order interactions and dependencies such as exponential random graph models \cite{robins2007ergm} are excluded from this study. We assume network sparsity---in a manner that implies that the number of edges varies linearly in the number of nodes---and no ``bottlenecks''---in the sense that, if it exists, the size of the set of nodes that have positive connection probabilities from and to any given pair of nodes scales with the network size.

\paragraph*{Prior work} Previous results on the geodesic length distribution (GLD) have focused almost exclusively on the simplest model of Erd\H{o}s--R\'{e}nyi (ER) graphs \cite{blondel2007distance, katzav2015analytical}, or average geodesic lengths in degree-configuration models \cite{vanderhofstad2005distanceconfigmodel, vanderhoorn2018distancedirectedconfig} and scalar latent variable models \cite{fronczak2004average}, and some of these results display appreciable discrepancies between theory and empirics in the tail-end of the distribution---especially for networks with small degrees \cite{blondel2007distance, katzav2015analytical}. Recent work on Gaussian RGGs \cite{rastelli2016gaussianlpm} builds on the results for scalar latent variable models \cite{fronczak2004average} and, as a result, shares their limitations. Related work has determined the distribution of picking a path between two nodes in a given network \cite{franccoisse2017bagofpaths}, but this approach does not directly model the distribution of geodesic \emph{lengths} between node pairs. Branching processes have been used to describe the distribution of geodesic lengths in small-world networks \cite{barbourreinert2001smallworlds, barbourreinert2004smallworldscorrection, barbourreinert2006discretesmallworlds, barbourreinert2013gossipsmallworlds} and bipartite SBMs \cite{barbourreinert2011bipartitesbm} in terms of the limiting random variable associated with the process. However, the accuracy of approximations entailed by the process relies on the dominance of the largest eigenvalue of its associated mean matrix which precludes, for instance, network models that are close to periodic. (Although, we remark that multipartite branching processes allow for progress in specific cases of exact periodicity, as in the bipartite SBM \cite{barbourreinert2011bipartitesbm}.) Notably, previous results are unable to generalize to distances on small components in the subcritical regime. Description of the full GLD for a general family of network models has been a desirable but elusive objective \cite{jackson2017socialnetworkeconomics} before this work.

\paragraph*{Our contributions} We derive a set of recursive equations in the asymptotic limit that stipulate the full GLD between any two nodes on the giant component of a network in the supercritical regime, and on small components otherwise. The proposed approach provides, to the best of our knowledge, the most accurate description of the distribution of geodesic lengths between any node pair for a very general class of (possibly directed) sparse networks that works in both the supercritical and subcritical regimes. We determine a closed-form expression for the GLD's survival function which is tight for finite lengths in the asymptotic limit, and has a natural interpretation of traversing independent walks in the network. The closed form is specified by an iterated integral operator defined over functions on the node space, wherein the operator's kernel function encodes the probability of an edge existing between two nodes. The integral operator is analogous to the expected adjacency matrix in the finite-dimensional setting, and summarizes the dependence of a node function on all other nodes in the network. If the kernel is symmetric, it leads to an expression for the GLD in terms of the spectral decomposition of the integral operator. Under specific scenarios, our method recovers known results on geodesics, such as the small-world property of ER graphs \cite{erdos1960evolution, albert2002networks}, or the ultrasmall property of Barab\'{a}si--Albert (BA) graphs \cite{cohen2003ultrasmall}. We provide new results for the illustrative models considered namely SBMs, Gaussian RGGs, RDPGs and sparse graphons, and also apply them to real-world networks. Prior work in the field has produced results concerning \emph{specific} phenomena on \emph{specific} network models. However, our approach to the GLD unifies the study of shortest paths, and related phenomena of connectedness and centralities discussed in our follow-up work \cite{loomba2024geodesics2}, into one theoretical framework.

\paragraph*{Overview} The article is organized as follows. In Sec. \ref{sec:spd} we describe our probabilistic framework and obtain the GLD in the general setting of networks sampled from a sparse ``ensemble average'' network model. Then, in Sec. \ref{sec:general_graphs}, we consider this formalism in general random network families wherein nodes are additionally located in some arbitrary node space. In Sec. \ref{sec:geodesics_specific} we highlight specific cases of popularly used statistical network models and real-world networks. We conclude in Sec. \ref{sec:geo1_discussion} with a summary of our results and limitations of this framework.

\section{\label{sec:spd}Sparse ensemble average network}
We first derive the distribution of geodesic lengths between two nodes in a network using a recursive approach. In Sec. \ref{sec:spd_supercritical} we consider the supercritical regime, describe a lemma that permits the construction of a geodesic via intervening nodes, and discuss the sparsity condition required to generate a set of recursive equations. To supply the initial condition, in Sec. \ref{sec:perc_prob} we derive the percolation probability of a node. In Sec. \ref{sec:spd_subcritical}, the formalism is extended to the subcritical regime. Then in Sec. \ref{sec:closedform_bound}, we extract a closed-form expression for the GLD. 

\paragraph*{Setup} We consider $n$ nodes indexed by the set $[n]\triangleq\{1,2,\dots,n\}$ in a network without self-loops represented by the $n\times n$ adjacency matrix $\mat{A}$. That is, for two nodes indexed by $i\in[n]$ (``source'') and $j\in[n]\setminus\{i\}$ (``target''), $A_{ij}=1$ if there is an edge from $i$ to $j$ and $A_{ij}=0$ otherwise, and $A_{ii}=0$. We assume knowledge of the expectation in the ensemble average sense, that is, access to the expected adjacency matrix $\avg{\mat{A}}$; we use the notation $\avg{\cdot}$ when averaging over the ensemble. We consider a network to be a sparse ensemble average network (SEAN) if it is generated from a network model for which the following assumptions hold.
\begin{enumerate}
    \item \emph{Conditionally independent edges:} Conditioned on the identity of a node pair, the existence of an edge between them is independent of the edge between another node pair:
    \begin{equation}
        \label{eq:ciem}
        A_{ij}\sim\bernoulli{\avg{A_{ij}}}.
    \end{equation}
    \item \emph{Sparsity:} If a node pair has a non-zero connection probability, then it is appropriately small:
    \begin{equation}
        \label{eq:sparsity_constraint}
        \avg{A_{ij}}=\bigtheta{n^{-1}}\textrm{ or }0\textrm{ \footnote{$f(n)=\bigtheta{g(n)}$ if $\limsup_{n\to\infty}\frac{f(n)}{g(n)}<\infty$ and $\liminf_{n\to\infty}\frac{f(n)}{g(n)}>0$.}},
    \end{equation}
    and $\avg{A_{ii}}=0$ (i.e. no self-loops \footnote{Our formalism works just as well for a network with loops as long as loops are added sparsely. That is, the chance of a loop for any given node is either $\bigtheta{n^{-1}}$ or $0$, and so we can asymptotically ignore loops, as we do here.}).
    \item \emph{No bottlenecks:} If there is a set of nodes with non-zero connection probabilities from and to a given node pair, then its size is non-vanishing:
    \begin{equation}
        \label{eq:nonultrabottlenecked}
        \abs{\{k\in[n]\setminus\{i,j\}: \avg{A_{ik}}>0, \avg{A_{kj}}>0\}} = \bigomega{n}\textrm{ or }0 \textrm{ \footnote{$f(n)=\bigomega{g(n)}$ if $\liminf_{n\to\infty}\frac{f(n)}{g(n)}>0$.}}.
    \end{equation}
    \item  \emph{Asymptotic limit:} The network size is large: $n\gg 1.$
\end{enumerate}
The corresponding network model is said to be a SEAN model. For directed networks, all edges are added independently of each other. For undirected networks, we additionally enforce $\avg{\mat{A}}=\avg{\mat{A}}^T$ and, without loss of generality, use node indices as an arbitrary ordering: for $i<j$, $A_{ij}$ is generated independently from Eq. \eqref{eq:ciem}, and set $A_{ji}= A_{ij}$. We remark that Eqs. \eqref{eq:ciem} and \eqref{eq:sparsity_constraint} imply that the network is sparse in the sense that nodes have an asymptotically bounded expected degree. We emphasize that Eq. \eqref{eq:nonultrabottlenecked} is a technical condition that is not required for most standard network models---see Sec. \ref{sec:general_graphs}---as they trivially satisfy it. Let $\lambda_{ij}\in\integernonneg$ be the random variable denoting the length of a shortest path or geodesic from node $i$ to node $j$.

\subsection{\label{sec:spd_supercritical}Supercritical regime}

In this section we focus on the supercritical regime. We say that a node pair $(i,j)$ is supercritical if asymptotically there can exist a giant component---a connected component whose size is of the order of the network size, i.e. $\bigtheta{n}$---such that there exists a path from $i$ to $j$ going via nodes on that giant component. For an undirected network, if node $i$ can reach node $j$ on a giant component, then it can reach (and be reached from) \emph{every} other node on that giant component: being supercritical is a node-level property. For directed networks, it is not necessary for a (directed) path to exist from $i$ to $j$, even if one exists from $j$ to $i$, which requires consideration of giant in-/out-components specific to a node pair: being supercritical is a node pair-level property. Given this non-trivial difference, in the main text we mostly consider networks in an undirected setting, with results for directed networks reserved for Appendix \ref{sec:apdx_asymmetric}.

Without loss of generality, we assume that $\avg{\mat{A}}$ is not permutation similar to a block diagonal matrix, i.e. it is irreducible \footnote{If $\avg{\mat{A}}$ were reducible, there would exist node subsets that can never have edges between them, and we can simply consider the geodesics separately for subgraphs induced by those node subsets. This assumption is not necessary for our formalism and can be relaxed, but it simplifies the exposition since if $\avg{\mat{A}}$ is irreducible, then asymptotically there can only exist a \emph{unique} giant component.}. Let $\phi_i$ be the event that node $i$ is on the giant component, and $\neg\phi_i$ be the event that it is not. We consider the distribution of $\lambda_{ij}$ conditioned on the source node $i$ being on the giant component. It is useful to define matrices $\Psiaf_l, \Omegaaf_l$ encoding the survival function and conditional probability mass function (PMF) of the GLD respectively:
\begin{subequations}
\label{eq:def_omega_psi}
\begin{align}
    \label{eq:def_psi}
    [\Psiaf_l]_{ij}&\triangleq \condprob{\lambda_{ij}>l}{\phi_i},\\
    \label{eq:def_omega}
    [\Omegaaf_l]_{ij}&\triangleq \condprob{\lambda_{ij}=l}{\lambda_{ij}>l-1,\phi_i},
\end{align}
\end{subequations}
which we refer to as the ``survival function matrix'' and ``conditional PMF matrix'' respectively. It is important to note here that the geodesic length distribution is not a \emph{single} distribution but instead a (possibly continuous) set of distributions indexed by node pairs (or node pair locations, as in Sec. \ref{sec:general_graphs}). We use the notation $[\cdot]_{ij}$ to refer to the $(i,j)^\mathrm{th}$ element of a matrix to avoid any confusion where it might arise (when the matrix has an associated super/subscript or is a product or power of matrices). We remark that, due to sparsity, the conditional PMF matrix $\Omegaaf_l$ is asymptotically well-defined for any $l\in\integernonneg$ \footnote{Because the network is sparse, node $j$ can have an edge to any other node with probability strictly less than $1$, and therefore there is a strictly positive---although possibly vanishing---probability of $j$ to be entirely isolated and therefore of the geodesic length $\lambda_{ij}$ being infinite, and thus being longer than $l-1$ for any finite $l$.}.

\paragraph*{Recursive setup} Since the geodesic being longer than $l$ necessarily implies that it is longer than $l-1$, it is convenient to recursively model for the non-existence of a geodesic from $i$ to $j$ up to some length $l$:
\begin{equation}
    \label{eq:spd_recursion_0}
    \begin{split}
        \condprob{\lambda_{ij}>l}{\phi_i} =\ &\condprob{\lambda_{ij}>l}{\lambda_{ij}>l-1,\phi_i}\\&\times\condprob{\lambda_{ij}>l-1}{\phi_i}.
    \end{split}
\end{equation}
We can write the first factor of the recursive equation as the conditional probability of no path of length $l$ existing between $i,j$, which can be written by accounting for the non-existence of paths of length $l$ from $i$ to $j$ via any node $u\in[n]\setminus \{i,j\}$ such that there is a direct edge from $u$ to $j$. Due to sparsity and no bottlenecks, the probability of existence of a geodesic of any finite length $l$ is asymptotically of order $\bigtheta{n^{-1}}$ or $0$ while the contribution of any given node to a possible geodesic between a given node pair is of order $\bigtheta{n^{-2}}$ or $0$ (see Lemma \ref{lemma:vanishingbridge} in Appendix \ref{sec:apdx_lemmas12_bridging} for details). It then follows that there is vanishing correlation between pairs of geodesics of length $l$ from $i$ to $j$ via nodes $u$ and $v\in [n]\setminus\{i,j,u\}$ (from Lemma \ref{lemma:vanishingbridge}), which simplifies the first factor in Eq. \eqref{eq:spd_recursion_0} into a product of individual probabilities:
\begin{equation}\label{eq:spd_recursion}
    \begin{split}        
        &\condprob{\lambda_{ij}> l}{\lambda_{ij}>l-1,\phi_i} \approx \\ &\prod_{u\in[n]\setminus\{i,j\}}\left[1-\condprob{\lambda_{iu}=l-1,\lambda_{uj}=1}{\lambda_{ij}>l-1,\phi_i}\right].
    \end{split}
\end{equation}
Here, and at various points in the article, we use ``$\approx$'' to denote a first-order asymptotic approximation in the sense that the relative difference between the quantity of interest---on the left-hand side (LHS)---and its approximation---on the right-hand side (RHS)---vanishes with increasing network size $n$. We emphasize that this approximation holds on the giant component for finite geodesic lengths as $n\to\infty$: as we consider longer geodesics of length that scales as $\bigtheta{\log n}$, the probability of a node to be at that distance approaches $\bigtheta{1}$ instead of $\bigtheta{n^{-1}}$, even for sparse networks. This induces finite-size effects on the GLD, which become incrementally concentrated around the mode of the GLD as network size increases; see Fig. \ref{fig:spd_er_d2_errs} and Appendix \ref{sec:apdx_finite_size} for an extended treatment of finite-size effects. To simplify the term on the RHS of Eq. \eqref{eq:spd_recursion}, we use Lemma \ref{lemma:1} in Appendix \ref{sec:apdx_lemmas12_bridging}:
\begin{equation}\label{eq:lemma1}
\begin{split}
\condprob{\lambda_{iu}=l-1, \lambda_{uj}=1}{\lambda_{ij}\ge l,\phi_i} \approx &\ \condprob{\lambda_{iu}=l-1}{\phi_i}\\ &\times\prob{A_{uj}=1}.    
\end{split}
\end{equation}
Finally, note that
\begin{equation*}
    \begin{split}
        \condprob{\lambda_{iu}=l-1}{\phi_i} =&\ \condprob{\lambda_{iu}=l-1}{\lambda_{iu}>l-2,\phi_i}\\&\times\condprob{\lambda_{iu}>l-2}{\phi_i},
    \end{split}
\end{equation*}
which, alongside Eqs. \eqref{eq:lemma1} and \eqref{eq:spd_recursion}, yields 
\begin{equation}
\label{eq:spd_recursion_2}
    \begin{split}
        &\condprob{\lambda_{ij}>l}{\lambda_{ij}>l-1,\phi_i} = 1-\condprob{\lambda_{ij}= l}{\lambda_{ij}>l-1,\phi_i}\\
        &\approx\prod_{u\in[n]\setminus\{i,j\}}[1-\condprob{\lambda_{iu}=l-1}{\lambda_{iu}>l-2,\phi_i}\\&\hspace{53pt}\times\condprob{\lambda_{iu}>l-2}{\phi_i}\prob{A_{uj}=1}].
    \end{split}
\end{equation}
Using the definitions in Eq. \eqref{eq:def_omega_psi} we can write Eqs. \eqref{eq:spd_recursion_0}, \eqref{eq:spd_recursion_2} succinctly in terms of the survival function matrix and conditional PMF matrix of the GLD:
\begin{subequations}
    \label{eq:spd_main}
    \begin{align}
        \label{eq:spd_main_psi}
        [\Psiaf_l]_{ij} &= (1 - [\Omegaaf_l]_{ij})[\Psiaf_{l-1}]_{ij}, \\
        \label{eq:spd_main_omega}
        [\Omegaaf_l]_{ij} &\approx 1 - \exp\left(\sum_u\log\left(1-[\Omegaaf_{l-1}]_{iu}[\Psiaf_{l-2}]_{iu}\avg{A_{uj}}\right)\right),
    \end{align}
\end{subequations}
where $\exp$ and $\log$ refer to element-wise exponentiation and natural logarithm, respectively. The above pair of recursive equations, together with the initial conditions
\begin{equation*}
    \begin{split}
        [\Omegaaf_1]_{ij}&\triangleq \condprob{\lambda_{ij}=1}{\lambda_{ij}>0,\phi_i}=\condprob{A_{ij}=1}{\phi_i},\\
        [\Psiaf_0]_{ij}&\triangleq \condprob{\lambda_{ij}>0}{\phi_i}=1,
    \end{split}
\end{equation*}
completely define the distribution of geodesic lengths, from which other network properties of interest can be extracted; see our follow-up work in Ref. \onlinecite{loomba2024geodesics2}. Note that $[\Omegaaf_0]_{ij}$ does not exist, while $1-  [\Psiaf_\infty]_{ij}$ naturally encodes the probability of a shortest path between $i,j$ being of infinite length, i.e. $j$ not being on the giant component, as $i$ is already on the giant component. Since loops are ignored, $[\Psiaf_l]_{ii}\triangleq 0$ and $[\Omegaaf_l]_{ii}\triangleq 0$.

To define the initial condition $[\Omegaaf_1]_{ij}=\condprob{A_{ij}=1}{\phi_i}$, we use Lemma \ref{lemma:deg_gcc} in Appendix \ref{sec:apdx_connect_gcc} which shows:
\begin{equation}
    \label{eq:prob_connect_exact}
    \condprob{A_{ij}=1}{\phi_i} \approx \prob{A_{ij}=1}\left\{1+\left[\frac{1}{\prob{\phi_i}}-1\right]\prob{\phi_j}\right\}.
\end{equation}
Estimating the RHS of Eq. \eqref{eq:prob_connect_exact} requires the probability of a node to be on the giant component, or ``percolating'', which we derive next in Sec. \ref{sec:perc_prob}.

\subsection{\label{sec:perc_prob}Percolation probability}

In this section we provide two solutions for the percolation probability $\prob{\phi_i}$ of node $i$ which respectively yield an ``analytic form'' and an ``approximate analytic form'' of the GLD.

\paragraph*{Analytic form} In principle, percolation probabilities can be extracted from the survival function of the GLD. By continuity of the distribution,
\begin{equation}\label{eq:perc_prob_limit}
    \condprob{\lambda_{ij}=\infty}{\phi_i} = \lim_{l\to\infty} \condprob{\lambda_{ij}>l}{\phi_i}.
\end{equation}
That is, the steady state of recursive Eq. \eqref{eq:spd_main_psi} is indicative of the amount of probability mass at $\lambda_{ij}=\infty$, when $i$ is on the giant component. Then $\forall (i, j)$ we have
\begin{equation}\label{eq:perc_prob_pinf}
    \prob{\phi_j}=1-\condprob{\lambda_{ij}=\infty}{\phi_i}.
\end{equation}
Eqs. \eqref{eq:spd_main}--\eqref{eq:perc_prob_pinf} provide us with the full ``analytic form'' of the GLD.

\paragraph*{Approximate analytic form}At first, computing the RHS of Eq. \eqref{eq:perc_prob_pinf} appears to be a circular problem, since we require the limiting value of the survival function of the GLD to obtain its initial condition. (We will shortly find a self-consistent equation to independently compute percolation probabilities; see Eq. \eqref{eq:gcc_consistency_sparse}.) However, we empirically observe that precision in the initial condition is important for agreement of the survival function for only smaller geodesic lengths, whereas we obtain the expected limiting value when using a na\"{i}ve approximation of the initial condition by setting
\begin{equation}
    \label{eq:prob_connect_apx}
    \condprob{A_{ij}=1}{\phi_i} \approxnaive \prob{A_{ij}=1},
\end{equation}
where we use ``$\approxnaive$'' to denote this na\"{i}ve approximation, to distinguish it from a first-order asymptotic approximation. Therefore, running the recursive setup once gives access to $\condprob{\lambda_{ij}=\infty}{\phi_i}$, which can be used to derive the (exact) analytic form of the GLD from Eqs. \eqref{eq:spd_main}, \eqref{eq:prob_connect_exact} and \eqref{eq:perc_prob_pinf} in a second recursion. Henceforth, we refer to the GLD obtained by using the na\"{i}ve initial condition in Eq. \eqref{eq:prob_connect_apx}, alongside Eq. \eqref{eq:spd_main}, as the ``approximate analytic form'' of the GLD. We note from Eq. \eqref{eq:prob_connect_exact} that the approximation in Eq. \eqref{eq:prob_connect_apx} holds, in particular, when the network is almost surely connected: $\prob{\phi_i}\to 1$. More generally, the approximation will underestimate the probability mass on short geodesics; see Appendix \ref{sec:apdx_connect_gcc}. For an ER graph of mean degree $\avgdegree$, where $\avg{A_{ij}}=\frac{\avgdegree}{n}$, Fig. \ref{fig:spd_er_d2} indicates good agreement between the empirical and analytic cumulative distribution function (CDF) of the GLD. We note that the approximate analytic CDF remains a good approximation even for shorter geodesics.

\begin{figure}
    \centering
    \includegraphics[width=\columnwidth]{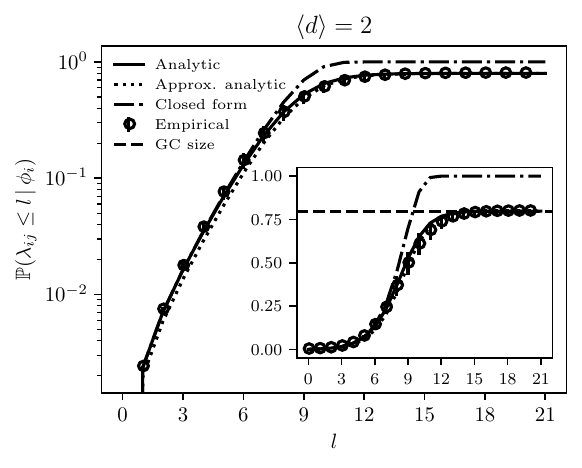}
    \caption{Analytic CDF of geodesic lengths for an ER graph agree with the empirical CDF, where the source node is on the giant component. Network size is fixed at $n=1024$ and mean degree at $\avgdegree=2$. Solid and dotted lines indicate analytic solutions derived from analytic (Eqs. \eqref{eq:spd_main}, \eqref{eq:prob_connect_exact}, \eqref{eq:gcc_consistency}) and approximate analytic forms (Eqs. \eqref{eq:spd_main}, \eqref{eq:prob_connect_apx}) respectively, while dash-dotted line indicates the closed form obtained from Eq. \eqref{eq:sf_avg}. Symbols ($\circ$) and bars indicate empirical estimates: mean and standard error over 10 network samples. Dashed asymptote indicates size of the giant component as estimated from the self-consistent Eq. \eqref{eq:gcc_consistency}. The approximate analytic form marginally underestimates the probability mass for shorter lengths, as is evident on the logarithmic scale (main plot). The closed form shows good agreement for shorter lengths, but deviates strongly for longer ones (inset plot on the linear scale)---saturating to unity for any percolating network. There is good agreement between the analytic and empirical estimates, with some deviation around the mode of the distribution due to finite-size effects---see Appendix \ref{sec:apdx_finite_size}.}
    \label{fig:spd_er_d2}
\end{figure}

\paragraph*{Self-consistent equation}Alternatively, we can derive percolation probabilities without invoking the GLD. Corollary \ref{lemma:perc} in Appendix \ref{sec:percprob} shows that, in a SEAN, the percolation probabilities are asymptotically given by the following self-consistent equation:
\begin{equation}\label{eq:gcc_consistency_sparse}
    \prob{\phi_i}\approx 1-\exp\left(-\sum_{j\ne i}\prob{A_{ij}=1}\prob{\phi_j}\right).
\end{equation}
We remark that, due to sparsity, Eq. \eqref{eq:gcc_consistency_sparse} yields $\prob{\phi_i}=\bigtheta{1}$ or $0$, i.e. the percolation probability of a node in a sparse network is asymptotically independent of network size $n$. Let $\boldsymbol{\rho}$ be the vector encoding $\rho_i\triangleq \prob{\phi_{i}}$. Then the percolation probability of every node is given by the non-trivial solution to the following transcendental vector equation:
\begin{equation}\label{eq:gcc_consistency}
    \boldsymbol\rho = \ones{n}-\exp\left(-\avg{\mat{A}}\boldsymbol\rho\right),
\end{equation}
where $\ones{n}$ is an all-ones vector of length $n$. For ER graphs with mean degree $\avgdegree$ this becomes a scalar equation: $\rho=1-\exp(-\avgdegree\rho)$. In Fig. \ref{fig:gcc_consistency_spld} in Appendix \ref{sec:apdx_finite_size}, we demonstrate that there is a strong agreement between the percolation probabilities estimated empirically, obtained from the approximate analytic form (using Eqs. \eqref{eq:spd_main}, \eqref{eq:perc_prob_limit}--\eqref{eq:prob_connect_apx}), and obtained from the self-consistent Eq. \eqref{eq:gcc_consistency} via function iteration. In Appendix \ref{sec:apdx_asymmetric} we show that, for directed networks, there are circumstances when using the approximate analytic form of the GLD is more useful. In all results that follow, we use Eqs. \eqref{eq:spd_main}, \eqref{eq:prob_connect_exact}, and \eqref{eq:gcc_consistency} to yield the ``analytic form'' of the GLD in the supercritical regime.

\subsection{\label{sec:spd_subcritical}Subcritical regime}

We now extend the GLD formalism to the subcritical regime. A node pair $(i,j)$ is said to be subcritical if asymptotically there cannot exist a giant component such that there exists a path from $i$ to $j$ that goes through it. This implies that, asymptotically, either there cannot exist \emph{any} path between $i,j$---in which case $\avg{\mat{A}}$ is reducible and the GLD trivially has all probability mass at infinity---or $i,j$ can only exist on a small component yielding the subcriticality condition
\begin{equation*}
    \prob{\phi_i}=\prob{\phi_j}=0,
\end{equation*}
in which case we consider the GLD on the small component containing $i,j$.

Analogous to Eq. \eqref{eq:def_omega_psi}, we can define the conditional PMF and survival function matrices without the conditioning on $\phi_i$ (which is an impossible event). Also, the recursive setup described in Sec. \ref{sec:spd_supercritical}, alongside the key result in Lemma \ref{lemma:1}, can be applied verbatim in the subcritical regime. (In fact, the asymptotic approximations in Eqs. \eqref{eq:spd_recursion} and \eqref{eq:lemma1} hold on the small component for \emph{all} geodesic lengths since the existence of geodesics of longer lengths becomes rarer.) The only difference is in the initial condition for the conditional PMF:
\begin{equation}
    \label{eq:init_omega_subcritical}
    [\Omegaaf_1]_{ij}\triangleq \condprob{\lambda_{ij}=1}{\lambda_{ij}>0}=\prob{A_{ij}=1},
\end{equation}
which is completely defined under the ensemble average model. That is, Eqs. \eqref{eq:spd_main} and \eqref{eq:init_omega_subcritical} yield the analytic GLD when $i,j$ are in the subcritical regime. In the sections that follow we focus on the supercritical regime, but remark that results for the subcritical regime naturally follow by dropping the conditioning on $\phi_i$. In Appendix \ref{sec:apdx_finite_size}, we show the empirics and analytics agree for subcritical ER graphs of varying mean degree (Fig. \ref{fig:spd_er_smallcomponent}), and for a subcritical bipartite SBM (Fig. \ref{fig:bipartite_cdf_smallcomponent}).

\subsection{\label{sec:closedform_bound}Closed form of the GLD}

While the recursive formulation in Eq. \eqref{eq:spd_main} is powerful, additional approximations allow for further analytical progress. Since the conditional PMF matrix $\Omegaaf_{l-1}$, survival function matrix $\Psiaf_{l-2}$, and expected adjacency matrix $\avg{\mat{A}}$ encode probability measures, their entries will always be non-negative and no larger than unity, allowing us to apply a generalized version of Bernoulli's inequality to the RHS of Eq. \eqref{eq:spd_main_omega} to obtain the bound:
\begin{equation}
\label{eq:omegaaf_bernoulli_bound}
\Omegaaf_l \le (\Omegaaf_{l-1}\odot\Psiaf_{l-2})\avg{\mat{A}} \le \Omegaaf_{l-1}\avg{\mat{A}} \le \Omegaaf_{1}\avg{\mat{A}}^{l-1},
\end{equation}
where $\le$ and $\odot$ respectively indicate element-wise comparison and product. We can leverage network sparsity in Eq. \eqref{eq:sparsity_constraint} to show that the final bound in Eq. \eqref{eq:omegaaf_bernoulli_bound} is asymptotically tight. As $n\to\infty$, Eq. \eqref{eq:sparsity_constraint} allows for a first-order approximation for the logarithm in Eq. \eqref{eq:spd_main_omega}:
\begin{equation}
    \label{eq:spd_analytic_omega_1}
    \Omegaaf_l \approx \ones{n}\ones{n}^T-\exp\left(-(\Omegaaf_{l-1}\odot\Psiaf_{l-2})\avg{\mat{A}}\right).
\end{equation}
Let $l=2$, for which $[\Psiaf_{l-2}]_{iu}\triangleq \condprob{\lambda_{iu}>l-2}{\phi_i}= \condprob{\lambda_{iu}>0}{\phi_i}=1$. Then we can write from Eq. \eqref{eq:spd_analytic_omega_1}:
\begin{equation}
    \label{eq:spd_analytic_omega_2}
    \Omegaaf_2 \approx \ones{n}\ones{n}^T-\exp\left(-\Omegaaf_1\avg{\mat{A}}\right).
\end{equation}
Due to sparsity in Eq. \eqref{eq:sparsity_constraint}, and Eq. \eqref{eq:prob_connect_exact}, we obtain:
\begin{equation}
    \label{eq:spd_analytic_omega_3}
    \Omegaaf_1=\bigtheta{n^{-1}}\textrm{ or }0\implies\Omegaaf_1\avg{\mat{A}}=\bigtheta{n^{-1}}\textrm{ or }0.
\end{equation}
Application of a first-order approximation for the exponential in Eq. \eqref{eq:spd_analytic_omega_2} yields:
\begin{equation}
    \label{eq:spd_analytic_omega_4}
    \Omegaaf_2 \approx \Omegaaf_1\avg{\mat{A}}.
\end{equation}
From Eqs. \eqref{eq:spd_main_psi} and \eqref{eq:spd_analytic_omega_2} we obtain:
\begin{equation}
    \label{eq:spd_analytic_omega_5}
    \Psiaf_1 \approx\exp\left(-\Omegaaf_1\avg{\mat{A}}\right).
\end{equation}
Next, consider Eq. \eqref{eq:spd_analytic_omega_1} with $l=3$, for which using Eqs. \eqref{eq:spd_analytic_omega_3} and \eqref{eq:spd_analytic_omega_5} we can assume $\Psiaf_{l-2}=\Psiaf_1\approx \ones{n}\ones{n}^T$, where $\ones{k}$ is the all-ones vector of size $n$. This yields an equation for $l=3$ similar to Eq. \eqref{eq:spd_analytic_omega_2}: 
\begin{equation*}
    \Omegaaf_3 \approx \ones{n}\ones{n}^T-\exp\left(-\Omegaaf_2\avg{\mat{A}}\right)\approx\ones{n}\ones{n}^T-\exp\left(-\Omegaaf_1\avg{\mat{A}}^2\right),
\end{equation*}
where we have used Eq. \eqref{eq:spd_analytic_omega_4}. Due to sparsity, we can apply identical arguments as above to obtain: 
\begin{equation*}
\begin{split}
    &\Omegaaf_1\avg{\mat{A}}^2=\bigtheta{n^{-1}}\textrm{ or }0
    \implies\Omegaaf_3\approx\Omegaaf_1\avg{\mat{A}}^2\\
    &\implies\Psiaf_2\approx\exp\left(-\Omegaaf_1\left(\avg{\mat{A}}+\avg{\mat{A}}^2 \right)\right).    
\end{split}
\end{equation*}
In the infinite-size limit, by induction for any finite $l$, we effectively propagate the sparsity assumption through
\begin{equation}
    \label{eq:spd_analytic_omega_8}
    \Omegaaf_{l}=\bigtheta{n^{-1}}\textrm{ or }0,
\end{equation}
that results in
\begin{subequations}
\label{eq:spd_analytic}
    \begin{align}
        \label{eq:spd_analytic_omega}
        \Omegaaf_l &\approx \Omegaaf_{l-1}\avg{\mat{A}} \implies \Omegaaf_l \approx \Omegaaf_{1}\avg{\mat{A}}^{l-1}, \\
        \label{eq:spd_analytic_psi}
        \Psiaf_l &\approx \exp\left(-\sum_{k=1}^l\Omegaaf_k\right).
    \end{align}
\end{subequations}
We emphasize that this induction relies on the application of the sparsity assumption: asymptotically, the survival function $[\Psiaf_{l-2}]_{iu}\triangleq \condprob{\lambda_{iu}>l-2}{\phi_i}$ is unity or, equivalently, the conditional PMF $\condprob{\lambda_{ij}=l}{\lambda_{ij}>l-1,\phi_i}$ is the PMF $\condprob{\lambda_{ij}=l}{\phi_i}$. In the subcritical regime (and therefore ignoring the conditioning on $\phi_i$), this holds for any value of $l$ in a network of finite (but large) size $n$, since $j$ is almost surely on a different component from $i$'s, i.e. $\prob{\lambda_{ij}>l-1}\approx 1$. In the supercritical regime, this holds for any finite value of $l$ in the infinite-size limit, since a shortest path between any two arbitrary nodes is almost surely no shorter than $l$, and thus the almost-sure event $\lambda_{ij}>l-1$ does not inform the almost-never event $\lambda_{ij}=l$. (We will rigorously state in Ref. \onlinecite{loomba2024geodesics2} that there is a direct connection between the existence of a giant component in the infinite-size limit and the support over which this approximation holds.) For finite-sized networks in the supercritical regime, this approximation will evidently work only for smaller geodesic lengths. More precisely, for geodesic lengths that scale as $\bigtheta{\log n}$---which is around the mode of the GLD---correlations due to finite network size become apparent, Eq. \eqref{eq:spd_analytic_omega_8} no longer holds, and Eq. \eqref{eq:spd_analytic_omega} stops being a tight approximation on the actual conditional PMF. (See Appendix \ref{sec:apdx_finite_size} for details.) As shown by the bound in Eq. \eqref{eq:omegaaf_bernoulli_bound}, the ``conditional PMF'' encoded by the RHS in Eq. \eqref{eq:spd_analytic_omega} will always overestimate the analytic conditional PMF given by Eq. \eqref{eq:spd_main_omega} and in practice it may not be a valid probability measure. Nevertheless, it yields an expression for the survival function matrix in Eq. \eqref{eq:spd_analytic_psi} in terms of the sum of powers of $\avg{\mat{A}}$:
\begin{equation}
    \label{eq:sf_avg}
    \Psicf_l \triangleq \exp\left(-\Omegaaf_1\sum_{k=1}^l\avg{\mat{A}}^{k-1}\right),
\end{equation}
which, as such, \emph{is} a valid probability measure, and alongside Eqs. \eqref{eq:prob_connect_exact}, \eqref{eq:gcc_consistency} completely describes the GLD. We refer to Eqs. \eqref{eq:prob_connect_exact}, \eqref{eq:gcc_consistency} and \eqref{eq:sf_avg} as the ``closed form'' of the GLD, and emphasize that these approximations underestimate probability mass on longer geodesics: the closed form of the conditional PMF of the GLD is an upper bound on its analytic form, which is tight for shorter lengths in finite-size networks. These approximations will not be useful if we are interested in computing the size of the giant component, since the closed form in Eq. \eqref{eq:sf_avg} will always have a limiting value of unity in the supercritical regime (see Theorems 1 and 4 in Ref. \onlinecite{loomba2024geodesics2}). However, if we are interested in geodesics shorter than the modal length, or the expectation of geodesic lengths rather than the full GLD, then these are seen to be reasonable approximations. We demonstrate this behavior for ER graphs in Fig. \ref{fig:spd_er}, and for other statistical network models in Sec. \ref{sec:geodesics_specific}.

\paragraph*{Approximate closed form of the GLD}Finally, we consider a scenario which produces an insightful interpretation for the survival function of the GLD. The na\"{i}ve approximation to the initial condition from Eq. \eqref{eq:prob_connect_apx}---which given Eq. \eqref{eq:prob_connect_exact} is asymptotically tight if $\forall i: \prob{\phi_i}\approx 1$ i.e. the network is almost surely connected---can be inserted in Eq. \eqref{eq:sf_avg} to yield an ``approximate closed form'' of the survival function of the GLD:
\begin{equation}
    \label{eq:sf_avg_uncorrected}
    \Psiacf_l \triangleq \exp\left(-\sum_{k=1}^l\avg{\mat{A}}^{k}\right).
\end{equation}
Eq. \eqref{eq:sf_avg_uncorrected} shows that the approximate closed form of the survival function of the GLD at length $l$ is encoded by the sum of powers of $\avg{\mat{A}}$ from $1$ to $l$. This is reminiscent of the well-known result that the sum of powers of $\mat{A}$ encodes the number of walks of up to length $l$ \cite{newman2018networks}. We emphasize that if all node pairs are in the subcritical regime, then Eq. \eqref{eq:init_omega_subcritical} yields exactly $\Omegaaf_1=\avg{\mat{A}}$, i.e. Eq. \eqref{eq:sf_avg_uncorrected} is exactly the closed form of the survival function of the GLD, evident in Figs. \ref{fig:spd_er_smallcomponent} and \ref{fig:bipartite_cdf_smallcomponent} in Appendix \ref{sec:apdx_finite_size}.

\paragraph*{Interpretation}For an alternate comprehension of the expression in Eq. \eqref{eq:sf_avg_uncorrected}, apply a first-order approximation to its RHS---applicable for SEANs in the infinite-size limit---to obtain for node pair $(i,j)$: $\idx{\Psiacf_l}{ij}\approx 1-\sum_{k=1}^l\idx{\avg{\mat{A}}^k}{ij}$. The approximate closed form of the survival function \emph{$\idx{\Psiacf_l}{ij}$ can itself be approximated as the probability of $i$, which is on the giant component, failing to connect to $j$ (in the asymptotic limit) via any of the asymptotically independent walks of lengths $1$ through $l$.} To see how, note that $\avg{\mat{A}}=\bigtheta{n^{-1}}\textrm{ or }0\implies\avg{\mat{A}}^k=\bigtheta{n^{-1}}\textrm{ or }0$. If all walks of length $k$ between $i,j$ are independent of one another---an independent-walk condition---then $\idx{\avg{\mat{A}}^k}{ij}$ encodes the probability of a walk of length $k$ between them, since higher-order terms cancel out due to the sparsity of $\avg{\mat{A}}^k$ noted above. If walks of length $1,2,\dots, l$ are independent of one another---an independent-walk-length condition---then $\sum_{k=1}^l\idx{\avg{\mat{A}}^k}{ij}$ encodes the probability of a walk existing between $i,j$ of any length up to $l$ i.e. $\condprob{\lambda_{ij}\le l}{\phi_i}$, since  higher-order terms cancel out again due to sparsity. Then $1-\sum_{k=1}^l\idx{\avg{\mat{A}}^k}{ij}$ is $\condprob{\lambda_{ij}>l}{\phi_i}$, which is exactly the LHS of Eq. \eqref{eq:sf_avg_uncorrected}. (Intuitively, sparsity implies that, asymptotically, walks of finite length are unlikely to share edges, which in turn implies the asymptotic conditions of independent walks and independent walk lengths---which, as previously discussed, are more tenable for shorter geodesic lengths $l$ in finite-size networks.)  This interpretation provides a retrospective derivation for the approximate closed form of the survival function of the GLD.

\begin{figure}
    \centering
    \includegraphics[width=\columnwidth]{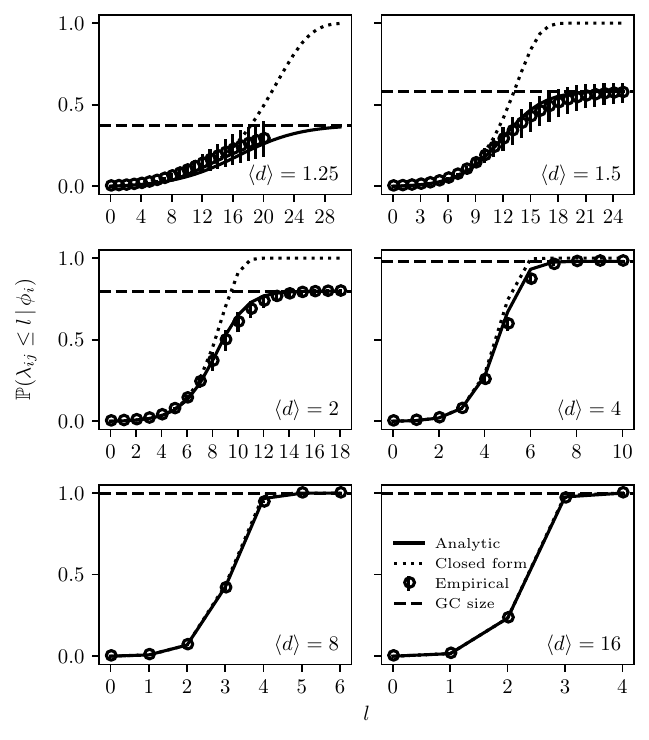}
    \caption{Empirical, analytic and closed-form CDF of geodesic lengths where the source node is on the giant component, for an ER graph with varying connectivity. Network size is fixed at $n=1024$, while mean degree varies as $\avgdegree\in\{1.25, 1.5, 2, 4, 8, 16\}$. Solid and dotted lines indicate analytic (Eqs. \eqref{eq:spd_main}, \eqref{eq:prob_connect_exact}, \eqref{eq:gcc_consistency}) and closed-form solutions (Eqs. \eqref{eq:prob_connect_exact}, \eqref{eq:gcc_consistency}, \eqref{eq:sf_avg}), respectively. Symbols ($\circ$) and bars indicate empirical estimates: mean and standard error over 10 network samples. Dashed asymptote indicates size of the giant component as estimated from the self-consistent Eq. \eqref{eq:gcc_consistency}. The analytic GLD is in good agreement for all connectivities at all lengths, while the closed-form GLD is in good agreement for all connectivities at shorter lengths (with deviation beginning around the mode).}
    \label{fig:spd_er}
\end{figure}

\section{\label{sec:general_graphs}Sparse general random network}

Since $\avg{\mat{A}}$ can encode any underlying statistical network model, we can consider equations for the GLD when we may not have access to the full ensemble average $\avg{\mat{A}}$, but knowledge of a (possibly inferred) statistical network model that treats both edges and node identities as random variables. Many widely used network models, from SBMs to RGGs, satisfy this condition. We consider a network to be a sparse general random network (SGRN; a special case of SEAN---see later) if it is generated from a network model for which the following assumptions hold.
\begin{enumerate}
    \item \emph{Independent node locations:} Nodes are independently located at some coordinate in a topological node space $V$ based on probability measure $\mu$---whose definition is independent of $n$
    \begin{equation}
        \label{eq:gnodelocation}
        X_i\sim\mu.
    \end{equation}
    \item \emph{Conditionally independent edges:} Conditioned on the identity of a node pair, the existence of an edge between them is independent of the edge between another node pair, and only depends on their respective locations in $V$ through a connectivity kernel $\nu:V\times V\to[0,1)$:
    \begin{equation}
        \label{eq:bernoulli_model_general}
        A_{ij}\sim\bernoulli{\nu(X_i,X_j)}.
    \end{equation}
    \item \label{eq:gsparsityassume} \emph{Sparsity:} If a node pair has a non-zero connection probability, then it is appropriately small:
    \begin{equation}
        \label{eq:sparsity_constraint_general}
        \nu(x,y)=\bigtheta{n^{-1}}\textrm{ or }0 \textrm{ almost everywhere}.
    \end{equation}
    \item \emph{Continuity:} The connectivity kernel $\nu$ is a continuous function from $V\times V$ to $[0,1)$.
    \item \emph{Asymptotic limit:} The network size is large: $n\gg 1$.
\end{enumerate}
We call the corresponding network model an SGRN model. The node space $V$ can take various forms---such as a discrete space for SBMs \cite{holland1983sbm}, a Euclidean space for RGGs \cite{penrose2003rgg}, or an inner product space for RDPGs \cite{young2007rdpg}, that we consider in detail in Sec. \ref{sec:geodesics_specific}. We enforce $A_{ii}=0$, i.e. no self-loops. For undirected networks, we additionally enforce that $\nu$ is symmetric, i.e. $\nu(x,y)=\nu(y,x)$, and, without loss of generality, use node indices as an arbitrary ordering: for $i<j$, $A_{ij}$ is generated independently from Eq. \eqref{eq:bernoulli_model_general}, and set $A_{ji}= A_{ij}$. We emphasize that while both directed and undirected networks can be generated using symmetric connectivity kernels, networks generated from asymmetric connectivity kernels must necessarily be directed; by permitting asymmetry, this model is an extension of the sparse inhomogeneous random graph model \cite{bollobas2007phase, bollobas2011sparsegraphs}, which we specifically consider in Sec. \ref{sec:graphons}.

We emphasize that SGRN is a special case of SEAN, since the assumptions of independent node locations and continuity of the connectivity kernel imply the assumption of no bottlenecks in Eq. \eqref{eq:nonultrabottlenecked}; see Lemma \ref{lemma:generalultrabottleneck} in Appendix \ref{sec:percprob}. Therefore, all asymptotic results derived in Sec. \ref{sec:spd} hold by replacing each sum over $n$ nodes by $n$ times the expectation over the node space $V$. Since nodes are completely identified by their locations in $V$, in this section we consider analogous properties of interest for given node location: the distribution of geodesic lengths between nodes at a given pair of locations $(x,y)\in V\times V$, and the percolation probability of a node at a given location $x\in V$.

\paragraph*{Analytic form of the GLD} Analogous to the survival function matrix $\Psiaf_l$ and conditional PMF matrix $\Omegaaf_l$ of the GLD in SEANs, defined in Eq. \eqref{eq:def_omega_psi}, we define the survival function kernel $\psiaf_l(\cdot,\cdot)$ and conditional PMF kernel $\omegaaf_l(\cdot,\cdot)$ of the GLD in SGRNs:
\begin{equation*}
\begin{split}
    \psiaf_l(x,y)&\triangleq \condprob{\lambda_{ij}>l}{\phi_i,X_i=x,X_j=y},\\
    \omegaaf_l(x,y)&\triangleq \condprob{\lambda_{ij}=l}{\lambda_{ij}>l-1,\phi_i,X_i=x,X_j=y},
\end{split}
\end{equation*}
where the source (target) node is at $x\in V$ ($y\in V$). Then, asymptotically, the sum over $n$ nodes in Eq. \eqref{eq:spd_main} becomes $n$ times the integral over node space $V$:
\begin{subequations}
    \label{eq:spd_main_general}
    \begin{align}\label{eq:spd_main_general_psi}
        \psiaf_l(x,y) &=[1 - \omegaaf_l(x,y)]\psiaf_{l-1}(x,y), \\
        \omegaaf_l(x,y) &\approx 1 - \exp\bigg(n\int_V\log(1-(\omegaaf_{l-1}\cdot\psiaf_{l-2})(x,z)\nonumber\\
        \label{eq:spd_main_general_omega}
        &\hspace{100pt}\times\nu(z,y))d\mu(z)\bigg).
    \end{align}
\end{subequations}
Eq. \eqref{eq:spd_main_general}, together with the initial conditions 
\begin{equation*}
    \begin{split}
        \omegaaf_1(x,y)&\triangleq \condprob{\lambda_{ij}=1}{\lambda_{ij}>0,\phi_i,X_i=x,X_j=y}\\&=\condprob{A_{ij}=1}{\phi_i,X_i=x,X_j=y},\\
        \psiaf_0(x,y)&\triangleq \condprob{\lambda_{ij}>0}{\phi_i,X_i=x,X_j=y}=1,
    \end{split}
\end{equation*}
completely define the distribution of geodesic lengths. We remark that $\psiaf_l(x,x)$ and $\omegaaf_l(x,x)$ encode the distribution of geodesic lengths between nodes with identical locations in $V$.

\paragraph*{Percolation probability} To define the initial condition $\omegaaf_1(\cdot,\cdot)$ we require percolation probabilities in $V$. As in Sec. \ref{sec:perc_prob}, let $$\rho(x)\triangleq\condprob{\phi_i}{X_i=x}$$ be the probability that a node at $x$ is on the giant component. Following the same argument as for the SEAN model, and replacing the sum over nodes by $n$ times the integral over node space, we can asymptotically write $\rho(\cdot)$ as the solution to a self-consistent integral equation analogous to Eq. \eqref{eq:gcc_consistency}:
\begin{equation}\label{eq:gcc_consistency_general}
    \rho(x) \approx 1-\exp\left(-n\int_V\nu(x,y)\rho(y)\,d\mu(y)\right).
\end{equation}
This leads to the base case, analogous to Eq. \eqref{eq:prob_connect_exact}:
\begin{equation}\label{eq:init_omega_general}
\omegaaf_1(x,y)\approx\nu(x,y)\left\{1+\left[\frac{1}{\rho(x)}-1\right]\rho(y)\right\},
\end{equation}
needed to solve Eq. \eqref{eq:spd_main_general}.

\paragraph*{Closed form of the GLD} Applying asymptotic approximations identical to those made for the SEAN model, we obtain a set of equations analogous to Eq. \eqref{eq:spd_analytic} for the closed form of the GLD:
\begin{subequations}
    \label{eq:spd_general}
    \begin{align}
    \label{eq:spd_general_omega}
    \omegacf_l(x,y) &\triangleq n\int_V\omegacf_{l-1}(x,z)\nu(z,y)\,d\mu(z),\\
    \label{eq:spd_general_psi}
    \psicf_l(x,y) &\triangleq \exp\left(-\sum_{k=1}^l\omegacf_k(x,y)\right),
    \end{align}
\end{subequations}
To our knowledge, this is the best first-order approximation that allows one to analytically access the GLD for sparse general random network families with conditionally independent edges, in both subcritical and supercritical regimes. Substituting $\omegaaf_1(x,y)$ in Eq. \eqref{eq:spd_general_omega} yields:
\begin{equation}
    \label{eq:spd_analytic_general_omega}
    \begin{split}        
    \omegacf_l(x,y) = n^{l-1}\int_V\int_V&\cdots\int_V\omegaaf_1(x,z_1)\nu(z_1,z_2)\dots\\
    \times&\nu(z_{l-1},y)\,d\mu(z_1)\,d\mu(z_2)\dots d\mu(z_{l-1}).
    \end{split}
\end{equation}
From Eqs. \eqref{eq:spd_general_psi} and \eqref{eq:spd_analytic_general_omega}, the closed form of the survival function of the GLD is encoded by the sum of iterated integrals.

\paragraph*{Integral operator} Viewing the expected adjacency matrix $\avg{\mat{A}}$ of the SEAN model as a linear operator acting on an $n$-dimensional vector space, it is convenient to define an analogous integral operator $T$ for the SGRN model acting on the space of functions in $V$, whose kernel is given by $n$ times the connectivity kernel:
\begin{equation}
    \label{eq:integral_op}
    (Tf)(x)\triangleq n\int_V\nu(x,y)f(y)\,d\mu(y).
\end{equation}
Due to sparsity assumed in Eq. \eqref{eq:sparsity_constraint_general}, $T$ is a compact operator; see Appendix \ref{sec:apdx_general_sparsity}. $T$ maps a function evaluated at node location $x$ to an expectation of the function evaluated at other node locations $y$, weighted by the node density at $y$ and the probability of a node at $y$ to connect to a node at $x$. For example, if $\forall x\in V: f(x)=1$, then $(Tf)(\cdot)$ is the expected out-degree function; see Eq. \eqref{eq:general_degrees_out}. When applied to the percolation probability function $\rho(x)$, the self-consistent Eq. \eqref{eq:gcc_consistency_general} can be written as $\rho(x)=1-\exp\left(-(T\rho)(x)\right)$. Therefore, many network quantities of interest can be extracted using $T$. For the rest of this section, we assume that the connectivity kernel $\nu$ is symmetric, implying that $T$ is self-adjoint, and refer the reader to Appendix \ref{sec:apdx_asymmetric} for a discussion on asymmetric connectivity kernels. This allows us to apply the spectral theorem for compact self-adjoint operators \cite{riesz1955hilbert}: Let $T$ have rank $N$, then there exists an orthonormal system of eigenfunctions $\{\varphi_i\}_{i=1}^N$ of $T$, corresponding to ordered non-zero eigenvalues $\{\tau_i\}_{i=1}^N$, such that $\{|\tau_i|\}_{i=1}^N$ is monotonically non-increasing with index $i$ and---due to compactness of $T$---either $N$ is finite, or $\lim_{i\to\infty}\tau_i=0$. Also, the following eigenfunction expansions hold:
\begin{subequations}
\begin{align}
    \label{eq:transform_rep}
    (Tf)(x) &= \sum_{i=1}^N\tau_i\left(\int_V f(y)\varphi_i(y)\,d\mu(y)\right)\varphi_i(x),\\
    \label{eq:kernel_rep}
    \nu(x,y) &= \frac{1}{n}\sum_{i=1}^N\tau_i\varphi_i(x)\varphi_i(y).
\end{align}
\end{subequations}

\paragraph*{Eigenvalues and homophily}Given Eq. \eqref{eq:kernel_rep}, we note that if an eigenvalue $\tau_i$ is positive (negative), then it raises the connection probability for node locations having the same (opposite) sign of the eigenfunction $\varphi_i(\cdot)$. Since nodes with the same (opposite) sign of an eigenfunction can be seen as being similar (dissimilar) along that ``dimension'', positive eigenvalues indicate \emph{homophily}---the phenomenon of similar nodes being more likely to connect to one another, that is widely observed in social networks \cite{mcpherson1987homophily, mcpherson2001birds}---whereas negative eigenvalues indicate \emph{heterophily}---dissimilar nodes being more likely to connect to one another, such as in multipartite graphs.

\paragraph*{Spectral closed form of the GLD}As in Sec. \ref{sec:closedform_bound}, substituting Eq. \eqref{eq:kernel_rep} in Eq. \eqref{eq:spd_analytic_general_omega}, we can integrate all intermediate variables from $z_{l-1}$ to $z_{2}$ by exploiting the orthonormality of $\{\varphi_i\}_{i=1}^N$:
\begin{equation*}
    \int_V\varphi_i(x)\varphi_j(x)\,d\mu(x)=\delta_{ij},
\end{equation*}
where $\delta_{ij}$ is the Kronecker delta which results in
\begin{equation}
    \label{eq:spd_analytic_general_eig_omega}
    \omegacf_l(x,y) = \int_V\sum_{i=1}^N\tau_i^{l-1}\omegaaf_1(x,z)\varphi_i(z)\varphi_i(y)\,d\mu(z).
\end{equation}
Inserting Eq. \eqref{eq:spd_analytic_general_eig_omega} in Eq. \eqref{eq:spd_general_psi}, we can push the outermost sum over length $k$ through due to compactness of $T$. Let $\widetilde{S}_l(a)\triangleq{1+a+\dots+ a^{l-1}}$ be the geometric sum of $a$ starting at $1$ up to $l$ terms---where $\widetilde{S}_0(a)\triangleq 0$---then the closed form of the survival function is given by:
\begin{equation}
\label{eq:spd_analytic_general_eig}
    \psicf_l(x,y) = \prod_{i=1}^N\exp\left(-\int_V\widetilde{S}_l(\tau_i)\omegaaf_1(x,z)\varphi_i(z)\varphi_i(y)\right).
\end{equation}

\paragraph*{Approximate closed form of the GLD}Similarly to Eq. \eqref{eq:sf_avg_uncorrected}, we can solve Eqs. \eqref{eq:spd_general_omega} and \eqref{eq:spd_general_psi}, with a na\"{i}ve initial condition analogous to Eq. \eqref{eq:prob_connect_apx}:
\begin{equation}
    \label{eq:prob_connect_apx_general}
    \omegaaf_1(x,y)\approxnaive\nu(x,y).
\end{equation}
Let $S_l(a)\triangleq{a+a^2+\dots+ a^{l}}$ be the geometric sum of $a$ starting at $a$ up to $l$ terms---where $S_0(a)\triangleq 0$. Then from Eq. \eqref{eq:spd_analytic_general_eig} we obtain the approximate closed form of the survival function, analogous to Eq. \eqref{eq:sf_avg_uncorrected}:
\begin{equation}
\label{eq:spd_analytic_general_eig_uncorrected}
    \psiacf_l(x,y) \triangleq \prod_{i=1}^N\exp\left(-\frac{1}{n}S_l(\tau_i)\varphi_i(x)\varphi_i(y)\right).
\end{equation}

\paragraph*{Reformulation as an extreme value distribution} If $\forall i\in[N]:\tau_i\ne 1$, we can write $S_l(\tau_i)=\tau_i\frac{\tau_i^l-1}{\tau_i-1}$. Define $a_i(x,y)\triangleq\frac{\tau_i\varphi_i(x)\varphi_i(y)}{n(\tau_i-1)}$, $b_i\triangleq\log|\tau_i|$, and $\sgn(\cdot)$ as the sign function. Then we can rewrite Eq. \eqref{eq:spd_analytic_general_eig_uncorrected} as:
\begin{equation}
    \label{eq:sf_gompertz}
    \psiacf_l(x,y)=\prod_{i=1}^N\exp\left(-a_i(x,y)\left[\sgn(\tau_i)^le^{b_il}-1\right]\right).
\end{equation}
Furthermore, if $\forall i\in[N]:\tau_i> 1$ and the nodes at $(x,y)$ are homophilous i.e. $\varphi_i(x)\varphi_i(y)>0$, then Eq. \eqref{eq:sf_gompertz} can be viewed as a product of survival functions of a (discrete version of the) Gompertz distribution, wherein the $i^\mathrm{th}$ term has shape and scale parameters $(a_i, b_i)$ that depend only on the eigenpair $(\tau_i,\varphi_i)$. Notably, the scale parameter $b_i$ of geodesic lengths is given by $\log\tau_i$. The Gompertz distribution---one of the three extreme value distributions with an exponential tail \cite{gnedenko1943distribution}---has been previously shown to model lengths of self-avoiding walks in ER graphs \cite{tishby2016sarw}.

\paragraph*{GLD for a random node pair}We can further define a distribution for the geodesic length between a randomly selected node pair in the network, which should be informative about geodesic lengths for typical node pairs, by averaging over the source and target node locations:
\begin{equation}
    \label{eq:spd_general_psi_agg}
    \psiacf(l) \triangleq \expectwrt{\mu^2}{\psiacf_l(x,y)}=\int_{V}\int_{V}\psiacf_l(x,y)\,d\mu(x)\,d\mu(y).
\end{equation}
We use the notation $\expectwrt{\mu^2}{\cdot}$ when averaging over $V\times V$. By applying Jensen's inequality \cite{jensen1906fonctions} to Eq. \eqref{eq:spd_general_psi_agg}---in the form $\exp(\expect{\cdot})\le\expect{\exp(\cdot)}$---and using Eq. \eqref{eq:spd_analytic_general_eig_uncorrected}, we can obtain a lower bound on $\psiacf(l)$, which will be tight if the variance in survival function over different node pairs is small:
\begin{equation}
    \label{eq:spd_general_psi_agg_bound}
    \psiacf(l) \ge \prod_{i=1}^N\exp\left(-\frac{1}{n}S_l(\tau_i)\expectwrt{\mu}{\varphi_i(x)}^2\right).
\end{equation}
This permits a description of the network's expected geodesic length in terms of the eigenvalues and expectation of eigenfunctions of $T$. Fig. \ref{fig:grgg_sbm_eig} illustrates the eigendecomposition of $T$ for a homophilous stochastic block model over $V=\{1,2,\dots, 32\}$, constructed by discretizing a Gaussian random geometric graph over $V=\real$.
\begin{figure}
    \centering
    \includegraphics[width=\columnwidth]{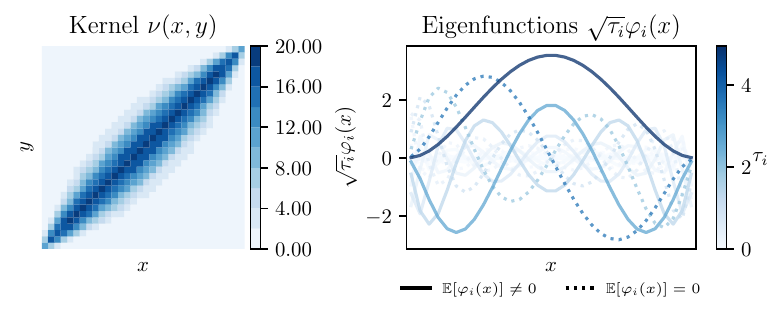}
    \caption{The closed-form GLD can be decomposed over the eigenfunctions $\{\varphi_i(x)\}_{i=1}^N$ of the integral operator $T$ of the connectivity kernel $\nu(x,y)$ (Eqs. \eqref{eq:spd_analytic_general_eig}, \eqref{eq:spd_analytic_general_eig_uncorrected}), shown here for a 32-block SBM (over $V=[32]$) formed by discretizing a Gaussian RGG (over $V=\real$); see Fig. \ref{fig:spl_grgg_sbm_rank1} in Appendix \ref{sec:apdx_general} for model details. When sorted by decreasing eigenvalues $\tau_i$, eigenfunctions $\varphi_i$ of even-numbered $i$ do not contribute to the bound on the approximate closed form of the survival function of the GLD for an average node pair, in Eq. \eqref{eq:spd_general_psi_agg_bound}.}
    \label{fig:grgg_sbm_eig}
\end{figure}

\section{\label{sec:geodesics_specific}Key statistical network models}
In this section, we consider network models that are important special cases of the SGRN model considered in Sec. \ref{sec:general_graphs}. We focus on symmetric connectivity kernels, equivalent to the sparse inhomogeneous random graph models of Ref. \onlinecite{bollobas2007phase}. In particular, we elaborate on the percolation probability of nodes, and the approximate closed form of the survival function of GLD between node pairs---which is asymptotically tight for finite lengths in the supercritical regime, when the network is almost surely fully connected, and for all lengths in the subcritical regime.

\subsection{Stochastic block model (SBM)}\label{sec:sbm}

Stochastic block models (SBMs) have been widely used for modeling social networks \cite{holland1983sbm, karrer2011dcsbm}, since they directly capture notions of social homophily that ``like befriends like'' \cite{mcpherson1987homophily, mcpherson2001birds}, and social segregation \cite{moody2001race}. They are discrete space models wherein nodes are divided into communities or ``blocks'', whose probabilistic connections are modeled by block-level parameters. In essence, SBMs are a form of ``network coarsening'' wherein edges between nodes can be aggregated into edges between blocks of nodes \cite{peixoto2014nestedsbm}. This property can be leveraged to study the GLD in empirical networks (see Sec. \ref{sec:graphcoarsening}). Theorem \ref{lemma:general_sbm_equiv} in Appendix \ref{sec:apdx_general} shows that there exists an $\epsilon$-equivalent SBM for any SGRN model, that is, one can approximate a continuous space model by an SBM via discretization up to a desired level of accuracy. Establishing a framework for the GLD in SBMs therefore yields applications in a large variety of settings.

\paragraph*{Setup} Consider a discrete node space $V=[k]$ wherein each of $n$ nodes belong to one of $k$ blocks, according to a categorical distribution: $\mu(x)=\pi_x$ where $\vect{\pi}$ is a $(k-1)$-standard simplex vector in $\real^k$ called the ``distribution vector''. The probability of two nodes connecting to each other depends on the blocks they belong to: $\nu(x,y)=B_{xy}/n$ where $B_{xy}\ge 0$ measures the ``affinity'' between blocks $x$ and $y$ and is called the ``block matrix'' $\mat{B}$. Using Eq. \eqref{eq:gcc_consistency_general} we obtain an equation for the length-$k$ ``block percolation'' vector $\boldsymbol{\rho}$, where $\rho_x$ is the percolation probability for a node in block $x$:
\begin{equation}
    \label{eq:gcc_consistency_sbm}
    \boldsymbol{\rho}\approx\ones{k}-\exp\left(-\mat{B}\mat{\Pi}\boldsymbol{\rho}\right),
\end{equation}
where we define $\mat{\Pi}\triangleq\diag{\boldsymbol\pi}$ as the diagonal distribution matrix, and $\ones{k}$ is the all-ones vector of length $k$.

\paragraph*{Analytic form of the GLD}From Eq. \eqref{eq:spd_main_general} we get recursive equations for $k\times k$ ``block survival function'' and ``block cumulative PMF'' matrices $\Psiaf_l, \Omegaaf_l$ respectively:
\begin{subequations}
    \label{eq:spd_main_sbm}
    \begin{align}
        \Psiaf_l &= (\ones{k}\ones{k}^T - \Omegaaf_l)\odot\Psiaf_{l-1}, \\
        \Omegaaf_l &\approx \ones{k}\ones{k}^T - \exp \Bigg(n\sum_{x=1}^k\pi_x\log\Bigg(\ones{k}\ones{k}^T\nonumber\\
        &\hspace{100pt}-\frac{[\Omegaaf_{l-1}\odot\Psiaf_{l-2}]_{:x}[\mat{B}]_{x:}}{n}\Bigg)\Bigg),
    \end{align}
\end{subequations}
with the initial condition from Eq. \eqref{eq:init_omega_general} yielding:
\begin{equation}
    \label{eq:spd_sbm_init}
    \Omegaaf_1\approx\frac{\mat{B} + (\mat{R}^{-1}-\eyes{k})\mat{BR}}{n},
\end{equation}
where $\mat{R}\triangleq\diag{\boldsymbol{\rho}}$, $\Psiaf_0=\ones{n}\ones{n}^T$, and $\eyes{k}$ is the identity matrix of size $k\times k$. We have used the notation $[\cdot]_{:i}$ ($[\cdot]_{i:}$) to indicate the $i^\mathrm{th}$ column (row) vector of the matrix. Fig. \ref{fig:bipartite_cdf} shows the distribution of geodesic lengths between nodes in a $2$-block SBM with a bipartite structure, obtained by solving Eqs. \eqref{eq:gcc_consistency_sbm}--\eqref{eq:spd_sbm_init}, which is in good agreement with the empirical GLD. Refer to Fig. \ref{fig:bipartite_cdf_smallcomponent} in Appendix \ref{sec:apdx_finite_size} for the GLD of a bipartite SBM in a subcritical regime.
\begin{figure}
    \centering
    \includegraphics[width=\columnwidth]{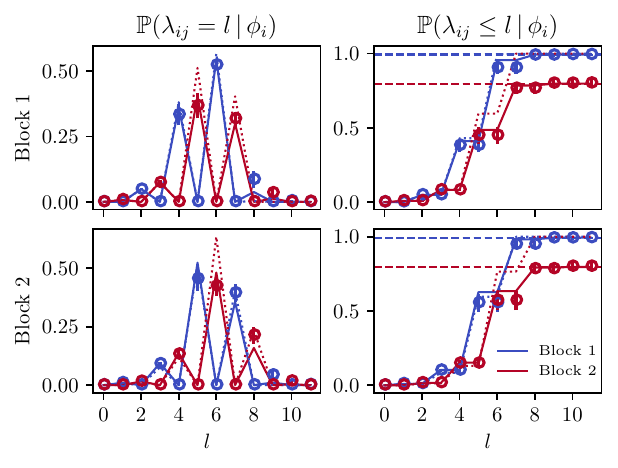}
    \caption{Empirical, analytic, and approximate closed-form CDF of geodesic lengths where the source node is on the giant component, agree with each other for a bipartite SBM. The block matrix $\mat{B}=\big(\begin{smallmatrix} 0 & 8\\ 8 & 0 \end{smallmatrix}\big)$, distribution vector $\boldsymbol\pi=(0.2, 0.8)$, and network size $n=1024$. Rows correspond to the block membership of source node. Left column depicts the PMF, which highlights bipartitivity of the network, and the right column depicts the CDF, whose tail value agrees with the percolation probability of target node, indicated by the dashed asymptote and calculated from Eq. \eqref{eq:gcc_consistency_sbm}. Solid lines represent analytic form using Eqs. \eqref{eq:gcc_consistency_sbm}--\eqref{eq:spd_sbm_init}, while dotted lines represent approximate closed form using Eq. \eqref{eq:spd_sbm}, and symbols ($\circ$) with bars represent empirics, i.e. mean and standard error over 10 samples.}
    \label{fig:bipartite_cdf}
\end{figure}

\paragraph*{Approximate closed form of the GLD}Using Eq. \eqref{eq:spd_analytic_general_omega} yields the survival function of the GLD as the summation over matrix powers (see Appendix \ref{sec:apdx_sbm}), producing the $k\times k$ block survival function matrix:
\begin{equation}
    \label{eq:spd_sbm}
    \Psiacf_l = \exp\left(-\frac{\mat{B}}{n}\sum_{i=1}^l(\mat{\Pi B})^{i-1}\right),
\end{equation}
analogous to Eq. \eqref{eq:sf_avg_uncorrected} wherein $\mat{\Pi B}$ works analogously to $\avg{\mat{A}}$. Fig. \ref{fig:bipartite_cdf} shows the approximate closed form of the GLD between nodes in a $2$-block bipartite SBM, obtained by solving Eq. \eqref{eq:spd_sbm}. As previously discussed in Sec. \ref{sec:closedform_bound}, the approximate closed form GLD agrees with the empirical GLD for shorter geodesic lengths in finite-size networks. If $\mat{\Pi B}-\eyes{k}$ is non-singular, we can evaluate the expression in Eq. \eqref{eq:spd_sbm} as a matrix series to obtain $$\Psiacf_l = \exp\left(-\frac{\mat{B}}{n}((\mat{\Pi B})^l-\eyes{k})(\mat{\Pi B}-\eyes{k})^{-1}\right).$$ Generically $\Psiacf_l$ is a matrix, but for the special case of an ER graph---where the probability of an edge is identical across all node pairs i.e. $k=1$, $\mat{B}\triangleq\avgdegree\ones{1}\ones{1}^T$, $\mat{\Pi}\triangleq \eyes{1}$---the survival function is a scalar:
\begin{equation}
    \label{eq:spd_er}
    \psiacf(l) = \begin{cases}
    \exp\left(-\frac{\avgdegree(\avgdegree^l-1)}{n(\avgdegree-1)}\right) &\mbox{if }\avgdegree\ne 1,\\
    \exp\left(-\frac{l}{n}\right) &\mbox{otherwise.}
    \end{cases}
\end{equation}
Evidently, larger the mean degree, shorter the geodesic lengths in the network. Previous work has demonstrated a closed-form expression for the survival function of geodesic lengths in ER graphs, given by Eq. (14) in Ref. \onlinecite{fronczak2004average} as $\psiacf(l) = \exp\left(-\frac{\avgdegree^l}{n}\right)$ which is in disagreement with Eq. \eqref{eq:spd_er}. Particularly at $l=0$, Eq. \eqref{eq:spd_er} will correctly evaluate to $1$, while the other to $\exp(-n^{-1})$. Notably, as discussed in Sec. \ref{sec:closedform_bound} and unlike Eq. (14) in Ref. \onlinecite{fronczak2004average}, Eq. \eqref{eq:spd_er} works especially well when $\avg{d}<1$ and the network is in the subcritical regime; see Fig. \ref{fig:spd_er_smallcomponent} in Appendix \ref{sec:apdx_finite_size} where the empirics and analytics from Eq. \eqref{eq:spd_er} agree well for subcritical ER graphs of varying mean degree.

\paragraph*{Illustrative example}For a less trivial example, we consider a $k$-block SBM with equi-sized blocks and constant mean degree $\avgdegree$, such that $\mat{B}=\delta \eyes{k}+\left(\avgdegree-\delta/k\right)\ones{k}\ones{k}^T$, where $\delta\in[-\avgdegree k/(k-1),\avgdegree k]$ quantifies the amount of homophily---positive $\delta$, wherein nodes are more likely to connect to nodes from the same block---or heterophily---negative $\delta$, wherein nodes are more likely to connect to nodes from other blocks. As before, let $S_l(a)$ be the geometric sum of $a$ starting at $a$ up to $l$ terms, then from Eq. \eqref{eq:spd_sbm} the survival function block matrix is given by (see Appendix \ref{sec:apdx_sbm})
\begin{equation}
    \label{eq:spd_sbm_k}
    \Psiacf_l = \exp\boldsymbol{\bigg(}-\frac{1}{n}\left\{S_l(\delta/k)k\eyes{k}+[S_l(\avgdegree)-S_l(\delta/k)]\ones{k}\ones{k}^T\right\}\boldsymbol{\bigg)}.
\end{equation}
The form of $\Psiacf_l$ is analogous to that of the block matrix $\mat{B}$---naturally $\Psiacf_1=\mat{B}$, but for larger $l$ the off-diagonal (inter-block) elements of the block survival function matrix evolve as the exponential of difference in geometric sums of $S_l(\delta/k)$ and $S_l(\avgdegree)$. In particular, consider a 2-block perfectly heterophilous SBM, i.e. $k=2$ and $\delta/k=-\avgdegree$, which corresponds to a bipartite network since nodes never connect directly with nodes of their own block. Consequently, all paths between nodes of different blocks must be of odd length. The expression in Eq. \eqref{eq:spd_sbm_k} correctly suggests that the inter-block survival function does not change for even values of $l$ due to canceling out of the even powers of $\avgdegree$, that is, there is no probability mass at even values of $l$.

\paragraph*{Spectral approximate closed form of the GLD}We next consider a general SBM with a symmetric block matrix $\mat{B}$; let $\mat{Q\Lambda Q}^T$ be the eigendecomposition of the symmetric matrix $\mat{\Pi}^\frac{1}{2}\mat{B\Pi}^\frac{1}{2}$ such that columns of the orthogonal matrix $\mat{Q}$ encode the eigenvectors and the diagonal matrix $\mat{\Lambda}$ encodes the corresponding eigenvalues. Then, we can write $(\mat{\Pi B})^{i-1}=\mat{\Pi}^\frac{1}{2}\mat{Q\Lambda}^{i-1}\mat{Q}^T\mat{\Pi}^{-\frac{1}{2}}$. Putting in Eq. \eqref{eq:spd_sbm}, we obtain 
\begin{equation}
    \label{eq:spd_sbm_eig}
    \Psiacf_l = \exp\left(-\frac{1}{n}\mat{\Pi}^{-\frac{1}{2}}\mat{Q}S_l(\mat{\Lambda})\mat{Q}^T\mat{\Pi}^{-\frac{1}{2}}\right),
\end{equation}
which is the matrix analogue of Eq. \eqref{eq:spd_analytic_general_eig_uncorrected}. Clearly, for general symmetric block matrices, a weighted geometric sum of eigenvalues of $\mat{B\Pi}$ governs the whole GLD, with positive eigenvalues---indicating homophily---and negative eigenvalues---indicating heterophily---contributing differently to the GLD. An equation similar to---but different from---Eq. \eqref{eq:spd_sbm_eig} has been previously described for the GLD in SBMs with the additional constraint of an irreducible and aperiodic block matrix, in terms of the limiting random variables of two independent branching processes initiated at the source and target nodes of a network \cite{barbourreinert2011bipartitesbm}. Since the expectation of the limiting random variables is unity, Eq. \eqref{eq:spd_sbm_eig} can be partially interpreted as an application of Jensen's inequality \cite{jensen1906fonctions} to Eq. (6.31) in Ref. \onlinecite{barbourreinert2011bipartitesbm}---in the form $\exp(\expect{\cdot})\le\expect{\exp(\cdot)}$---although we note that Eq. (6.31) in Ref. \onlinecite{barbourreinert2011bipartitesbm} ignores the contribution of non-leading eigenvalues. Consequently, in the scenario where the non-leading eigenvalues are negligible, Eq. \eqref{eq:spd_sbm_eig} provides an asymptotic lower bound to the survival function of the GLD from Eq. (6.31) in Ref. \onlinecite{barbourreinert2011bipartitesbm}, which aligns well with the derivation of the closed form in Sec. \ref{sec:closedform_bound}.

\subsection{\label{sec:graphcoarsening}Labeled empirical network}

Previous work in estimating dissimilarity measures on empirical networks has determined a Gibbs–Boltzmann distribution on picking a path between two nodes over a bag-of-paths in a given network \cite{franccoisse2017bagofpaths}. However, this bag-of-paths approach does not directly model the distribution of geodesic lengths between node pairs. Our proposed approach provides the desired distribution, but for networks generated by an underlying sparse statistical network model with conditionally independent edges. One way to apply our approach is to first infer such a model given the observed network; inference can be performed in a myriad of ways \cite{newman2016inferenceannotated, newman2018inference, goldenberg2010inference}, but may potentially incur computational overhead if the network is very large.

A computationally cheap form of inference exploits the network-coarsening property of SBMs described in Sec. \ref{sec:sbm}, wherein nodes are completely defined by their block membership. This has the advantage of reducing the parametrization from one in terms of a large number of nodes $n$, to one in terms of a small number of blocks $k\ll n$. For a network represented by the adjacency matrix $\mat{A}$, with known node labels indicated by the $n\times k$ assignment matrix $\mat{Z}$, assuming that $(\mat{A},\mat{Z})$ is generated by a $k$-block SBM permits a maximum likelihood estimate of its parameters $(\mat{B},\vect{\pi})$ regardless of how node labels were determined, and with no computational cost except for summation of node and edge counts. We refer the reader to Appendix \ref{sec:apdx_gcsbm} for more details. The GLD associated with the inferred SBM can then be used to study geodesic lengths in the empirical network.

\paragraph*{Illustrative example: completely observed network} We consider a real-world network of email communications between $n=986$ members of a European research institution \cite{snapeuemail, snapcollab-euemail-gnutella, snapnets} denoted by the adjacency matrix $\mat{A}_\text{eue}$. Each node has an attribute corresponding to one of the $k=42$ departments which the individual belongs to, which can serve to provide the assignment matrix $\mat{Z}_\text{dep}$ at no additional cost. We can also derive more meaningful community labels for the nodes by applying community detection methods \cite{girvan2002community, newman2006modularity}, such as modularity maximization \cite{clauset2004modmax, hagberg2008networkx}, which can provide a different assignment $\mat{Z}_\text{mod}$ across $k=5$ modules. We also infer a hierarchical SBM for this network that can infer blocks at multiple levels of coarsening \cite{peixoto2014nestedsbm, peixoto2014graphtool}, generating a hierarchy of labels yielding assignments $\mat{Z}_\text{sbm2}, \mat{Z}_\text{sbm3}$ at levels 2 and 3 of the inferred hierarchical SBM, possessing $k=43$ and $k=9$ blocks respectively. Using Eq. \eqref{eq:sbm_mle} we can derive corresponding SBMs, and the analytic form of the GLD between block pairs using Eq. \eqref{eq:spd_main_sbm}. In Fig. \ref{fig:spl_statistics_empirical} we compare the empirical and analytic means of geodesic lengths for every block pair in this network, which are in good agreement for all assignment procedures considered. Notably, this includes $\mat{Z}_\text{dep}$ which requires no additional computational overhead. The agreement is stronger for block pairs with shorter geodesics between them. The departure for longer geodesics is likely due to correlations within $\mat{A}_\text{eue}$, wherein longer-than-expected geodesics would beget even longer geodesics, resulting in the analytics mostly underestimating the empirics. 

\begin{figure}
    \centering
    \includegraphics[width=\columnwidth]{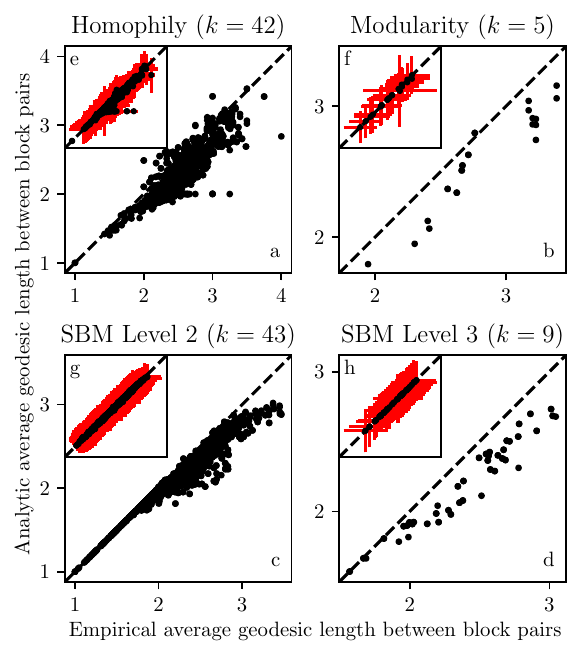}
    \caption{The empirical average geodesic lengths of a real-world email network $\mat{A}_\text{eue}$ (on the $x$-axis) are well-approximated by the average obtained from the analytic form of the GLD (on the $y$-axis, using Eqs. \eqref{eq:gcc_consistency_sbm}, \eqref{eq:spd_main_sbm}, \eqref{eq:spd_sbm_init}), when nodes with the same ``label'' are grouped into a single block to form an SBM with $k$ blocks. Panels (a--d) correspond to different labelings: (top-left) $\mat{Z}_\text{dep}$ leverages the homophily assumption by using a node attribute as the label---here, the department of the e-mailer; (top-right) $\mat{Z}_\text{mod}$ uses modularity maximization \cite{clauset2004modmax, hagberg2008networkx} to infer network modules assigned as the label; (bottom) $\mat{Z}_\text{sbm2}, \mat{Z}_\text{sbm3}$ use a hierarchical SBM \cite{peixoto2014nestedsbm, peixoto2014graphtool} to infer a hierarchy of blocks which allows for multi-level coarsening. Points indicate the mean geodesic length between nodes of block pairs. To confirm that any deviations are only due to suboptimal SBM fits or having a \emph{single} empirical network, insets (e--h) show the mean (black markers) and standard deviation (red bars) of geodesic lengths between block pairs, averaged over 10 samples of the corresponding SBM fit via coarsening. Both the mean and standard deviation of the empirical GLD are well approximated by the analytic GLD.} 
    \label{fig:spl_statistics_empirical}
\end{figure}

\paragraph*{Illustrative example: partially observed network} To demonstrate the empirical power of the GLD, we consider a scenario where the network is not fully observed but there exists a reasonable statistical model for how it could have been generated. Facebook's individual-level friendship network is very large, and can help one understand social segregation at a national scale. However, analyzing its global properties like shortest paths is infeasible, due to privacy concerns and computational intractability. Facebook's social connectedness index dataset encodes, at the region level, the probability of two individuals in any two region pairs to be friends with each other on the platform \cite{scidata}, thus circumventing the concerns of viewing or analyzing a fully observed network. By assuming an SBM, wherein every block is a sub-national region, we can apply our GLD framework to understand segregation on the Facebook friendship network at a sub-national level via shortest paths. Fig. \ref{fig:uk_map_aspl} shows analytic estimates of the mean geodesic lengths for the Facebook friendship network across different regions of the United Kingdom, computed using the approximate closed form of the GLD in Eq. \eqref{eq:spd_sbm}, and compares them against the mean degrees. Evidently, the mean degrees and mean geodesic lengths are negatively correlated, with regions in London and Wales having respectively some of the lowest/highest degrees and longest/shortest geodesic lengths, and regions of England lying in between. However, regions in Scotland and Northern Ireland have a longer mean geodesic length relative to what would be expected based only on their mean degrees, likely due to high spatial homophily: our framework reveals higher-order structural information which the mean degrees alone miss out on.

\begin{figure*}
    \centering
    \includegraphics[width=2\columnwidth]{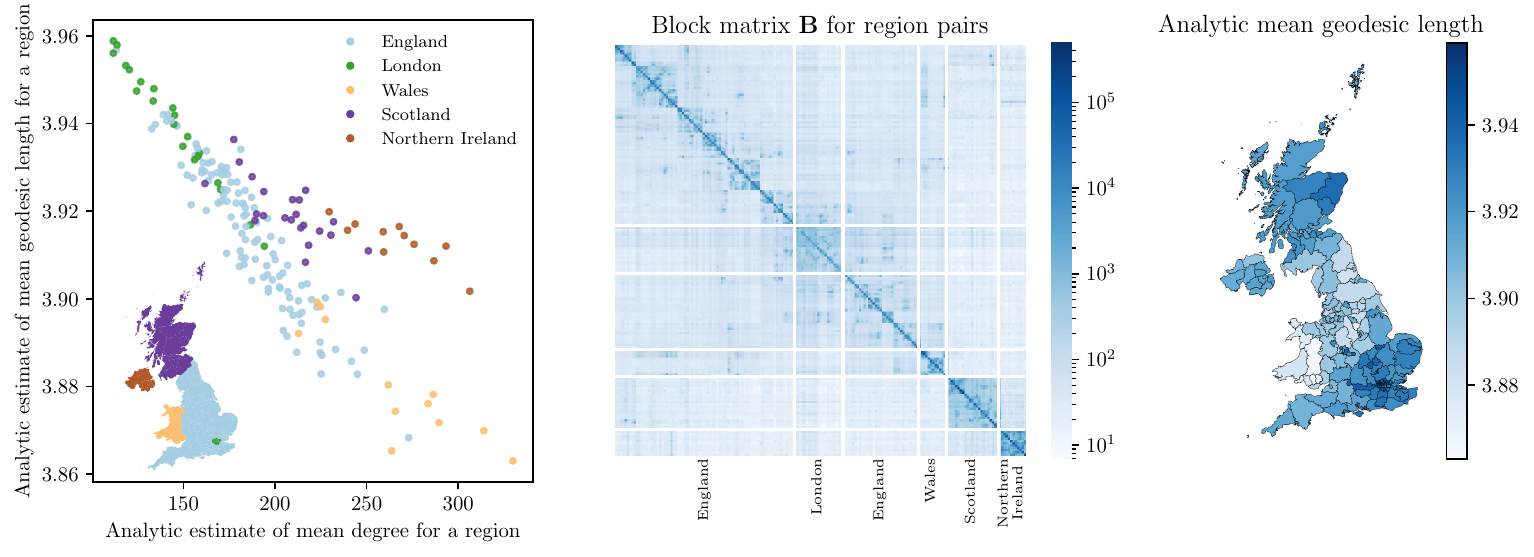}
    \caption{Analytic estimation of mean geodesic lengths in partially observed empirical networks: Sub-national mapping of mean geodesic lengths in the Facebook friendship network for different regions in the United Kingdom, computed from the approximate closed form of the GLD in Eq. \eqref{eq:spd_sbm} by assuming a stochastic block model whose block matrix $\mat{B}$ is given by the social connectedness index data \cite{scidata}, appropriately scaled such that the mean degree is $189.4$ \cite{backstrom2010fbgeo}, and whose distribution vector $\vect{\pi}$ is given by the true population distribution \cite{ukeurostatspopest}. The mean degrees and mean geodesic lengths are negatively correlated (left), but regions in Scotland and Northern Ireland have a longer mean geodesic length relative to what would be expected based only on their mean degrees, likely due to high spatial homophily evident in the relatively higher diagonal dominance of the block matrix $\mat{B}$ (center). The GLD reveals higher-order structural information which the mean degrees alone miss out on.}
    \label{fig:uk_map_aspl}
\end{figure*}

\subsection{\label{sec:rdpg:}Random dot product graph (RDPG)}

While SBMs are used often due to their simplicity and ability to model any community structure, they can be limited by their discreteness. This has led to the exploration of continuous space network models of which the random dot product graph (RDPG) \cite{young2007rdpg} is an important example. Consider some $k$-dimensional bounded real vector space $X \subset\real^k$ wherein nodes are ``embedded'' such that the probability of a pair of nodes at $(\boldsymbol{x},\boldsymbol{y})\in X\times X$ connecting is proportional to a function of the dot product of their positions $\boldsymbol{x}^T\boldsymbol{y}$. This notion of ``graph embeddings'' is especially popular in statistical machine learning, wherein continuous representations of discrete objects (such as graphs or networks) are learned \cite{hamilton2017graphrepresentation, cai2018surveygraphembedding, goyal2018graphembedding}, making them amenable to downstream prediction tasks. Many approaches for generating such an embedding rely on the dot product model: techniques based on matrix factorization, like graph factorization \cite{ahmed2013graphfactorization}, assume that the probability of nodes located at $\boldsymbol{x}$ and $\boldsymbol{y}$ to connect is proportional to $\boldsymbol{x}^T\boldsymbol{y}$, while random-walk based techniques, like ``node2vec'' \cite{grover2016node2vec}, assume that the connection probability is proportional to $\exp(\boldsymbol{x}^T\boldsymbol{y})$.

\paragraph*{Symmetric \& positive semi-definite kernels}We emphasize that a sparse RDPG is a special case of the SGRN model described in Sec. \ref{sec:general_graphs} when the connectivity kernel is symmetric and positive semi-definite, i.e. all eigenvalues of $T$ (Eq. \eqref{eq:integral_op}) are non-negative \footnote{Contrary to SBMs, this excludes the possibility of heterophilous structures like bipartitivity, but results in uniform absolute convergence of the kernel's eigenexpansion in Eq. \eqref{eq:kernel_rep} \cite{riesz1955hilbert, mercer1909kernel}.}. In other words, any positive semi-definite kernel can be written as a dot product in some feature space (see Appendix \ref{sec:apdx_rdpg}). Therefore, for the purpose of deriving the GLD, it is sufficient to consider the simplest setting of $V$ being Euclidean and the kernel being linear in the dot product. For a connectivity kernel that is non-linear in the dot product, we can use random Fourier features \cite{rahimi2007random} to derive an explicit feature map to a space in which the connectivity kernel is linear in the dot product. We refer the reader to Appendix \ref{sec:apdx_rdpg} for an illustrative example with a non-linear dot product connectivity kernel.

\paragraph*{Setup}Consider nodes in a non-negative bounded subspace $V=X\subset\real^k_{\ge 0}$ with density $\mu$, and the connectivity kernel $\nu(\boldsymbol{x}, \boldsymbol{y}) = \beta\boldsymbol{x}^T\boldsymbol{y}/n$ where $\beta> 0$. This is a common setting for RDPGs: in the canonical degree-configuration model where $k=1$, $X$ encodes precisely the expected degree of a node, with the node density $\mu$ governing the degree distribution \cite{young2007rdpg}. We define:
\begin{subequations}
\begin{align}
    \label{eq:def_rdpg_meanvec}
    \boldsymbol{\phi}&\triangleq\int_X \boldsymbol{x}\,d\mu, \\
    \label{eq:def_rdpg_mommat}
    \mat{\Phi}&\triangleq \beta\int_X \boldsymbol{x}\boldsymbol{x}^T\,d\mu,
\end{align}
\end{subequations}
where $\boldsymbol{\phi}$ indicates the length-$k$ mean vector in $X$, and $\mat{\Phi}$ refers to the $k\times k$ matrix of second moments scaled by $\beta$ that encodes the covariance in space $X$ as per the measure $\mu$ and is necessarily positive semi-definite. Given that $\rho(\boldsymbol{x})$ encodes the percolation probability for a node at $\boldsymbol{x}$, define $\boldsymbol{\rho}\triangleq\int_X \boldsymbol{x}\rho(\boldsymbol{x})\,d\mu$ to be the ``mean percolation vector'' in $X$, and Eq. \eqref{eq:gcc_consistency_general} yields:
\begin{subequations}
\label{eq:gcc_consistency_rdpg}
\begin{align}
    \label{eq:gcc_consistency_rdpg_rhox}
        &\rho(\boldsymbol{x})\approx 1-\exp\left(-\beta\boldsymbol{x}^T\int_X\boldsymbol{y}\rho(\boldsymbol{y})\,d\mu\right)\nonumber\\&\hspace{20pt}=1-\exp\left(-\beta\boldsymbol{x}^T\boldsymbol{\rho}\right),\\
    \label{eq:gcc_consistency_rdpg_rho}
    \begin{split}
        \implies&\int_X\boldsymbol{x}\rho(\boldsymbol{x})\,d\mu\approx \int_X\boldsymbol{x}\,d\mu-\int_X\boldsymbol{x}\exp \left(-\beta\boldsymbol{x}^T\boldsymbol{\rho}\right)d\mu,\\
        \implies&\boldsymbol{\rho}\approx\boldsymbol{\phi}-\int_X\boldsymbol{x}\exp\left(-\beta\boldsymbol{x}^T\boldsymbol{\rho}\right)d\mu,
    \end{split}
\end{align}
\end{subequations}
where we apply the definition of $\boldsymbol\phi$ from Eq. \eqref{eq:def_rdpg_meanvec}. Eq. \eqref{eq:gcc_consistency_rdpg_rho} is a self-consistent vector equation for $\boldsymbol{\rho}$ which once solved can be used to solve the self-consistent scalar Eq. \eqref{eq:gcc_consistency_rdpg_rhox} for the percolation probability of any node location. To solve the former, we can make use of a discrete SBM approximation of the $k$-dimensional RDPG---as described in Appendix \ref{sec:apdx_general}---and solve for $\boldsymbol{\rho}$ numerically via function iteration.

\paragraph*{Approximate closed form of the GLD}For this and subsequent sections (Secs. \ref{sec:rgg} and \ref{sec:graphons}), we focus on the approximate closed form of the survival function of the GLD. In particular, Eq. \eqref{eq:spd_general_omega} for the conditional PMF would read as $\omegaacf_l(\boldsymbol{x},\boldsymbol{y})=\beta\boldsymbol{x}^T\left[\beta\int_X\boldsymbol{z}\boldsymbol{z}^T\,d\mu(\boldsymbol{z})\right]^{l-1}\boldsymbol{y}/n$, translating Eq. \eqref{eq:spd_general_psi} for the survival function into
\begin{equation}
    \label{eq:spd_rdpg}
    \psiacf_l(\boldsymbol{x}, \boldsymbol{y}) = \exp\left(-\frac{\beta}{n}\boldsymbol{x}^T\left(\sum_{k=0}^{l-1}\mat{\Phi}^k\right)\boldsymbol{y}\right),
\end{equation}
where we use the definition of $\mat{\Phi}$ from Eq. \eqref{eq:def_rdpg_mommat}. Since $\mat{\Phi}$ is symmetric, let $\mat{Q\Lambda Q}^T$ be its eigendecomposition, such that columns of the orthogonal matrix $\mat{Q}$ encode the eigenvectors and the diagonal matrix $\mat{\Lambda}$ encodes corresponding eigenvalues. Let $\widetilde{S}_l(a)$ be the geometric sum of $a$ starting at $1$ up to $l$ terms. Then Eq. \eqref{eq:spd_rdpg} provides:
\begin{equation*}
    \psiacf_l(\boldsymbol{x}, \boldsymbol{y}) = \exp\left(-\frac{\beta}{n}\boldsymbol{x}^T\mat{Q}\widetilde{S}_l(\mat{\Lambda})\mat{Q}^T\boldsymbol{y}\right).
\end{equation*}
This is analogous to the expressions obtained via eigendecomposition for the SGRN model in Eq. \eqref{eq:spd_analytic_general_eig_uncorrected}, for ER graphs in Eq. \eqref{eq:spd_er}, and for SBMs in Eq. \eqref{eq:spd_sbm_eig}. Alternately, consider the expected survival function of the whole network in Eq. \eqref{eq:spd_general_psi_agg}. By applying Jensen's inequality \cite{jensen1906fonctions} to Eq. \eqref{eq:spd_rdpg}, we obtain a lower bound on $\psiacf(l)$, similar to the one in Eq. \eqref{eq:spd_general_psi_agg_bound}, which will be tight if the variance in survival function over node pairs is small:
\begin{equation}
    \label{eq:spd_rdpg_network}
    \psiacf(l) \ge \exp\left(-\frac{\beta}{n}\boldsymbol{\phi}^T\widetilde{S}_l(\mat{\Phi})\boldsymbol{\phi}\right),
\end{equation}
where we use the definition of $\boldsymbol\phi$ from Eq. \eqref{eq:def_rdpg_meanvec}. This permits a description of the network's expected geodesic length entirely in terms of the first and second moments in $X$, which can be especially useful when we have access to just the sample mean and covariance, instead of the full distribution in $X$.

\paragraph*{Illustrative example} We consider $X$ restricted to the $(k-1)$-standard simplex in $\real^k$, and $\mu$ corresponding to the Dirichlet distribution on that simplex given by the concentration vector $\boldsymbol{\alpha}\in[0,\infty)^k$. This allows us to interpret the node's location as the probability of belonging to one of $k$ communities located at the corners of the simplex---a continuous analogue of the SBM. Let $\avgdegree$ be the mean degree of the network, and $\bar\alpha\triangleq\boldsymbol\alpha^T\ones{k}$, then it can be shown that the approximate closed form of the conditional PMF of the GLD is given by: $\omegaacf_l(\boldsymbol{x},\boldsymbol{y})=\frac{\avgdegree\bar\alpha^2}{n\left\lVert\boldsymbol{\alpha}\right\rVert^2}\boldsymbol{x}^T\left\{\frac{\avgdegree\bar\alpha}{\left\lVert\boldsymbol{\alpha}\right\rVert^2(1+\bar\alpha)}\left[\diag{\boldsymbol{\alpha}}+\boldsymbol{\alpha}\boldsymbol{\alpha}^T\right]\right\}^{l-1}\boldsymbol{y}$ (see Appendix \ref{sec:apdx_rdpg}). In Fig. \ref{fig:spl_rdpg}, we plot various node and node pair statistics for a ``Dirichlet RDPG'' when $n=512, \avgdegree=4$ and $\boldsymbol\alpha=\{0.8, 0.8, 2\}$ using Eq. \eqref{eq:degree_rdpg} for a node's degree, Eq. \eqref{eq:gcc_consistency_rdpg_rhox} for a node's percolation probability, and the approximate closed form of the GLD in Eq. \eqref{eq:spd_rdpg} to compute an analytic estimate of expected geodesic length between node pairs. In Fig. \ref{fig:spl_emp_vs_ana} we further show that this analytic estimate of the geodesic length is in good agreement with the empirical estimates.

\begin{figure*}
    \centering
    \subfloat[Dirichlet RDPG ($k=3$)]{\label{fig:spl_rdpg} \includegraphics[width=\columnwidth]{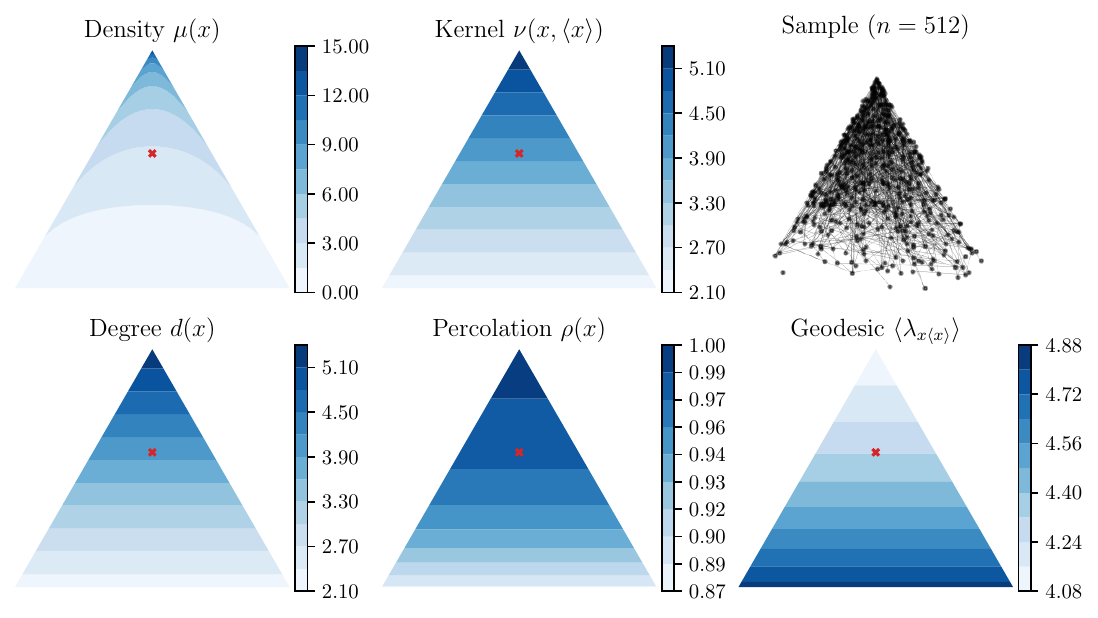}} 
    \subfloat[Gaussian RGG ($k=2$)]{\label{fig:spl_grgg} \includegraphics[width=\columnwidth]{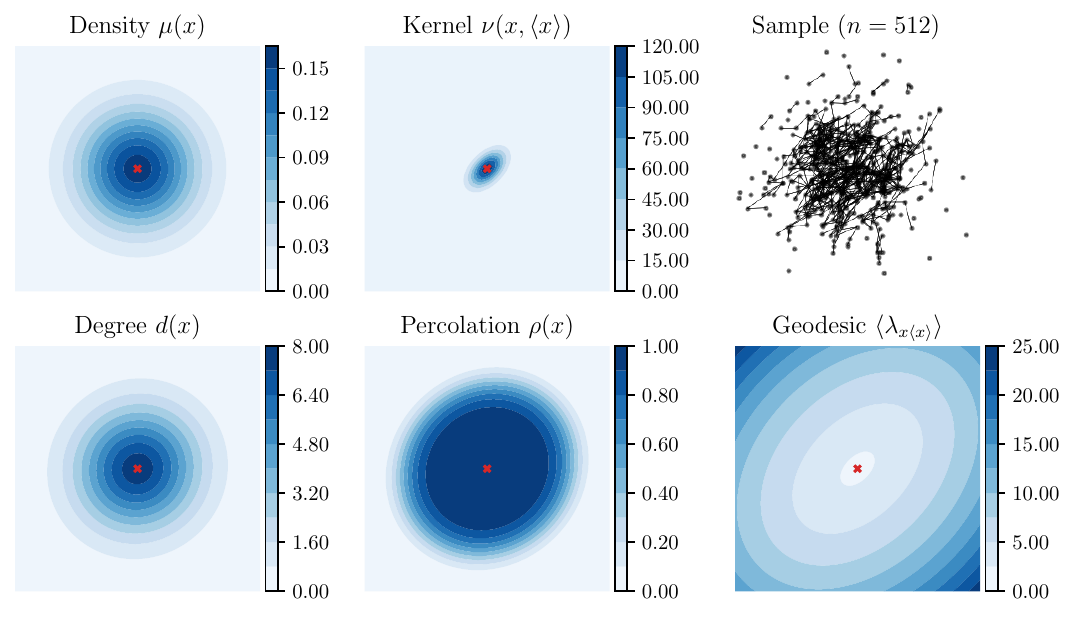}}\\
    \subfloat[Max graphon ($k=1$)]{\label{fig:spl_maxg} \includegraphics[width=\columnwidth]{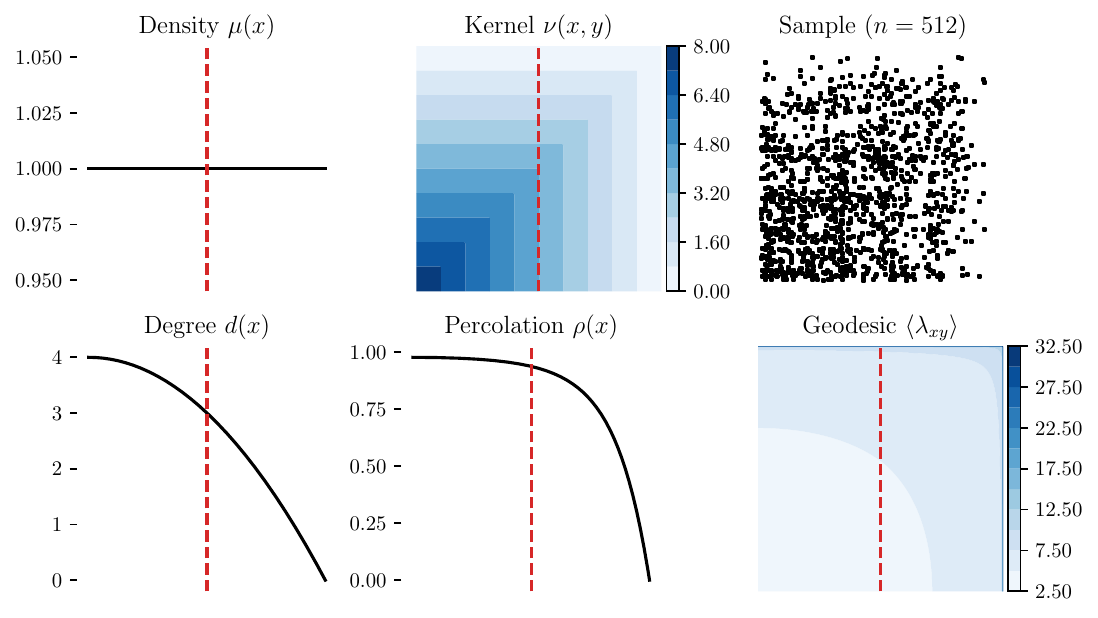}}
    \subfloat[Scale-free graphon ($k=1$)]{\label{fig:spl_sfg} \includegraphics[width=\columnwidth]{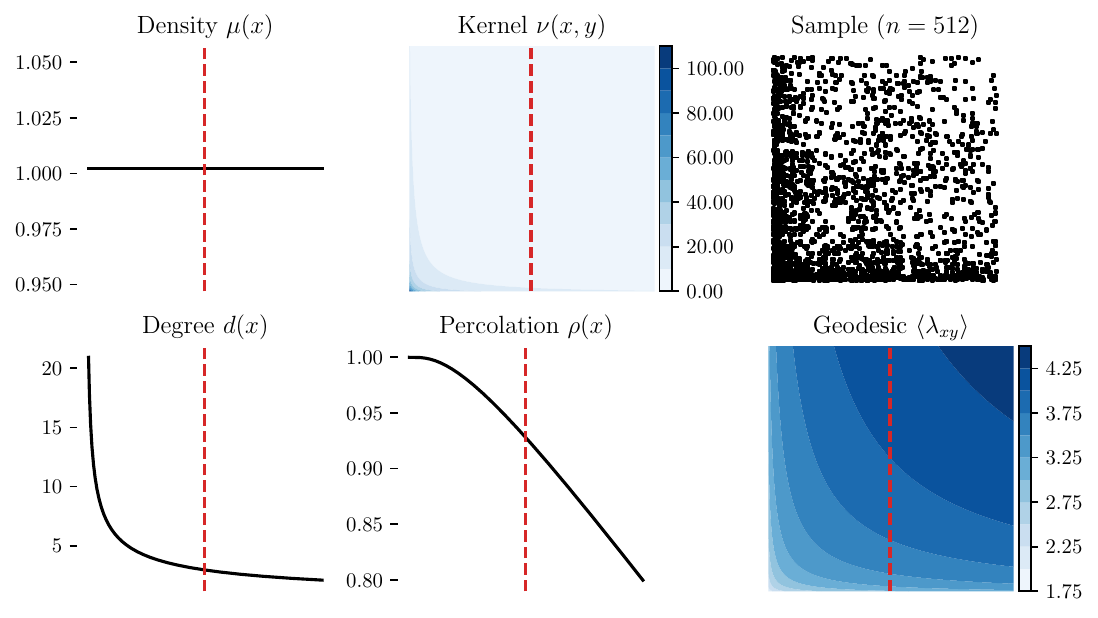}}
    \caption{Node and node pair functions for various statistical network models considered in Sec. \ref{sec:geodesics_specific}. Density function $\mu(x)$ refers to the distribution of nodes in node space $V$, connectivity kernel $\nu(x,y)$ refers to probability of connection between node pairs, degree function $\meandegree(x)$ refers to the expected degree of a node at $x$, percolation function $\rho(x)$ indicates the probability of a node at $x$ to be on the giant component, and geodesic function $\avg{\lambda_{xy}}$ refers to the expected geodesic length between node pairs. For one-dimensional models---(c) max graphon ($V=[0,1]$) and (d) scale-free graphon ($V=[0.01,1]$)---node functions ($\mu$, $\rho$, and $\meandegree$) are shown on $V$, while node pair functions ($\nu$ and $\lambda$) and network sample are shown on $V\times V$. For higher dimensional models---(a) Dirichlet RDPG (functions shown on the standard 2-simplex) and (b) Gaussian RGG (functions shown on $[-3,3]\times[-3,3]$)---node pair functions are shown between $x$ and the mean in $V$ indicated by $\avg{x}$, which is itself marked by red dashed lines or crosses. Description of model parameters and equations used to compute these functions are in respective sections. We note that $\lambda$ is estimated from the approximate closed form of the GLD.}
    \label{fig:spl_models}
\end{figure*}

\begin{figure}
    \centering
    \includegraphics[width=\columnwidth]{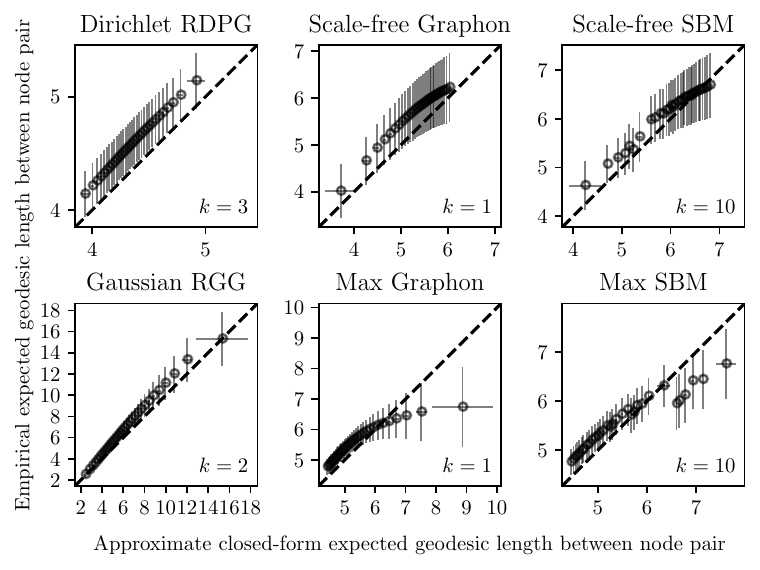}
    \caption{Estimates of expected geodesic lengths between node pairs using the approximate closed form of the GLD (on the $x$-axis) agree with empirical estimates (on the $y$-axis) for the sparse general random network models considered in Sec. \ref{sec:geodesics_specific}. For each model, a total of 1024 network samples---each of size $n=512$---were generated in two steps: 32 node-location samples (based on node density $\mu$) followed by 32 edge samples (based on the connectivity kernel $\nu$) for each of those node-location samples. For every node pair in a given node-location sample, if they were on the giant component then their expected geodesic length was empirically estimated by averaging over the empirical geodesic lengths for that node pair in their edge samples, and analytically estimated using the approximate closed form of the GLD. To prevent clutter, the range of closed-form estimates so-derived was divided into 32 partitions such that every partition contained data for an equal number of node pairs, and the mean ($\circ$) and standard deviation (bars) of closed-form and empirical estimates within each partition were plotted. Any deviations between theory and empirics are discussed in the corresponding text.}
    \label{fig:spl_emp_vs_ana}
\end{figure}

\subsection{\label{sec:rgg}Gaussian random geometric graph (RGG)}

An inner product space is a convenient abstraction when the node space is latent, as in graph embeddings, but for real-world spaces equipped with distances over which the probability of connection decays---like \emph{spatial} networks \cite{barnett2007spatially}---it is appropriate to consider a metric space of nodes: a notion captured by random geometric graph (RGG) models \cite{penrose2003rgg, penrose2016rgg, dettmann2016rgg}. The metric space usually depends on the nature of networks being modeled, such as a Euclidean space for communication networks or a hyperbolic space for social networks \cite{krioukov2010hyperbolic, barthelemy2011spatial}.

\paragraph*{Setup}In this section, we focus on \emph{soft} random geometric graphs, wherein the probability of connection decays smoothly with distance; specifically, we consider a $k$-dimensional Euclidean space $V=\real^k$, with a squared-exponential decay function as the connectivity kernel \cite{penrose2016rgg}. This is akin to having an ellipsoidal connection ``bubble'' around each node:
\begin{equation}
    \label{eq:grgg_nu}
    \nu(\boldsymbol{x},\boldsymbol{y}) = \frac{\beta}{n}\exp\left(-\frac{1}{2}(\boldsymbol{x}-\boldsymbol{y})^T\mat{R}^{-1}(\boldsymbol{x}-\boldsymbol{y})\right),
\end{equation}
where $\beta > 0 $ is the probability of connecting to a node with identical coordinates, and $\mat{R}$ is a $k\times k$ symmetric positive-definite matrix encoding the scale of connections in this node space. 
For concreteness, assume a standard multivariate Gaussian distribution of nodes centered at the origin:
\begin{equation}
    \label{eq:grgg_mu}
    \mu(\boldsymbol{x})=(2\pi)^{-\frac{k}{2}}\exp\left(-\frac{1}{2}\boldsymbol{x}^T\boldsymbol{x}\right).
\end{equation}
We remark that this formalism extends to a multivariate Gaussian node distribution through an affine transformation of the node space $V$ and scale matrix $\mat{R}$ (see Appendix \ref{sec:apdx_grgg}). We refer to this as a Gaussian RGG \footnote{Arguably, this should be termed a \emph{doubly} Gaussian RGG, where both the node distribution and connectivity kernel are Gaussian; see Ref. \onlinecite{garrod2018connectivity}.}. We now consider the percolation probability for a node at $\boldsymbol{x}$. Similar to the discrete approximation used for RDPGs, we can discretize the node space into a fine grid to derive percolation probabilities (see Appendix \ref{sec:apdx_general}). However, here we apply an ansatz that the percolation probability is given by a generalization of the Gaussian function:
\begin{equation}
    \label{eq:rho_grgg}
    \rho(\boldsymbol{x}) = a\exp\left(-\left(\boldsymbol{x}^T\mat{C}\boldsymbol{x}\right)^b\right),
\end{equation}
where $0\le a\le 1$ governs percolation probability at the origin, $b\ge 1$ controls the shape of the percolation surface, and $\mat{C}$ is a $k\times k$ symmetric positive semi-definite matrix that indicates the scale of the surface. As shown in Appendix C\,2 of Ref. \onlinecite{loomba2024geodesics2}, $\mat{C}$ commutes with the scale matrix $\mat{R}$, implying that $\mat{C}$ and $\mat{R}$ preserve each others eigenspaces. Consequently, we only need to infer the $k$ non-negative eigenvalues of $\mat{C}$, yielding a total of $k+2$ parameters to fit the percolation surface at grid locations via constrained optimization. We can then use Eq. \eqref{eq:rho_grgg} to obtain percolation probabilities at any node location in $\real^k$.

\paragraph*{Approximate closed form of the GLD} Using the expression in Eq. \eqref{eq:spd_general_omega}, we can write the approximate closed form of the conditional PMF of the GLD succinctly via a set of recursive coefficients: see Eq. \eqref{eq:grgg_omega_ansatz} in Appendix \ref{sec:apdx_grgg}. To make further analytical progress, we consider a special scenario of high ``spatial homophily'': $\mat{R}\to \mat{0}$. Using Eq. \eqref{eq:spd_general_omega} with the na\"{i}ve initial condition in Eq. \eqref{eq:prob_connect_apx_general}, the approximate closed form of the conditional PMF of the GLD can be written as:
\begin{equation}
    \label{eq:spd_grgg_homophily}
        \omegaacf_l(\boldsymbol{x},\boldsymbol{y}) = \frac{\left(\avgdegree 2^\frac{k}{2}\right)^l}{n\sqrt{|l\mat{R}|}}\exp\left(-\frac{1}{2}(\boldsymbol{x}-\boldsymbol{y})^T(l\mat{R})^{-1}(\boldsymbol{x}-\boldsymbol{y})\right),
\end{equation}
where $|\cdot|$ is the matrix determinant and $\avgdegree$ is the mean degree; see Appendix \ref{sec:apdx_grgg}. Here, we can interpret $\omegaacf_l(\boldsymbol{x},\boldsymbol{y})$ as a ``geodesic kernel'' encoding $\prob{\lambda_{\boldsymbol{x}\boldsymbol{y}}=l}$, where a node ``inflates'' its bubble of nearest neighbors, as defined by the scale matrix $\mat{R}$, by a factor of $l$ to form geodesics of length $l$. (For $l=1$, this simply reduces to the connectivity kernel in Eq. \eqref{eq:grgg_nu}.) It then follows that the approximate closed form of the survival function of the GLD is given by Eq. \eqref{eq:spd_general_psi} as $\psiacf_l(\boldsymbol{x},\boldsymbol{y}) = \exp\left(-\sum_{q=1}^l\omegaacf_q(\boldsymbol{x},\boldsymbol{y})\right)$ with, as previously, an analogous interpretation to Eq. \eqref{eq:sf_avg_uncorrected} in terms of independent geodesics. In Fig. \ref{fig:spl_grgg}, we plot various node and node pair statistics for a Gaussian RGG when $n=512, \avgdegree=4$ and the scale matrix is given by  $\mat{R}=\big(\begin{smallmatrix} 0.08 & 0.04\\ 0.04 & 0.08 \end{smallmatrix}\big)$, using Eq. \eqref{eq:grgg_degree} for a node's degree, Eq. \eqref{eq:rho_grgg} for a node's percolation probability, and the approximate closed form of the GLD in Eq. \eqref{eq:grgg_omega_ansatz} for the analytic estimate of expected geodesic lengths between node pairs. We also show in Fig. \ref{fig:spl_emp_vs_ana} that this estimate is in good agreement with the empirics.

\subsection{Sparse graphon}\label{sec:graphons}

In the most general setting, any conditionally-independent-edge model with a symmetric connectivity kernel can be expressed by considering a sequence of graphs in some continuum limit, called graph functions or graphon \cite{lovasz2006graphon, lovasz2012graphon, orbanz2014graphons}. Typically, the node space for a graphon is restricted to the real interval $[0,1]$ where nodes are distributed according to the standard uniform distribution $\mathcal{U}(0,1)$. Then all burden of modeling edge probabilities is transferred to the symmetric kernel $W:[0,1]^2\to[0,1]$, referred to as the ``$W$-graphon''. Given their flexibility, these functions can get arbitrarily complex. While $W$-graphons are usually formulated as the dense limit of a graph sequence with $\bigtheta{n^2}$ edges \cite{lovasz2006graphon}, here we are interested in the sparse limit with $\bigtheta{n}$ edges \cite{caronfox2017sparsegraphons}. In particular, we consider a sparse sequence of graphs where each graph is finitely exchangeable, also referred to as the ``inhomogeneous random graph model'' \cite{bollobas2011sparsegraphs}. For brevity, throughout this article we refer to $W_n:[0,1]^2\to[0,1]$ such that $W_n=\bigtheta{n^{-1}}\textrm{ or }0$ as a ``sparse graphon'', or simply as ``graphon''. The GLD framework of Sec. \ref{sec:general_graphs} translates immediately to sparse graphons, with the simplicity offered by the symmetry of $W_n$, and by assuming $x\sim\mathcal{U}(0,1)$, with regards to numerical integration.

\paragraph*{Illustrative example: max graphon}As an example, consider a sparse version of the ``max graphon'', that arises as the limit of a uniform attachment process \cite{borgs2011maxgraphon, klimm2021modularity} given by $W_n(x,y)=\frac{\beta}{n}(1-\max(x,y))$, where $\beta >0$. In Fig. \ref{fig:spl_maxg} we show various node and node pair functions for this graphon with $n=512$, $\beta=8$, wherein the expected degree, percolation probability, and expected geodesic lengths are all obtained via numerical integration of Eqs. \eqref{eq:general_degrees_deg}, \eqref{eq:gcc_consistency_general} and \eqref{eq:spd_analytic_general_eig_uncorrected} respectively. In Fig. \ref{fig:spl_emp_vs_ana}, we show that the empirical and analytic estimates of the geodesic lengths are in good agreement, except for longer geodesic lengths---likely due to the vanishingly low percolation probabilities of the tail-end of the node space as $x\to 1$. As previously noted, we can discretize any continuous network model at a chosen scale to obtain an equivalent SBM representation of it---see Appendix \ref{sec:apdx_general} for a discussion on discretizing graphons in particular. In Fig. \ref{fig:spl_emp_vs_ana}, we also plot empirics and analytics for the SBM corresponding to max graphons, that corroborate well with results obtained via numerical integration.

\paragraph*{Sparse multiplicative graphons}We next consider a scenario where the formalism simplifies further. Let $f:[0,1]\to[0,1]$ be a function such that $f=\bigtheta{n^{-\frac{1}{2}}}\textrm{ or }0$. Then, we define a sparse multiplicative graphon over a node pair $W^\times_n(x,y)$ to be one which can be written as the product of that function applied to the nodes separately:
\begin{equation}
    \label{eq:def_mult_graphon}
    W^\times_n(x,y)\triangleq f(x)f(y),
\end{equation}
that are equivalent to canonical degree-configuration models since $f(x)\propto\meandegree(x)$ where we use $\meandegree(x)$ to denote the expected degree of a node at $x$; see Appendix \ref{sec:apdx_graphons} and Eq. \eqref{eq:mult_graphon_foo_deg}. We remark that an asymmetric and directed version of this graphon can be obtained by considering two distinct functions $f$ and $g$, but we restrict our discussion here to symmetric multiplicative graphons. To derive their network properties, it will be useful to define the first and second moments of $f(x)$ that can be shown to encode the first and second moments of the degree distribution (see Appendix \ref{sec:apdx_graphons}):
\begin{subequations}
\label{eq:mult_graphon_stats}
\begin{align}
    \label{eq:mult_graphon_stats_zeta}
    \zeta&\triangleq\sqrt{n}\int_0^1f(x)\,dx=\sqrt{\avgdegree},\\
    \label{eq:mult_graphon_stats_eta}
    \eta&\triangleq n\int_0^1f(x)^2\,dx=\frac{\avg{\degree^2}}{\avgdegree}-1.
\end{align}
\end{subequations}
We now consider the percolation probability at $x$: define $\rho\triangleq\sqrt{n}\int_0^1f(x)\rho(x)\,dx$, then we obtain from Eq. \eqref{eq:gcc_consistency_general} for percolation probability and Eq. \eqref{eq:mult_graphon_stats_zeta}:
\begin{subequations}
\label{eq:gcc_consistency_r1g}
\begin{align}
    \label{eq:gcc_consistency_r1g_rhox}
    \begin{split}
        \rho(x)&\approx 1-\exp\left(-n\int_0^1f(x)f(y)\rho(y)\,dy\right)\\&=1-\exp\left(-\sqrt{n}\rho f(x)\right),\\
    \end{split}\\
    \label{eq:gcc_consistency_r1g_rho}
    \begin{split}
        \implies&\int_0^1f(x)\rho(x)dx\approx \int_0^1f(x)\left[1-\exp\left(-\sqrt{n}\rho f(x)\right)\right]\,dx\\
        \implies&\rho\approx\zeta-\sqrt{n}\int_0^1f(x)\exp\left(-\sqrt{n}\rho f(x)\right)\,dx.
    \end{split}
\end{align}
\end{subequations}
Eq. \eqref{eq:gcc_consistency_r1g_rho} is a self-consistent scalar equation for $\rho$ which once solved can be used to solve the self-consistent scalar Eq. \eqref{eq:gcc_consistency_r1g_rhox} for the percolation probability of any node location.

\paragraph*{Approximate closed form of the GLD} Exploiting the multiplicative nature of the connectivity kernel in Eq. \eqref{eq:def_mult_graphon}, we obtain the conditional PMF $\omegaacf_l(x, y)$ and survival function $\psiacf_l(x,y)$ of the approximate closed form GLD from Eqs. \eqref{eq:spd_general_omega}, \eqref{eq:spd_general_psi} as
\begin{subequations}
    \label{eq:spd_r1g_main}
    \begin{align}
    \label{eq:spd_r1g_omega}
    \omegaacf_l(x, y) &= \eta^{l-1}f(x)f(y)=\eta^{l-1}\frac{\meandegree(x)\meandegree(y)}{n\avgdegree},\\
    \label{eq:spd_r1g_psi}
    \psiacf_l(x, y) &= \begin{cases}
    \exp\left(-\frac{\eta^l-1}{\eta-1}\frac{\meandegree(x)\meandegree(y)}{n\avgdegree}\right) &\mbox{if }\eta\ne 1,\\
    \exp\left(-l\frac{\meandegree(x)\meandegree(y)}{n\avgdegree}\right) &\mbox{otherwise,}
    \end{cases}
    \end{align}
\end{subequations}
where we express $f(x)$ in terms of the expected degree function $\meandegree(\cdot)$ using Eq. \eqref{eq:mult_graphon_foo_deg}. From Eq. \eqref{eq:spd_r1g_psi} we note that the distribution of geodesic lengths between two nodes in a degree-configuration model is encoded by the product of their expected degree. We also observe that a larger variance in the degree distribution renders a larger value for $\eta$ from Eq. \eqref{eq:mult_graphon_stats_eta}, and therefore shorter geodesic lengths \cite{vanderhofstad2005distanceconfigmodel}. If we consider the expected survival function of the GLD for the whole network $\psiacf(l)$, as defined in Eq. \eqref{eq:spd_general_psi_agg}, then applying Jensen's inequality \cite{jensen1906fonctions} to Eq. \eqref{eq:spd_r1g_psi} yields a lower bound on $\psiacf(l)$, analogous to the bounds in Eqs. \eqref{eq:spd_general_psi_agg_bound}, \eqref{eq:spd_rdpg_network}:
\begin{equation}
    \label{eq:spd_r1g_network}
    \psiacf(l) \ge \begin{cases}
    \exp\left(-\frac{\avgdegree(\eta^l-1)}{n(\eta-1)}\right) &\mbox{if }\eta\ne 1,\\
    \exp\left(-l\frac{\avgdegree}{n}\right) &\mbox{otherwise,}
    \end{cases}
\end{equation}
where we have used the definition of mean degree in Eq. \eqref{eq:general_degrees_mean}. Together, Eqs. \eqref{eq:spd_r1g_network} and \eqref{eq:mult_graphon_stats_eta} provide a bound for the closed form  of the GLD in a degree-configuration model entirely in terms of the first and second moments of the degree distribution. This bound is tight when the variance in survival function across node pairs is small. For instance, in an ER graph where every node has the same expected degree $\avgdegree$, there is no variance in the survival function across node pairs. Each node has a Poisson degree distribution yielding $\eta=\avgdegree$ from Eq. \eqref{eq:mult_graphon_stats_eta}, which substituted in Eq. \eqref{eq:spd_r1g_network} leads precisely to the expression we previously obtained in Eq. \eqref{eq:spd_er}.

\paragraph*{Illustrative example: random regular graphs} Since the GLD in multiplicative graphons---that impose constraints only on node degrees in an otherwise random graph---depends only on the first two moments of the degree distribution, we can consider an extreme where the degree distribution has zero variance: the example of random $\degree$-regular graphs, wherein every node has the same degree $\degree$, and $\degree\in\integerpos$ such that $n\degree$ is even, but connections are otherwise random between node pairs \cite{bollobas2001random}. Evidently, the degree constraint on random \emph{regular} graphs prohibits conditionally independent edges, whereas the framework used to derive Eq. \eqref{eq:spd_r1g_psi} assumes a conditionally-independent-edge model. Remarkably, since Eq. \eqref{eq:spd_r1g_psi} is based only on degree moments, and because the only constraint imposed in random regular graphs is---much like in multiplicative graphons---on node degrees, we can still obtain a GLD for random regular graphs. Because every node has the same degree $\degree$:
\begin{subequations}
\label{eq:def_random_reg_graph}
\begin{align}
    &\forall x\in V: \meandegree(x)=\avgdegree=\degree,\\
    &\avg{\degree^2}=\degree^2.
\end{align}
\end{subequations}
Then Eqs. \eqref{eq:mult_graphon_stats_eta} and \eqref{eq:def_random_reg_graph} yield $\eta=\degree-1$, and we can rewrite Eq. \eqref{eq:spd_r1g_psi} as:
\begin{equation}
    \label{eq:spd_randomregular_psi}
    \psiacf_l(x, y) = \begin{cases}
    \exp\left(-\frac{\degree[(\degree-1)^l-1]}{n(\degree-2)}\right) &\mbox{if }\degree\ne 2,\\
    \exp\left(-\frac{2l}{n}\right) &\mbox{otherwise.}
    \end{cases}
\end{equation}
It is worth analyzing Eq. \eqref{eq:spd_randomregular_psi} when $\degree=1$: every node in the random regular graph is attached to exactly one other node, thus the network is composed of $n/2$ disconnected edges and does not have a giant component. Picking a node at random, the probability that another random node is directly connected to it is asymptotically $n^{-1}$, and the probability mass at other geodesic lengths is zero, yielding the survival function of the GLD as $1-n^{-1}$ for any $l\ge 1$. This is precisely what we obtain from Eq. \eqref{eq:spd_randomregular_psi} by setting $\degree=1$ and applying a first-order approximation to the exponential. For $\degree=2$, the network is composed of one or more disconnected cycles. Picking a node at random on a cycle of asymptotically large length, there are exactly $2$ nodes at a distance of $l$ from it, on either side. Thus, the probability mass at length $l$ is asymptotically $2n^{-1}$, yielding the survival function of the GLD as $1-2ln^{-1}$. As before, we obtain this expression by applying a first-order approximation to the exponential in Eq. \eqref{eq:spd_randomregular_psi}. In Fig. \ref{fig:spd_randomregular}, we show that Eq. \eqref{eq:spd_randomregular_psi} is an extremely good approximation of the GLD for other (larger) degrees too.

\begin{figure}
    \centering
    \includegraphics[width=\columnwidth]{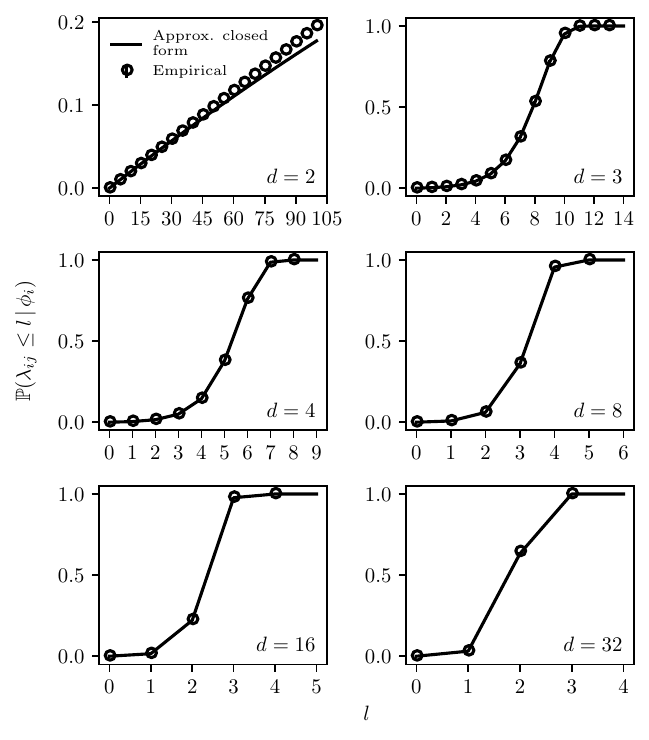}
    \caption{Empirical and approximate closed-form CDF of geodesic lengths for a random $\degree$-regular graph are in good agreement. Network size is fixed at $n=1024$, while degree varies as $\degree\in\{2, 3, 4, 8, 16, 32\}$. Solid line indicates the approximate closed-form solution (Eqs. \eqref{eq:spd_randomregular_psi}). Symbols ($\circ$) and bars indicate empirical estimates: mean and standard error over 10 network samples. (The variation over samples is negligible.) For $\degree=2$, the network is at the phase transition (see Sec. III\,A in Ref. \onlinecite{loomba2024geodesics2}), and symbols are shown at every fifth geodesic length for clarity.}
    \label{fig:spd_randomregular}
\end{figure}

\paragraph*{Illustrative example: scale-free networks} We next consider the other extreme where the degree distribution has potentially unbounded variance. A typical example is of ``scale-free'' networks, whose nodes follow a power law degree distribution. Although scale-free networks are usually generated by a dynamic process like preferential attachment \cite{albert2002networks, barabasi1999bagraph}, here we consider a \emph{static} version based on vertex-fitness that permits a model with conditionally independent edges, similar to the stochastic fitness model in Ref. \onlinecite{caldarelli2002scale}. We define a multiplicative scale-free graphon with $f(x)=\sqrt{\frac{\beta}{n}}\left(\frac{x}{h}\right)^{-\alpha}$, i.e. $W^\times_n(x,y)=\frac{\beta}{n}\left(\frac{xy}{h^2}\right)^{-\alpha}$, for some scalars $0<\alpha\le 1$, $0< h \ll 1$, $\beta>0$, with the node space restricted to the real interval $[h,1]$ i.e. $\mu(x)=\frac{1}{1-h}$ if $x\in[1,h]$ and $0$ otherwise---a minor modification to the usual assumption of a standard uniform distribution. Here, $\alpha$ controls the exponent of the power law governing the degree distribution: using Eq. \eqref{eq:degree_r1g}, the expected degree at $x$ is $\meandegree(x)\propto x^{-\alpha}$ which implies a degree distribution $\degree\propto \degree^{-\theta}$ where $\theta\triangleq 1+\frac{1}{\alpha}$. (The distribution is not a \emph{pure} power law, and might show departures from usually studied scale-free networks in the small $\theta$ regime, as shown in the derivation of the degree distribution in Eq. \eqref{eq:sfg_degree_true} in Appendix \ref{sec:apdx_graphons}.) The definition for $\eta$ in Eq. \eqref{eq:mult_graphon_stats_eta} yields:
\begin{equation}
    \label{eq:sfg_eta}
    \eta =
    \begin{cases}
    \beta\frac{h\log h^{-1}}{1-h} &\mbox{if }\alpha=1/2,\\
    \beta\frac{h^{2\alpha}-h}{(1-h)(1-2\alpha)} &\mbox{otherwise,}
    \end{cases}
\end{equation}
and for $\zeta$ in Eq. \eqref{eq:mult_graphon_stats_zeta} yields:
\begin{equation}
    \label{eq:sfg_gamma}
    \zeta=\begin{cases}
            \sqrt{\beta}\frac{h\log h^{-1}}{1-h} &\mbox{if }\alpha=1,\\
            \sqrt{\beta}\frac{h^{\alpha}-h}{(1-h)(1-\alpha)} &\mbox{otherwise},
        \end{cases}
\end{equation}
which can be inserted into Eqs. \eqref{eq:degree_r1g_overall}, \eqref{eq:spd_r1g_psi} to obtain the approximate closed form of the survival function of the GLD in scale-free graphons. We illustrate with three special cases based on the power law exponent of the degree distribution: (1) an ER graph (where the power law exponent $\theta\to\infty\implies\alpha\to 0$), (2) a BA scale-free graphon (where the power law exponent $\theta=3\implies\alpha=\frac{1}{2}$ \cite{barabasi1999bagraph}), and (3) a ``highly scale-free'' graphon (for which $\theta=2\implies\alpha=1$). We express the degree-controlling parameter $h$ in terms of the mean degree $\avgdegree$ (see Appendix \ref{sec:apdx_graphons}). Assuming $h\ll 1$, from Eqs. \eqref{eq:degree_net_r1g}, \eqref{eq:sfg_eta}, \eqref{eq:sfg_gamma} we obtain: 
\begin{equation}
    \label{eq:sfg_h}
    h \approx 
    \begin{cases}
    \frac{\avgdegree}{4\beta} &\mbox{if }\alpha=\frac{1}{2},\\
    \left(\sqrt{\frac{\beta}{\avgdegree}}\log\sqrt{\frac{\beta}{\avgdegree}}\right)^{-1} &\mbox{if }\alpha=1,
    \end{cases}
\end{equation}
and for the ER graph ($\alpha=0$) Eqs. \eqref{eq:degree_net_r1g}, \eqref{eq:sfg_eta} yield:
\begin{equation}
    \label{eq:sfg_er_degree}
    \avgdegree=\beta.
\end{equation}
Assuming $h\ll 1$, we obtain from Eqs. \eqref{eq:sfg_eta}, \eqref{eq:sfg_h}, \eqref{eq:sfg_er_degree}:
\begin{equation}
    \label{eq:sfg_eta_}
    \eta \approx 
    \begin{cases}
    \avgdegree &\mbox{if }\alpha\to 0,\\
    \frac{\avgdegree}{4}\log\left(\frac{4\beta}{\avgdegree}\right) &\mbox{if }\alpha=\frac{1}{2},\\
    \frac{\sqrt{\beta \avgdegree}}{\log \sqrt{\frac{\beta}{\avgdegree}}} &\mbox{if }\alpha=1.
    \end{cases}
\end{equation}
From Eq. \eqref{eq:sfg_eta_} we note that assuming sparsity from Eq. \eqref{eq:sparsity_constraint_general}---of the form $f(x)=\bigtheta{n^{-\frac{1}{2}}}\implies\beta=\bigtheta{1}$---yields asymptotically bounded values of $\eta$ for the BA and highly scale-free graphons, and thus for the degree's variance, unlike typical scale-free networks. However, we can use large values of $\beta=\bigtheta{n}$ while maintaining finite mean degree $\avgdegree$, which permits an assessment of our formalism when ``sparsity'' is not enforced. In Fig. \ref{fig:spl_sfg} we show various node and node pair functions for the BA scale-free graphon with $n=512=\beta$, $\avgdegree=4$ and $\alpha=1/2$. The expected degree, percolation probability, and expected geodesic lengths are obtained via Eqs. \eqref{eq:degree_r1g}, \eqref{eq:gcc_consistency_r1g} and \eqref{eq:spd_r1g_psi} respectively. In Fig. \ref{fig:spl_emp_vs_ana}, we show that the empirical and analytic estimates of expected geodesic lengths are in good agreement for this scale-free graphon, and its discretised SBM counterpart.

\section{Discussion\label{sec:geo1_discussion}}

In this article, we have derived an analytic geodesic length distribution (GLD) for every node pair in networks---directed or undirected---generated by a sparse statistical network model---symmetric or asymmetric---with conditionally independent edges and no bottlenecks, in the asymptotic limit---what we refer to as a sparse ensemble average network (SEAN) model. The distribution describes geodesic lengths on the giant component when it exists (in the supercritical regime), and on the small components otherwise (in the subcritical regime). The GLD is given by a pair of recursive equations which can be easily solved with initial conditions supplied by the form of the statistical network model. We have obtained a closed-form expression for the survival function of the GLD. In the supercritical regime, the expression is tight for finite lengths in the asymptotic limit, and for shorter lengths in finite-size networks. In the subcritical regime, it is tight for all lengths in asymptotically large networks. The expression provides an approximate closed form of the survival function of the GLD up to length $l$ that resembles the process of hitting a target node $j$ from a source node $i$ via any of the independent geodesics of independent lengths up to $l$, i.e. it is given by an exponential of the negative probability of independent geodesics up to length $l$. This closed form generalizes previous analytic \cite{blondel2007distance, katzav2015analytical} and closed-form approaches \cite{fronczak2004average} to model the GLD in ER graphs and scalar latent variable models, and in some settings can be viewed as the application of Jensen's inequality \cite{jensen1906fonctions} to the closed form which conditions on the limiting random variable associated with the corresponding branching process \cite{barbourreinert2001smallworlds, barbourreinert2006discretesmallworlds}.

\paragraph*{GLD in sparse general random network families} Transitioning away from the ensemble average setting where one has access to the expected adjacency matrix $\avg{\mat{A}}$, we have considered sparse general random network (SGRN) models in some node space $V$ that naturally do not contain any bottlenecks, and generalize inhomogeneous random graphs \cite{bollobas2007phase} to the asymmetric setting permitting interesting behaviors in directed networks. This encompasses a diverse set of models like stochastic block models (SBMs), (soft) random geometric graphs (RGGs), random dot product graphs (RDPGs) and (sparse) graphons. We have derived a closed-form expression for the survival function of the GLD in this general setting, determined by an iterated integral operator $T$ defined over functions on $V$, and the connectivity kernel represents the probability of an edge existing between two nodes in $V$. The operator $T$ is analogous to $\avg{\mat{A}}$ in the SEAN model. For symmetric kernels, this yields an expression for the GLD in terms of the spectral decomposition of $T$. For illustrative examples of each of the above-mentioned models, we have derived the approximate closed form of the GLD revealing novel insights, particularly for higher-dimensional models whose geodesics have not been previously studied analytically. Importantly, we do so without imposing any restrictions on the eigenvalues of the operator---such as periodicity of the kernel or the existence of a giant component (see Theorem 2 in Ref. \onlinecite{loomba2024geodesics2}). We showed for:
\begin{enumerate}
    \item SBMs that the survival function is expressed via an eigendecomposition of the block matrix; Eq. \eqref{eq:spd_sbm_eig}.
    \item RDPGs that the mean vector and covariance matrix of the node distribution can specify a lower bound of the survival function for geodesics between a random node pair; Eq. \eqref{eq:spd_rdpg_network}.
    \item Gaussian RGGs with high spatial homophily that the (conditional) probability mass function at length $l$ can be interpreted as a ``geodesic kernel'' that changes the connection scales by a factor of $l$; Eq. \eqref{eq:spd_grgg_homophily}.
    \item Multiplicative (sparse) graphons where the connectivity kernel of a node pair is a product of node functions, (equivalent to canonical degree-configuration models,) that the product of expected degrees of the source and target nodes specifies the survival function (Eq. \eqref{eq:spd_r1g_psi}), and the survival function of the GLD for a random node pair can be bounded from below using the first and second moments of the degree distribution; Eqs. \eqref{eq:spd_r1g_network}, \eqref{eq:mult_graphon_stats_eta}.
\end{enumerate}

Despite the assumptions involved, we have shown for various models---including ``Gaussian RGG'', ``Dirichlet RDPG'', random $\degree$-regular graphs, and ``scale-free graphons''---that there is good agreement in the approximate closed form and empirical estimates of expected geodesic lengths between node pairs.

\paragraph*{Scope of applications} From an applied perspective, we provide empirical corroboration of our framework by demonstrating how real-world networks, when their structure is known, can be cheaply ``coarsened'' into SBMs to compute their expected GLD analytically (see Fig. \ref{fig:spl_statistics_empirical}, and Fig. 3 in Ref. \onlinecite{loomba2024geodesics2}). For partially observed networks, an appropriate statistical network model, such as an SBM, can be inferred and the GLD can be obtained analytically; see Fig. \ref{fig:uk_map_aspl} where we map geodesics statistics for a large-scale social network at the country level. Our results on RDPG find relevance in the burgeoning field of statistical machine learning on graphs, wherein nodes are typically embedded in a Euclidean space $\real^k$ equipped with the dot product. We have shown that the matrix of second moments in $\real^k$ (or in a corresponding feature space $\real^d$ when the kernel is not linear in the dot product) completely defines the GLD for nodes located at $\boldsymbol{x},\boldsymbol{y}\in \real^k$. This is particularly useful for sequential learning or querying for distances from individual nodes in prohibitively large networks. More generally, our theoretical framework to analytically estimate the GLD can advance our understanding of graph representation learning based on message-passing graph neural networks (GNNs) \cite{rubin2023geodesic} and push the state-of-the-art, since incorporating knowledge of inter-node distances beyond immediate neighbors has been shown to provably improve GNN performance \cite{li2020distanceencodinggnn}.

\paragraph*{Limitations}Our assumptions imply limitations on where our approach is applicable, and suggest extensions to be considered for future work. The assumption of conditionally independent edges might hinder tight local clustering which is found commonly in some networks like social networks \cite{holland1971transitivity, holland1976local}, and sparsity constraints may not hold for broad degree distributions over finite network sizes---both being important assumptions for deriving the recursive equations for the GLD. The approximations involved in obtaining the closed-form expression discount probability mass at longer geodesic lengths, and thus may not hold well for some models, particularly for finite-size network with very large diameters. In practice however, we have observed good agreement between analytics and empirics for the statistics we considered, both for Gaussian RGGs that are highly spatial---and thus exhibit some degree of local clustering---and scale-free graphons---that have heavy-tailed degree distributions.

\begin{acknowledgments}
This work has been supported by EPSRC grant EP/N014529/1. The authors would like to thank Gesine Reinert and Mauricio Barahona for insightful feedback, and Asher Mullokandov, George Cantwell, Florian Klimm, Till Hoffmann and Matthew Garrod for helpful comments on the manuscript.
\end{acknowledgments}

\appendix

\section{Technical lemmas for the GLD}

Appendices \ref{sec:apdx_connect_gcc}, \ref{sec:apdx_lemmas12_bridging}, and \ref{sec:percprob} provide key technical lemmas on connection, bridging, and percolation probabilities, respectively, required to derive the geodesic length distribution (GLD) in sparse ensemble average networks (SEANs; Sec. \ref{sec:spd}). 

\subsection{\label{sec:apdx_connect_gcc}Connection probability on the giant component}

First, we derive connection probabilities for nodes on the giant component. We distinguish between three kinds of nodes: ``percolating'', ``dangling'', and ``non-percolating'' nodes that respectively have a non-vanishing, vanishing but non-zero, and zero probability of being on the giant component---see Appendix \ref{sec:percprob} for their technical definitions and further details.
\paragraph*{Notation} Let $G$ be a network, or graph, with $n$ nodes encoded by the adjacency matrix $\mat{A}$, and $G^{\setminus ij}$ indicate the subgraph with the edge $A_{ij}$ unconsidered/removed. We assume that $G$ is a SEAN. With a minor change in notation from the main text, let $\phi_i(G)$ refer to the event that node $i$ is on the giant component of a graph $G$, and $\prob{\phi_i(G)}$ be the ``percolation probability'' of node $i$ in graph $G$. We use the notation $f(n)=\littleo{g(n)}$ if $\lim_{n\to\infty}\frac{f(n)}{g(n)}=0$, $f(n)=\bigomega{g(n)}$ if $\liminf_{n\to\infty}\frac{f(n)}{g(n)}>0$, $f(n)=\order{g(n)}$ if $\limsup_{n\to\infty}\frac{f(n)}{g(n)}<\infty$, $f(n)=\littleomega{g(n)}$ if $\lim_{n\to\infty}\frac{f(n)}{g(n)}=\infty$, and $f(n)=\bigtheta{g(n)}$ if $f(n)=\bigomega{g(n)}$ and $f(n)=\order{g(n)}$. Below we derive the connection probabilities conditioned on the source node being on the giant component, that provides the initial condition for the GLD.

\begin{lemma}[Connection probability on the giant component]\label{lemma:deg_gcc}
    For nodes $i,j$ in a sparse ensemble average network $G$ that is undirected, if $i$ is a percolating node $(\prob{\phi_i(G)}=\bigomega{1})$ and $j$ is a percolating or non-percolating node $(\prob{\phi_j(G)}=\bigomega{1}$ or $0)$, then asymptotically:
    \begin{equation}\label{eq:probconnectgc}
    \begin{split}
        \condprob{A_{ij}=1}{\phi_i(G)}&\approx \prob{A_{ij}=1}\\&\times\left\lbrace 1+\left[\frac{1}{\prob{\phi_i(G)}}-1\right]\prob{\phi_j(G)}\right\rbrace. 
    \end{split}
    \end{equation}
\end{lemma}
\begin{proof}
    Without loss of generality, generate the graph $G$ such that the edge from node $i$ to $j$ is generated at the final step, which will be independent of anything else (Eq. \eqref{eq:ciem}), and consider the probability 
    \begin{equation*}
        \begin{split}                   \prob{A_{ij}=1,\phi_i(G),\phi_j(G)}&=\prob{A_{ij}=1,\phi_i(G^{\setminus{ij}}),\phi_j(G^{\setminus{ij}})} \\&+ \prob{A_{ij}=1,\phi_i(G^{\setminus{ij}}),\neg\phi_j(G^{\setminus{ij}})}\\
            &+\prob{A_{ij}=1,\neg\phi_i(G^{\setminus{ij}}),\phi_j(G^{\setminus{ij}})}\\
            &=\prob{A_{ij}=1}\\&\times\left[1-\prob{\neg\phi_i(G^{\setminus{ij}}),\neg\phi_j(G^{\setminus{ij}})}\right].
        \end{split}
    \end{equation*}
    From Lemma \ref{lemma:perccontribs}, the percolation events of two percolating nodes in $G$ are asymptotically independent, that is, $\prob{\neg\phi_i(G^{\setminus{ij}}),\neg\phi_j(G^{\setminus{ij}})}\approx\prob{\neg\phi_i(G^{\setminus{ij}})}\prob{\neg\phi_j(G^{\setminus{ij}})}$. Since $i$ and $j$ are, by assumption, non-dangling nodes in $G$, we know that removal of any single node has a vanishing effect on their percolation probabilities (Eq. \eqref{eq:percsubfullgraph}), therefore the removal of a single edge $A_{ij}$ has a vanishing effect on their percolation probabilities: $\prob{\phi_i(G^{\setminus ij})}\approx\prob{\phi_i(G)}$ and $\prob{\phi_j(G^{\setminus ij})}\approx\prob{\phi_j(G)}$. This yields:
    \begin{align*}
        \prob{A_{ij}=1,\phi_i(G),\phi_j(G)}\approx&\ \prob{A_{ij}=1}\\&\times\left[1-\prob{\neg\phi_i(G)}\prob{\neg\phi_j(G)}\right].
    \end{align*}
    Also, we can write
    \begin{equation*}
    \begin{split}        
        \prob{A_{ij}=1,\phi_i(G),\phi_j(G)}&=\condprob{\phi_j(G)}{A_{ij}=1,\phi_i(G)}\\&\quad\times\condprob{A_{ij}=1}{\phi_i(G)}\prob{\phi_i(G)}\\&=\condprob{A_{ij}=1}{\phi_i(G)}\prob{\phi_i(G)},
    \end{split}
    \end{equation*}
    since $(A_{ij}=1)\cap\phi_i(G)\implies\phi_j(G)$, i.e. $\condprob{\phi_j(G)}{A_{ij}=1,\phi_i(G)}=1$.
    Substituting the value derived for $\prob{A_{ij}=1,\phi_i(G),\phi_j(G)}$ from above then gives us the desired expression on the RHS of Eq. \eqref{eq:probconnectgc}.
\end{proof}

\begin{corollary}[Sparse connection probability of percolating nodes]\label{cor:sparsepercolating}
    For nodes $i,j$ in a sparse ensemble average network $G$ that is undirected, if $i$ is a percolating node $(\prob{\phi_i(G)}=\bigomega{1})$ and $j$ is a percolating or non-percolating node $)\prob{\phi_j(G)}=\bigomega{1}$ or $0)$, then asymptotically: $$\condprob{A_{ij}=1}{\phi_i(G)}=\bigtheta{n^{-1}}\textrm{ or }0.$$
\end{corollary}
\begin{proof}
    This follows from Lemma \ref{lemma:deg_gcc} and sparsity (Eq. \eqref{eq:sparsity_constraint}).
\end{proof}
From Sec. \ref{sec:spd}, the initial condition for the recursive equations to derive the GLD between node $i,j$ in an undirected setting is given by $\condprob{A_{ij}=1}{\phi_i(G)}$. Here, we have shown how this conditional probability of an edge relates to its marginal (i.e. unconditional) probability. To fix the initial condition we need to compute percolation probability $\prob{\phi_i(G)}$: the probability of being on the giant component for every node $i$ of the network. However, under some circumstances, we need not compute $\prob{\phi_i(G)}$. In particular, from Lemma \ref{lemma:deg_gcc}, it is evident that in the case where (1) $i$ is very likely to percolate i.e. $\prob{\phi_i(G)}\to 1$, or when (2) $j$ is very likely to \emph{not} percolate i.e. $\prob{\phi_j(G)}\to 0$, then we have $\condprob{A_{ij}=1}{\phi_i(G)}\to \prob{A_{ij}=1}$, i.e. the conditional and marginal probabilities coincide. When (3) both $i$ and $j$ are very likely to not percolate then $\condprob{A_{ij}=1}{\phi_i(G)}\to \prob{A_{ij}=1}\left[1+\frac{\prob{\phi_j(G)}}{\prob{\phi_i(G)}}\right]$. If the ratio of percolation probabilities is known, then the initial condition is also known. For instance, in an ER graph where every node is equivalent, we have in the just-supercritical regime: $\prob{\phi_i(G)}=\prob{\phi_j(G)}\to 0$, and obtain $\condprob{A_{ij}=1}{\phi_i(G)}\to 2\prob{A_{ij}=1}$. Thus, we need to compute the probability of percolating to compute the initial condition in intermediate settings of percolation. We remark that we can also derive the probability of an edge when the node is \emph{not} percolating: $\condprob{A_{ij}=1}{\neg\phi_i(G)} = \frac{\prob{A_{ij}=1,\neg\phi_i(G)}}{\prob{\neg\phi_i(G)}}=\frac{\prob{A_{ij}=1}-\prob{A_{ij}=1,\phi_i(G)}}{1-\prob{\phi_i(G)}}.$ Then using Lemma \ref{lemma:deg_gcc} we get:
\begin{equation}\label{eq:deg_nongcc}
    \condprob{A_{ij}=1}{\neg\phi_i(G)} \approx \prob{A_{ij}=1}[1-\prob{\phi_j(G)}].
\end{equation}
In Fig. \ref{fig:degree_gcc_nongcc}, we show how the expected degree varies for percolating and non-percolating nodes, for connectives above the percolation threshold of mean degree $\avgdegree=1$ \cite{erdos1960evolution}.
\begin{figure}[h]
    \centering
    \includegraphics[width=\columnwidth]{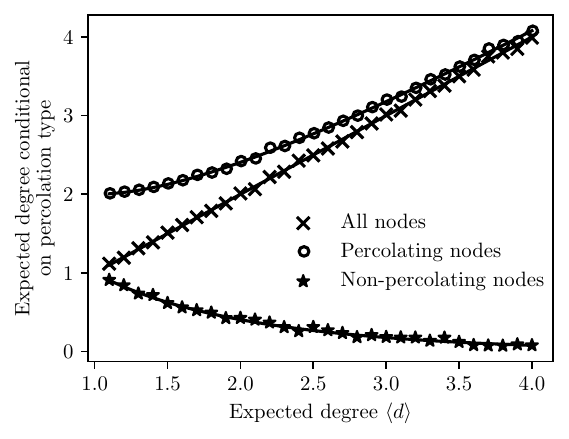}
    \caption{Empirical expected degree of nodes---in the whole network ($\times$), on the giant component ($\circ$), and \emph{not} on the giant component ($\star$)---for ER graphs with varying mean connectivity, agree with the analytic estimates. Network size is fixed at $n=1024$, while mean degree varies from $\avgdegree\in[1.1, 4]$. Lines indicate analytics corresponding to the probability of connection between two nodes $i,j$---$\prob{A_{ij}=1}$, $\condprob{A_{ij}=1}{\phi_i}$ from Lemma \ref{lemma:deg_gcc} and $\condprob{A_{ij}=1}{\neg\phi_i}$ from Eq. \eqref{eq:deg_nongcc}, respectively---scaled by $n$ for interpretation as mean degree, while symbols ($\times, \circ, \star$) indicate empirical mean estimates over 10 network samples---standard errors are small and not shown to aid legibility. For graphs with high connectivity $\condprob{A_{ij}=1}{\phi_i}\to \prob{A_{ij}=1}$, while for those with low connectivity $\condprob{A_{ij}=1}{\phi_i}\to 2\prob{A_{ij}=1}$.}
    \label{fig:degree_gcc_nongcc}
\end{figure}

\subsection{\label{sec:apdx_lemmas12_bridging}Bridging probability}

In this section we develop technical lemmas used to derive the GLD in sparse ensemble average networks.

\paragraph*{Notation} With a minor change in notation from the main text, let $\lambda_{ij}(G)\in\integernonneg$ refer to the length of a shortest path, or geodesic, from node $i$ to node $j$ in a network $G$. Let $\eta_i(G)$ denote the set of immediate neighbors of $i$ in $G$, and $\eta_i(G;l)$ denote the set of order-$l$ neighbors of $i$: nodes for which a shortest path from $i$ to them is of length $l$. (That is, $\eta_i(G)=\eta_i(G;1)$.) A node $k$ is said to be a bridging node of (finite) order $l$ for node pair $(i,j)$ in network $G$ ($i\ne j\ne k$) if a shortest path from $i$ to $j$ in $G$ is of length $l+1$ and $k$ lies on such a shortest path. We now study the asymptotic scaling of the probability of a node to exist on a shortest path of a given node pair, and of the length of a shortest path.

\begin{lemma}[Vanishing bridging probability]\label{lemma:vanishingbridge}
    Let $i,j,k$ be distinct nodes in a sparse ensemble average network $G$, and $\lambda_{ij}(G)$ be the length of a shortest path between $i,j$ in $G$. Then asymptotically, (a) the probability that $k$ is a bridging node of order $l$ is either $0$ or scales as $\bigtheta{n^{-2}}$, and (b) the probability that the shortest path from $i$ to $j$ is of length $l$ is either $0$ or scales as $\bigtheta{n^{-1}}$.
\end{lemma}
\begin{proof}
    A sparse ensemble average network (SEAN) has no dangling nodes (Corollary \ref{lemma:dangbottle}), i.e. all nodes are either percolating or non-percolating. We wish to consider the probability of $k$ being on a shortest path from $i$ to $j$. Without loss of generality, consider generating the graph's (conditionally) independent edges (Eq. \eqref{eq:ciem}) in a particular order, as follows. 
    
    First, an edge from $i$ to $j$ is considered, which due to sparsity exists with probability $0$ or $\bigtheta{n^{-1}}$, regardless of whether $i$ is non-percolating (Eq. \eqref{eq:sparsity_constraint}) or percolating (Corollary \ref{cor:sparsepercolating}). Next, consider the edge from $i$ to $k$, which again---due to sparsity---can exist with probability $0$ or $\bigtheta{n^{-1}}$ (Eq. \eqref{eq:sparsity_constraint}, Corollary \ref{cor:sparsepercolating}). Next, consider the edge from $k$ to $j$, for which a similar argument follows. Node $k$ can be a bridging node of order $1$ for $(i,j)$ if $i$ does not connect to $j$ directly, with probability $1-\bigtheta{n^{-1}}\sim 1$ from above, and $i$ connects to $k$ and $k$ connects to $j$, each of which have a probability of $0$ or $\bigtheta{n^{-1}}$. Due to the (conditionally) independent edge generation process considered here, we obtain the probability that $k$ is a bridging node of order $1$ for $(i,j)$ with probability either $0$ or $\bigtheta{n^{-2}}$.
    
    Similarly for another node $u$, consider generating the edge from $i$ to $u$ and from $u$ to $j$, which yields the probability for $u$ to be a bridging node of order $1$ for $(i,j)$ as either $0$ or $\bigtheta{n^{-2}}$. This process can be repeated for all such nodes in $S$ such that $(i,j)$ are potentially connected by $S$. Since $G$ has no bottlenecks (Eq. \eqref{eq:nonultrabottlenecked}): $|S|=0$ or $\bigtheta{n}$. Since a shortest path from $i$ to $j$ of length $2$ can exist (independently) via any of these $\bigtheta{n}$ nodes in $S$, the probability that the shortest path from $i$ to $j$ is of length $2$ is either $0$ or $\bigtheta{n^{-1}}$.

    Next, consider the (conditionally independent) edge from any neighbor of $i$ in $G$, i.e. from $u\in\eta_i(G)$, to node $k$ which exists with a probability of $0$ or $\bigtheta{n^{-1}}$ (Eq. \ref{eq:sparsity_constraint}, Corollary \ref{cor:sparsepercolating}). Due to sparsity (Eq. \ref{eq:sparsity_constraint}, Corollary \ref{cor:sparsepercolating}), $|\eta_i(G)|=0$ or $\bigtheta{1}$, i.e. the probability that there is an edge from (at least) one of the neighbors of $i$ to $k$ is either $0$ or $\bigtheta{n^{-1}}$. Node $k$ can be a bridging node of order $2$ for $(i,j)$ if $i$ does not have a shortest path to $j$ of length $2$---with probability $1-\bigtheta{n^{-1}}\sim 1$ (from above)---and (at least) one of the neighbors of $i$ has an edge to $k$---with probability $0$ or $\bigtheta{n^{-1}}$---and $k$ has an edge to $j$---with probability $0$ or $\bigtheta{n^{-1}}$ (from above). Due to the (conditionally) independent edge generation process considered here, we obtain the probability that $k$ is a bridging node of order $2$ for $(i,j)$ with probability either $0$ or $\bigtheta{n^{-2}}$. 

    This process can be repeated for all nodes $u$ that are not in $\eta_i(G)\cup\{j\}$ but are such that they have a positive connection probability to $j$ and (at least some) nodes in $\eta_i(G)$ have a positive connection probability to them. Since $G$ is not bottlenecked, the set of nodes potentially connecting a neighbor of $i$ with $j$, would be either of size $0$ or $\bigtheta{n}$. Since a shortest path from $i$ to $j$ of length $3$ can exist (independently) via any of these $\bigtheta{n}$ nodes, the probability that the shortest path from $i$ to $j$ is of length $2$ is either $0$ or $\bigtheta{n^{-1}}$.

    Due to sparsity (Eq. \eqref{eq:sparsity_constraint}, Corollary \ref{cor:sparsepercolating}) similar arguments follow for bridges of higher order since the size of neighborhoods of any finite order $l$ is asymptotically bounded: $|\eta_i(G;l)|=0$ or $\bigtheta{1}$.
\end{proof}

\begin{corollary}[Vanishing probability of shared bridging nodes]\label{cor:vanishingshare}
    Let $i\ne j$ and $u\ne v$ be nodes in a sparse ensemble average network $G$ such that it is not the case that both $i=u$ and $j=v$. If $i\ne u$ and $j\ne v$, then let $k$ be a node in $G$ distinct from $u$ and $v$---but possibly identical to $i$ or $j$---that lies on a shortest path from $i$ to $j$. If either $i=u$ or $j=v$ then let $k$ be a node in $G$ distinct from $i,j,u,v$ that lies on a shortest path from $i$ to $j$. Then asymptotically, the probability that $k$ is a bridging node for $(u,v)$---of any finite order $l$---is either $0$ or scales as $\bigtheta{n^{-2}}$.
\end{corollary}
\begin{proof}
    This follows from applying Lemma \ref{lemma:vanishingbridge} to any node $k$ acting as a bridging node of any finite order $l$ for $(u,v)$. (Note: If the graph is undirected we also exclude the condition that both $i=v$ and $j=u$.)
\end{proof}
Next, we consider a lemma that allows us to derive the set of recursive equations for the geodesic length distribution (GLD) in SEANs.

\begin{lemma}[First-order bridging probability]\label{lemma:1}
    Let $i,j,k$ be distinct nodes in a sparse ensemble average network $G$, $\phi_i(G)$ be the event that $i$ is on the giant component of $G$ (and $\prob{\phi_i(G)}=\bigomega{1}$), and $\lambda_{ij}(G)$ be the length of a shortest path between $i,j$ in $G$. Then asymptotically,
    \begin{equation}\label{eq:lemma1_apdx}
    \begin{split}
        &\condprob{\lambda_{ik}(G)=l-1,\lambda_{kj}(G)=1}{\phi_i(G),\lambda_{ij}(G)\ge l} \approx \\&\condprob{\lambda_{ik}(G)=l-1}{\phi_i(G)}\prob{A_{kj}=1}.
    \end{split}
    \end{equation}
\end{lemma}
\begin{proof}
    Write the LHS of Eq. \eqref{eq:lemma1_apdx} as:
    \begin{equation*}
        \begin{split}
            &\condprob{\lambda_{ik}(G)=l-1,\lambda_{kj}(G)=1}{\phi_i(G),\lambda_{ij}(G)\ge l} =\\& \condprob{\lambda_{kj}(G)=1}{\phi_i(G),\lambda_{ij}(G)\ge l,\lambda_{ik}(G)=l-1}\\&\times\condprob{\lambda_{ik}(G)=l-1}{\phi_i(G),\lambda_{ij}(G)\ge l}.
        \end{split}
    \end{equation*}
    Consider the second factor. If $\lambda_{ij}(G)=\infty$ then together with $\phi_i(G)$ it implies that $\neg\phi_j(G)$. Consequently, there's no way for $j$ to influence the shortest paths from $i$ to $k$ on the giant component allowing us to write the second factor as $\condprob{\lambda_{ik}(G)=l-1}{\phi_i(G)}$. If $l\le\lambda_{ij}(G)<\infty$ then there is a shortest path from $i$ to $j$ of finite length no smaller than $l$, which may share some nodes with a shortest path from $i$ to $k$ (on the giant component), and therefore influence $\lambda_{ik}(G)$ being $l-1$. However, from Lemma \ref{lemma:vanishingbridge} and Corollary \ref{cor:vanishingshare} for SEANs, we know that any given node on a shortest path from $i$ to $j$ will have a vanishing contribution of $0$ or $\bigtheta{n^{-2}}$ to the probability of shortest path from $i$ to $k$ being a certain length ($l-1$) which is either $0$ or of order $\bigtheta{n^{-1}}$. That is, in this case too the influence of $\lambda_{ij}(G)$ can be ignored, yielding 
    \begin{align*}
        &\condprob{\lambda_{ik}(G)=l-1}{\phi_i(G),\lambda_{ij}(G)\ge l}\approx\\&\condprob{\lambda_{ik}(G)=l-1}{\phi_i(G)}.
    \end{align*}
    Consider the first factor. Since $i$ is a percolating node the removal of a single edge from $i$ to $j$ would have a vanishing impact of order $\littleo{1}$ on its percolation probability (Table \ref{tab:perccontribs}) and therefore asymptotically $\phi_i(G)\equiv\phi_i(G^{\setminus ij})$. Since the length of a shortest path from $i$ to $j$ is strictly greater than the length of a shortest path from $i$ to $k$, the edge from $k$ to $j$ can neither contribute to the shortest path from $i$ to $k$ nor can it lower the shortest path length for $i$ to $j$, that is $\lambda_{ij}(G)\ge l\cap\lambda_{ik}(G)=l-1\equiv \lambda_{ij}(G^{\setminus ij})\ge l\cap\lambda_{ik}(G^{\setminus ij})=l-1$. This asymptotically yields the first factor as:
    \begin{align*}
    &\condprob{\lambda_{kj}(G)=1}{\phi_i(G^{\setminus ij}),\lambda_{ij}(G^{\setminus ij})\ge l,\lambda_{ik}(G^{\setminus ij})=l-1}\\&=\prob{A_{kj}=1}.
    \end{align*}
    We therefore obtain the RHS of Eq. \eqref{eq:lemma1_apdx}.    
\end{proof}
We emphasize that the connection probability from $k$ to $j$ appears on the RHS of Eq. \eqref{eq:lemma1_apdx} unconditioned on whether either node is on the giant component. One way to understand this informally is to first consider an expression much simpler to analyze than the LHS of Eq. \eqref{eq:lemma1_apdx}: 
\begin{align*}
&\condprob{\lambda_{ik}(G)=l-1,\lambda_{kj}(G)=1}{\phi_i(G)}=\\&\condprob{\lambda_{kj}(G)=1}{\lambda_{ik}(G)=l-1,\phi_i(G)}\\&\times\condprob{\lambda_{ik}(G)=l-1}{\phi_i(G)}.
\end{align*}
Consider the second factor on the RHS: knowing that $k$ is at a distance $l-1$ from $i$, and $\phi_i(G)$, and $l<\infty$, implies $\phi_k(G)$. That is, here we obtain an asymptotic expression which \emph{does} condition on $\phi_k(G)$: 
\begin{align*}
&\condprob{\lambda_{ik}(G)=l-1,\lambda_{kj}(G)=1}{\phi_i(G)}\approx\\&\condprob{\lambda_{ik}(G)=l-1}{\phi_i(G)}\condprob{A_{kj}=1}{\phi_k(G)}.
\end{align*}
Now, consider the expression on the LHS of Eq. \eqref{eq:lemma1_apdx}:
\begin{equation*}
    \begin{split}
        &\condprob{\lambda_{ik}(G)=l-1,\lambda_{kj}(G)=1}{\lambda_{ij}(G)\ge l,\phi_i(G)}=\\&\condprob{\lambda_{kj}(G)=1}{\lambda_{ik}(G)=l-1,\lambda_{ij}(G)\ge l,\phi_i(G)}\\&\times\condprob{\lambda_{ik}(G)=l-1}{\lambda_{ij}(G)\ge l,\phi_i(G)}.
    \end{split}
\end{equation*}
Asymptotically, in a sparse network with no bottlenecks the shortest path between $i$ and $j$ would have a vanishing correlation to the shortest path between $i$ and $k$ (Corollary \ref{cor:vanishingshare}), that is, for the second factor on the RHS we obtain $\condprob{\lambda_{ik}(G)=l-1}{\lambda_{ij}(G)\ge l,\phi_i}\approx\condprob{\lambda_{ik}(G)=l-1}{\phi_i(G)}$. The first factor here is similar to the first factor we obtained above, except with the additional conditioning on $\lambda_{ij}(G)\ge l$. The fact that $j$ is further than distance $l-1$ from $i$, alongside $\lambda_{ik}(G)=l-1$ and $\phi_i(G)$ tells us something important: namely that \emph{even in the absence} of $j$, $k$ will be on the giant component. Since $\condprob{A_{kj}=1}{\phi_k(G)}$ differs from $\prob{A_{kj}=1}$ only insofar as $k$ can be on the giant component \emph{because of} $j$---from Eq. \eqref{eq:prob_connect_exact} $\prob{\phi_j(G)}\to 0$ would give $\condprob{A_{kj}=1}{\phi_k(G)}\to \prob{A_{kj}=1}$---but here we know for a fact that $k$ does not use $j$ to be on the giant component, the appropriate form is the unconditional $\prob{A_{kj}=1}$ instead of $\condprob{A_{kj}=1}{\phi_k(G)}$, yielding the RHS of Eq. \eqref{eq:lemma1_apdx}. This intuition is formalized in a supplementary proof for Lemma \ref{lemma:1} provided in Ref. \onlinecite{loomba2024geodesics2}---which helps produce a result on higher-order bridging probabilities in Lemma 2 of Ref. \onlinecite{loomba2024geodesics2}.

\subsection{\label{sec:percprob}Percolation probability}

\paragraph*{Notation} Let $G$ be a network, or graph, with $n$ nodes and $G^{\setminus i}$ indicate the subgraph with the node $i$ unconsidered/removed. With a minor change in notation from the main text, let $\phi_i(G)$ refer to the event that node $i$ is on the giant component of $G$, and $\prob{\phi_i(G)}$ be the ``percolation probability'' of node $i$ in graph $G$ and define $\prob{\phi_i(G^{\setminus i})}\triangleq 0$. Let
\begin{equation}
    \label{eq:nonzeroprob}
    S_i^{>0}\triangleq\{j\in[n]: \prob{A_{ij}=1}>0\}
\end{equation}
be the set of nodes that have a strictly positive connection probability from $i$. A node $i$ is said to be ``percolating'' in $G$ if it has a non-vanishing probability of being on the giant component: $\prob{\phi_i(G)}=\bigomega{1}$. Correspondingly, 
    \begin{equation}
        \label{eq:percnodes}
        S_\mathrm{p}(G)\triangleq\{j\in[n]: \prob{\phi_j(G)} = \bigomega{1}\}
    \end{equation}
refers to the set of percolating nodes in $G$. A node $i$ is said to be ``non-percolating'' in $G$ if it has zero probability of being on the giant component: $\prob{\phi_i(G)}=0$. Correspondingly, 
    \begin{equation}
        \label{eq:nonpercnodes}
        S_\mathrm{n}(G)\triangleq\{j\in[n]: \prob{\phi_j(G)} = 0\}
    \end{equation}
refers to the set of non-percolating nodes in $G$. A node $i$ is said to be ``dangling'' in $G$ if it has a strictly positive but vanishing probability of being on the giant component: $0<\prob{\phi_i(G)}=\littleo{1}$. Correspondingly, 
    \begin{equation}
        \label{eq:danglingnodes}
        S_\mathrm{d}(G)\triangleq\{j\in[n]: 0<\prob{\phi_j(G)} = \littleo{1}\}
    \end{equation}
refers to the set of dangling nodes in $G$. Evidently, any node in $G$ can only be \emph{one} of the three percolation types, and the cardinality of each of these node sets will depend on $n$ in a manner determined by how the graph sequence grows. In this section we refer to $G$ as a ``sparse network'' if it satisfies all assumptions of a sparse ensemble average network (SEAN), but we relax the assumption of ``no bottlenecks'' from Eq. \eqref{eq:nonultrabottlenecked}. We now analyze how nodes of different percolation types contribute to the percolation behavior of any given node $i$ in a sparse network $G$.

\begin{lemma}[Node contributions to percolation]\label{lemma:perccontribs}
    For node $i$ in a sparse network $G$, the contribution of other nodes to the percolation probability of $i$ is given by Table \ref{tab:perccontribs}.
    \begin{table}[ht]
    \centering
    \begin{tabular}{cc|c|c|c|c}
        &\multirow{3}{*}{$S$} & \multicolumn{3}{|c}{Cardinality of $S\cap S_i^{>0}$}\\\cline{3-6}
        && (1) & (2) & (3) & (4) \\
        && $\bigomega{n}$ & $\littleo{n}$ & $\bigtheta{1}$ & $0$\\\hline
       (a) & $S_\mathrm{p}(G^{\setminus i})$ & $\bigomega{1}$ & $\littleo{1}$ & $\littleo{1}$ & $0$ \\
        (b) & $S_\mathrm{d}(G^{\setminus i})$ & $\littleo{1}$ & $\littleo{1}$ & $\littleo{n^{-1}}$ & $0$\\
        (c) & $S_\mathrm{n}(G^{\setminus i})$ &  $0$ & $0$ & $0$ & $0$
    \end{tabular}
    \caption{Contributions of sets of (a) percolating, (b) dangling, and (c) non-percolating nodes in the subgraph $G^{\setminus i}$---that also have a non-zero probability of connection from node $i$---to the percolation probability of $i$ in a sparse network $G$, depending on how their cardinalities scale (columns (1)--(4)). For example, a non-vanishing number [$\bigomega{n}$; column (1)] of percolating nodes [$S_\mathrm{p}(G^{\setminus i})$; row (a)] contribute non-vanishingly [$\bigomega{1}$; (1a)], while a non-vanishing number of dangling nodes [$S_\mathrm{d}(G^{\setminus i})$; row (b)] can only contribute vanishingly [$\littleo{1}$; (1b)] to the percolation of node $i$. (Here, we differentiate vanishing and zero contributions, i.e. we use $\littleo{\cdot}$ to denote vanishing but strictly positive contributions.) Uses Eqs. \eqref{eq:sparsity_constraint}, \eqref{eq:nonzeroprob}, \eqref{eq:percnodes}, \eqref{eq:nonpercnodes}, \eqref{eq:danglingnodes}, \eqref{eq:percunionbound}, and \eqref{eq:percnodesdependence2}.}
    \label{tab:perccontribs}
\end{table}
\end{lemma}
\begin{proof}
    Without loss of generality, generate $G$ such that the edges of node $i$ are generated at the final step, which will be independent of anything else (Eq. \eqref{eq:ciem}). Then $i$ can be on the giant component of $G$ if and only if it (independently) connects to at least one node with a positive probability of being on the giant component of $G^{\setminus i}$:
     $$\prob{\phi_i(G)}=\prob{\bigcup_{j\in[n]} (A_{ij}=1) \cap \phi_j(G^{\setminus i})}.$$
     Partitioning the node set into the set of percolating ($S_\mathrm{p}(G^{\setminus i})$; Eq. \eqref{eq:percnodes}), dangling ($S_\mathrm{d}(G^{\setminus i})$; Eq. \eqref{eq:danglingnodes}) and non-percolating ($S_\mathrm{n}(G^{\setminus i})$; Eq. \eqref{eq:nonpercnodes}) nodes in $G^{\setminus i}$, the above equation yields:     
     \begin{widetext}
     \begin{equation}\label{eq:percunionbound}
         \begin{split}
             \prob{\phi_i(G)}&=\prob{\left[\bigcup_{j\in S_\mathrm{p}(G^{\setminus i})} (A_{ij}=1) \cap \phi_j(G^{\setminus i})\right] \cup \left[\bigcup_{j\in S_\mathrm{d}(G^{\setminus i})} (A_{ij}=1) \cap \phi_j(G^{\setminus i})\right] \cup \left[\bigcup_{j\in S_\mathrm{n}(G^{\setminus i})} (A_{ij}=1) \cap \phi_j(G^{\setminus i})\right]}\\
             &\le \sum_{j\in S_\mathrm{p}(G^{\setminus i})}\prob{(A_{ij}=1) \cap \phi_j(G^{\setminus i})} + \sum_{j\in S_\mathrm{d}(G^{\setminus i})}\prob{ (A_{ij}=1) \cap \phi_j(G^{\setminus i})} + \sum_{j\in S_\mathrm{n}(G^{\setminus i})}\prob{ (A_{ij}=1) \cap \phi_j(G^{\setminus i})}\\
             &= \sum_{j\in S_\mathrm{p}(G^{\setminus i})}\prob{A_{ij}=1}\prob{\phi_j(G^{\setminus i})} + \sum_{j\in S_\mathrm{d}(G^{\setminus i})}\prob{A_{ij}=1}\prob{\phi_j(G^{\setminus i})},
         \end{split}
     \end{equation}
     \end{widetext}
     where the inequality is due to a union bound, the independence $\inde{(A_{ij}=1)}{\phi_j(G^{\setminus i})}$ is due to Eq. \eqref{eq:ciem}, and non-percolating nodes contribute $0$ due to Eq. \eqref{eq:nonpercnodes} yielding row (c) of Table \ref{tab:perccontribs}. From the definition of a dangling node, inserting Eq. \eqref{eq:danglingnodes}---alongside Eqs. \eqref{eq:sparsity_constraint} and \eqref{eq:nonzeroprob}---in the final expression of Eq. \eqref{eq:percunionbound} suggests that a set of $\bigtheta{1}$ dangling nodes can contribute no more than $\littleo{n^{-1}}$ and a set of $\littleo{n}$ or $\bigomega{n}$ dangling nodes can contribute no more than $\littleo{1}$ to the percolation probability of node $i$, yielding row (b) of Table \ref{tab:perccontribs}. Similarly, from the definition of a percolating node, inserting Eq. \eqref{eq:percnodes}---alongside Eqs. \eqref{eq:sparsity_constraint} and \eqref{eq:nonzeroprob}---in Eq. \eqref{eq:percunionbound} suggests that a set of $\bigtheta{1}$ or $\littleo{n}$ percolating nodes can contribute no more than $\littleo{1}$ to the percolation probability of node $i$, yielding columns (2) and (3) of row (a) of Table \ref{tab:perccontribs}.

     Consider the ansatz that node $i$ is a percolating node. Then it must be the case that $|S_\mathrm{p}(G^{\setminus i})|=\bigomega{n}$, for if that were not true then the first term on the RHS of Eq. \eqref{eq:percunionbound} would scale as $\littleo{1}$ implying that the RHS of Eq. \eqref{eq:percunionbound} scales as $\littleo{1}$, that is, $i$ is not a percolating node in $G$ which contradicts the ansatz (Eq. \eqref{eq:percnodes}). It then follows that the percolation probability of a percolating node $i$ in $G$ asymptotically depends only on a $\bigomega{n}$-sized set of percolating nodes in $G^{\setminus i}$:
     \begin{subequations}
         \begin{align}
            \label{eq:percnodesdependence}
            \prob{\phi_i(G)}&\approx \prob{\bigcup_{j\in S_\mathrm{p}(G^{\setminus i})} (A_{ij}=1) \cap \phi_j(G^{\setminus i})},\\
            \label{eq:percnodessize}
            |S_\mathrm{p}(G^{\setminus i})|&=\bigomega{n}.
         \end{align}
     \end{subequations}
    As shown above, and noted in Table \ref{tab:perccontribs} [column (3)], a single node can contribute no more than $\littleo{1}$ to the percolation probability of another node. Consequently, the addition of $i$ back to $G^{\setminus i}$ may only change the percolation type of a non-percolating node to that of a dangling node. Table \ref{tab:perccontribs} [column (3)] further shows that a \emph{single} node can contribute non-vanishingly to the percolation probability of another node only when the former is percolating and the latter is either dangling or non-percolating---since in that case the contribution is $\littleo{1}$ and strictly positive, and the prior percolation probabilities are $\littleo{1}$. Consequently, for all percolation types \emph{except} potentially when $i$ is percolating and $j$ is dangling or non-percolating:
    \begin{equation}
        \label{eq:percsubfullgraph}
        \prob{\phi_j(G^{\setminus i})}\approx\prob{\phi_j(G)}.
    \end{equation}
    Since all nodes $j$ on the RHS of Eq. \eqref{eq:percnodesdependence} are percolating, Eq. \eqref{eq:percsubfullgraph} yields for the percolating node $i$:
    \begin{equation}
        \label{eq:perprobpercolating}
         \prob{\phi_i(G)}\approx \prob{\bigcup_{j\in S_\mathrm{p}(G^{\setminus i})} (A_{ij}=1) \cap \phi_j(G)}.
    \end{equation}
     Again, due to the result in Table \ref{tab:perccontribs} [column (3)], a single node can contribute no more than $\littleo{1}$ to the percolation probability of another node, and since the percolation probability of a percolating node is of the order $\bigomega{1}$ (Eq. \eqref{eq:percnodes}) this implies an asymptotically vanishing dependence between the percolation events for percolating nodes. This, alongside Eq. \eqref{eq:perprobpercolating}, permits us to rewrite Eq. \eqref{eq:percnodesdependence}:
     \begin{equation}
         \label{eq:percnodesdependence2}
         \begin{split}
            &\prob{\phi_i(G)}\approx \prob{\neg\bigcap_{j\in S_\mathrm{p}(G)} \neg\left[(A_{ij}= 1) \cap \phi_j(G)\right]}\\&\approx 1 - \prod_{j\in S_\mathrm{p}(G)}\left[1-\prob{A_{ij}=1}\prob{\phi_j(G)}\right]\\
            &=1-\exp\left(\sum_{j\in S_\mathrm{p}(G)}\log\left(1-\prob{A_{ij}=1}\prob{\phi_j(G)}\right)\right)\\&\approx 1- \exp\left(-\sum_{j\in S_\mathrm{p}(G)}\prob{A_{ij}=1}\prob{\phi_j(G)}\right),
         \end{split}
     \end{equation}
     where in the second approximation we use the independence edge assumption (Eq. \eqref{eq:ciem}) and in the last step we apply a first-order approximation due to sparsity (Eq. \eqref{eq:sparsity_constraint}). Eqs. \eqref{eq:percnodessize} and \eqref{eq:percnodes} indicate that the RHS of Eq. \eqref{eq:percnodesdependence2} is of the order $\bigomega{1}$, which proves the ansatz and yields column (1) of row (a) of Table \ref{tab:perccontribs}.

     Finally, we remark that if node $i$ is a dangling node in $G$ such that $0<|S_\mathrm{p}(G^{\setminus i})|=\littleo{n}$ and $|S_\mathrm{d}(G^{\setminus i})|=\order{1}$, then Table \ref{tab:perccontribs} [row (a) columns (2, 3) and row (b) columns (3, 4)] suggests that the contribution of dangling nodes in $G$ to the percolation probability of $i$ is vanishing in comparison to that of percolating nodes in $G$. This allows us to write Eq. \eqref{eq:percnodesdependence} for such dangling nodes---that we call ``first-order dangling''---verbatim, and an identical argument yields an equation analogous to Eq. \eqref{eq:percnodesdependence2}.
\end{proof}

Importantly, we note from Table \ref{tab:perccontribs} that dangling nodes may exist in $G$ only if there are percolating nodes in $G$, that provide ``support'' to dangling nodes. A node $i$ is said to be ``supporting'' in $G$ if it is percolating in $G$ and has a strictly positive connection probability to (a) dangling node(s) in $G$. We also note from Table \ref{tab:perccontribs} that while some dangling nodes may connect to the giant component mostly through direct connections [row (a) columns (2, 3)], others may connect to the giant component mostly through indirect connections via other dangling nodes [row (b) columns (1, 2)]. A node $i$ is said to be ``first-order dangling'' in $G$ if it is dangling in $G$ and has a strictly positive connection probability to (a) percolating node(s) in $G$ and a strictly positive connection probability to no more than $\order{1}$ dangling node(s) in $G$. Similarly, a node $i$ is said to be higher-order dangling in $G$ if it is dangling in $G$ and has a strictly positive connection probability to at least $\littleomega{1}$ dangling nodes in $G$. Crucially, Lemma \ref{lemma:perccontribs} provides us with Eq. \eqref{eq:percnodesdependence2} to estimate the asymptotic percolation probability of a node if we know the percolation behavior of other nodes in the graph, which yields the following corollary.

\begin{corollary}[Percolation probability]\label{lemma:perc}
    For a percolating, first-order dangling, or non-percolating node $i$ in a sparse network $G$, asymptotically:
    \begin{equation}\label{eq:perc}
        \prob{\phi_i(G)}\approx1-\exp\left(-\sum_{j\in[n]}\prob{A_{ij}=1}\prob{\phi_j(G)}\right).
    \end{equation}
\end{corollary}
\begin{proof}
    The proof follows from a direct application of Lemma \ref{lemma:perccontribs}. Since nodes that are not percolating in $G$ contribute vanishingly to the percolation probabilities of percolating or first-order dangling nodes (Eq. \eqref{eq:percnodesdependence}), extending the domain of node $j$ on the RHS of the resulting Eq. \eqref{eq:percnodesdependence2} to the set of all nodes $[n]$ in $G$ has an asymptotically vanishing effect on the percolation probability of such nodes, yielding Eq. \eqref{eq:perc}. For non-percolating nodes, Eq. \eqref{eq:perc} is trivially satisfied. 
\end{proof}
Corollary \ref{lemma:perc} provides us with a \emph{self-consistent} equation (Eq. \eqref{eq:perc}) that can be solved for \emph{all} nodes in $G$---that are not higher-order dangling. However, the peculiar percolation behavior of dangling nodes that depends on an $\littleo{n}$-sized set of nodes makes analyzing their shortest paths more challenging. Thankfully, as we describe next, there is a connection between the existence of dangling nodes and the network containing bottlenecks that simplifies our analysis.

\paragraph*{Networks with bottlenecks}The assumption of sparsity (Eq. \eqref{eq:sparsity_constraint}) will allow us to asymptotically discount the contribution of any given edge to the properties of any given node. However, it may not necessarily allow us to asymptotically discount the contribution of any given node to the properties of shortest paths between any given node pair, requiring us to define an additional asymptotic condition that there are no ``bottlenecks'' between any node pair (Eq. \eqref{eq:nonultrabottlenecked}). Let $i,j$ be a pair of nodes in network $G$ whose expected adjacency matrix is given by $\avg{\mat{A}}$. Then the node pair $(i,j)$ is potentially connected by $S_\mathrm{via}^{ij}\triangleq\{k\in[n]: \avg{A_{ik}}>0, \avg{A_{kj}}>0\}$. If $0<|S_\mathrm{via}^{ij}|=\littleo{n}$ then $(i,j)$ is said to be ``bottlenecked'' by $S_\mathrm{via}^{ij}$. A node $i$ is said to be ``bottlenecking'' in $G$ if $\exists j\in [n],\exists k\in[n]\setminus\{j\},\exists S\subset [n]\setminus\{j,k\}$ such that $i\in S$ and $(j,k)$ is bottlenecked by $S$. The network $G$ and its corresponding network model are said to be bottlenecked if and only if there exists a bottlenecking node in $G$.

\begin{lemma}[Supporting nodes are bottlenecking]\label{lemma:suppbottle} A supporting node in a sparse network $G$ is a bottlenecking node for a node pair such that one of them is (first-order) dangling and the other one is percolating in $G$.
\end{lemma}
\begin{proof}
    Let $i$ be a supporting node in $G$, which implies that it is also a percolating node in $G$, and $\exists S^i_{\mathrm{dan}}\subset[n], |S^i_{\mathrm{dan}}|>0$, a (maximal) set of (first-order) dangling nodes in $G$ that it must be contributing $\littleo{1}$ to the percolation probability of (Table \ref{tab:perccontribs} row (a) columns (2, 3)). Every first-order dangling node $j\in S^i_{\mathrm{dan}}$ must connect to a set $S^j_{\mathrm{sup}}$ of size no larger than $\littleo{n}$ (and no smaller than $1$) of supporting (and percolating) nodes including $i$ (Table \ref{tab:perccontribs} row (a) columns (2, 3)), and to a set $S^j_{\mathrm{dan}}$ of size $\order{n}$ of dangling nodes (Table \ref{tab:perccontribs} row (b)). Let $S_\mathrm{p}(G)$ be the set of percolating nodes in $G$, $S_i^{>0}$ be the set of all nodes which $i$ has a positive connection probability to, and $S_{i\mathrm{p}}^{>0}\triangleq S_i^{>0}\cap S_\mathrm{p}(G)$. Since $i$ is percolating 
    \begin{equation}
        \label{eq:ipercolatingsize}
        |S_{i\mathrm{p}}^{>0}|=\bigomega{n}
    \end{equation}
    (Table \ref{tab:perccontribs} row (a) column (1)).

    Evidently, $S_j^{>0} = S^j_\mathrm{sup}\cup S^j_\mathrm{dan}$. If $|S^j_\mathrm{dan}|=\littleo{n}$ then $|S_j^{>0}|=\littleo{n}$. For any node $k\in S_{i\mathrm{p}}^{>0}$, let $S_{jik}^{>0}\triangleq\{l\in S_j^{>0}: k\in S_l^{>0}\}$ be the set of nodes in $S_j^{>0}$ that can connect to $k$. Since $|S_{jik}^{>0}|\le |S_j^{>0}|=\littleo{n}$ it follows that $|S_{jik}^{>0}|=\littleo{n}$, and because $i\in S_{jik}^{>0}$ we have $|S_{jik}^{>0}|>0$. It follows that $(j,k)$ is bottlenecked by $S_{jik}^{>0}$, and therefore $i$ is a bottlenecking node for a dangling ($j$) and percolating node ($k$).

    If $|S^j_\mathrm{dan}|=\bigomega{n}$, then it must be the case that $\exists k\in S_{i\mathrm{p}}^{>0}$ such that, if $S_{jik}^{>0\mathrm{d}}\triangleq\{l\in S^j_\mathrm{dan}: k\in S_l^{>0}\}$ is the set of nodes in $S^j_\mathrm{dan}$ that can connect to $k$, then $|S_{jik}^{>0\mathrm{d}}|=\littleo{n}$. For if that were not the case then $\forall k\in S_{i\mathrm{p}}^{>0}: |S_{jik}^{>0\mathrm{d}}|=\bigomega{n}$ and using Eq. \eqref{eq:ipercolatingsize} we then have $\bigomega{n^2}$ non-zero connection probabilities between $S^j_\mathrm{dan}$ and $S_{i\mathrm{p}}^{>0}$. However, since every dangling node in $S^j_\mathrm{dan}$ can have non-zero connection probabilities to no more than $\littleo{n}$ percolating nodes (as otherwise it would be a percolating node; Table \ref{tab:perccontribs} row (a)), it follows that we must have $\littleo{n^2}$ non-zero connection probabilities between $S^j_\mathrm{dan}$ and $S_{i\mathrm{p}}^{>0}$---that is, we have a contradiction. Consequently, we pick one such node $k\in S_{i\mathrm{p}}^{>0}$ such that $|S_{jik}^{>0\mathrm{d}}|=\littleo{n}$. Then defining $S_{jik}^{>0}$ as above the rest of the argument follows, i.e. $i$ is a bottlenecking node for a dangling ($j$) and percolating node ($k$).    
\end{proof}

\begin{corollary}\label{lemma:dangbottle} A sparse ensemble average network (SEAN) has no dangling nodes.
\end{corollary}
\begin{proof}
    Say there exist dangling nodes in a sparse ensemble average network (SEAN), which implies the existence of supporting nodes that would be bottlenecking (Lemma \ref{lemma:suppbottle}, which further implies that the network is bottlenecked, but a SEAN is assumed to have no bottlenecks (Eq. \eqref{eq:nonultrabottlenecked}). Therefore, we have a contradiction, and the network must not have any dangling nodes.
\end{proof}
Corollary \ref{lemma:dangbottle} provides us the key technical detail that a sparse ensemble average network (SEAN) has no dangling nodes, which allows us to ignore dangling nodes when considering the geodesics in SEANs. In particular, when deriving the initial condition of the GLD (Lemma \ref{lemma:deg_gcc} in Appendix \ref{sec:apdx_connect_gcc}) and bridging probabilities for the recursive step of the GLD (Lemma \ref{lemma:vanishingbridge} in Appendix \ref{sec:apdx_lemmas12_bridging}), as we did in the preceeding sections, we need to only consider nodes that are either percolating or non-percolating.

\paragraph*{Bottlenecks in practice} Since we assume that networks do not have bottlenecks, we briefly remark on distinguishing a sparse network (model) from a sparse network (model) without bottlenecks, i.e. SEAN. In theory, the distinction is based on the asymptotic scaling of the network's structure with the network's size. Therefore, if the actual statistical network model is known, then it's a matter of verifying under the model assumptions that there are no bottlenecks between any node pair. In particular, if the model satisfies the assumptions of an SGRN model, then we do not have to perform any verification as we know that the model cannot have any bottlenecks (Lemma \ref{lemma:generalultrabottleneck}). In practice, if we are given a particular---and potentially partially observed---network, then one would have to assume a ``best-fitting'' model that the network could have been reasonably sampled from. If the best-fitting network model belongs to the class of SGRN or SEAN models then the results immediately follow. If, on the other hand, the best-fitting model \emph{does not} belong to the class of SGRN/SEAN models---recall that SGRNs are a subclass of SEANs---then the network can potentially consist of bottlenecks and the results may have to be extended. In principle, one can always infer a model from the broad class of SGRN/SEAN models that best fits the network---as SBM inference methods do \cite{holland1983sbm, peixoto2014nestedsbm, newman2018inference}---and apply our results. However, the accuracy of the GLD so obtained would depend on how ``close'' the best-fitting SGRN/SEAN model is to the best-fitting network model.

\section{\label{sec:apdx_general_sparsity}Degree and sparsity}

The out- and in-degrees of node $i$, in a network encoded by the adjacency matrix $\mat{A}$, respectively encode the number of edges emanating from and incident on $i$:
\begin{subequations}
\label{eq:def_degree}
\begin{align}
    \degreeout_i &\triangleq \sum_{j\ne i}A_{ij}, \\
    \degreein_i &\triangleq \sum_{j\ne i}A_{ji}.
\end{align}
\end{subequations}
\paragraph*{Degree distribution and sparsity}
Since edges are added independently, from Eq. \eqref{eq:ciem} the out- and in-degrees of node $i$ in a sparse ensemble average network (SEAN) follow a Poisson binomial distribution, with the expectation from Eq. \eqref{eq:def_degree}:
\begin{subequations}
\label{eq:degree_ensemble_node}
\begin{align}
    \label{eq:degree_ensemble_node_out}
    \avg{\degreeout_i}&=\avg{\sum_{j\ne i}A_{ij}}=\sum_{j\ne i}\avg{A_{ij}},\\
    \label{eq:degree_ensemble_node_in}
    \avg{\degreein_i}&=\avg{\sum_{j\ne i}A_{ji}}=\sum_{j\ne i}\avg{A_{ji}},
\end{align}
\end{subequations}
where the second equality in Eqs. \eqref{eq:degree_ensemble_node_out}, \eqref{eq:degree_ensemble_node_in} arise by linearity of expectation. We further define the mean network degree as the mean network \emph{out}-degree:
\begin{equation}
    \label{eq:degree_ensemble_network}
    \avgdegree\triangleq\expect{\degreeout}=\frac{\sum_i\sum_{j\ne i}\avg{A_{ij}}}{n},
\end{equation}
which we remark can be equivalently defined as the mean network \emph{in}-degree, since both are equal. We use the notation $\expect{\cdot}$ when averaging over nodes. For undirected networks we define the degree of node $i$ as:
\begin{equation}
    \label{eq:def_degree_undirected}
    \degree_i \triangleq \degreeout_i = \degreein_i,
\end{equation}
whose expectation is provided by Eq. \eqref{eq:degree_ensemble_node}. In this work, we assume that the network is sparse (Eq. \eqref{eq:sparsity_constraint}) which implies asymptotically bounded expected degree of every node---as seen from Eq. \eqref{eq:degree_ensemble_node}. By Le Cam's theorem \cite{lecam1960poissonbinomial} the  in-/out-degree distribution of a given node will asymptotically (as $n\to\infty$) approach the Poisson distribution whose mean $\avg{d_i^+},\avg{d_i^-}$ is given by Eqs. \eqref{eq:degree_ensemble_node_out},  \eqref{eq:degree_ensemble_node_in}. Evidently, the degree distribution for the whole network will be a mixture of Poisson distributions, whose expectation $\avg{d}$ is given by Eq. \eqref{eq:degree_ensemble_network}.

\paragraph*{Expected degree functions in SGRNs} Consider the sparse general random networks (SGRNs) of Sec. \ref{sec:general_graphs} that are a special case of SEANs.  For a node located at $x\in V$, let $\degreeout(x),\degreein(x)$ be its out- and in-degree. Recall from Eq. \eqref{eq:degree_ensemble_node_out} that for a SEAN the expected out-degree for node $i$ is given by $\avg{\degreeout_i}=\sum_{j\ne i}\avg{A_{ij}}$. For the SGRN, the sum of expectation over $n-1$ nodes asymptotically translates to $n$ times the expectation over node space $V$. This yields expressions analogous to Eq. \eqref{eq:degree_ensemble_node} for the expected out- and in-degree at $x$, and expected network degree:
\begin{subequations}
    \label{eq:general_degrees}
    \begin{align}
        \label{eq:general_degrees_out}
        \avg{\degreeout(x)}&\triangleq\meandegreeout(x) = n\int_V\nu(x,y)d\mu(y),\\
        \label{eq:general_degrees_in}
        \avg{\degreein(x)}&\triangleq\meandegreein(x) = n\int_V\nu(y,x)d\mu(y),\\
        \label{eq:general_degrees_mean}
        \avgdegree&\triangleq \expectwrt{\mu}{\meandegreeout(x)}=\int_V\meandegreeout(x)d\mu(x),
    \end{align}
\end{subequations}
where we use the convention of defining the mean network degree as the mean network \emph{out}-degree, and the notation $\expectwrt{\mu}{\cdot}$ when averaging over the node space $V$. For undirected networks $\nu(x,y)=\nu(y,x)$ almost everywhere, in which case we define the degree of node at $x$ similarly to Eq. \eqref{eq:def_degree_undirected} as $\degree(x)\triangleq\degreeout(x)= \degreein(x)$, yielding:
\begin{equation}
    \label{eq:general_degrees_deg}
        \avg{\degree(x)}\triangleq\meandegree(x)=\meandegreeout(x)= \meandegreein(x).
\end{equation}
As before, the sparsity assumption in Eq. \eqref{eq:sparsity_constraint_general} asymptotically implies bounded node degrees and a Poisson degree distribution:
\begin{equation}
    \label{eq:degree_distribution}
    \degree(x)\sim\poisson{\meandegree(x)},
\end{equation}
almost everywhere, that is, everywhere except for $x\in U$ where $U \subset V$ is of zero measure. Similar expressions follow for the out- and in-degrees, and the network degree distribution is a mixture of Poisson distributions with the expectation $\avgdegree$.

\paragraph*{Integral operator}Sparsity also implies that the connectivity kernel scaled by the network size $n$ is square integrable, i.e. $n\nu\in\mathcal{L}^2(V\times V, \mu\times\mu)$. This allows us to define the Hilbert-Schmidt integral operator $T:\mathcal{L}^2(V,\mu)\to\mathcal{L}^2(V,\mu)$ in Eq. \eqref{eq:integral_op} that acts on the space of $\mathcal{L}^2$ functions in $V$. Being a Hilbert-Schmidt operator, $T$ is compact: it is either of finite rank or the limit of finite-rank operators. By a similar argument we obtain $n\omegaaf_1\in\mathcal{L}^2(V\times V, \mu\times\mu)$. By Fubini's theorem, the finitely iterated integrals in Eq. \eqref{eq:spd_analytic_general_omega}---for the conditional probability mass function (PMF)---can be computed in any order.

\section{\label{sec:apdx_finite_size}Finite-size effects in the GLD}
Since we consider the asymptotic limit of the network size to discount correlations in the network, in this section we analyze finite-size effects for the analytic form of the GLD, given by the recursive Eq. \eqref{eq:spd_main}. The proof for Lemma \ref{lemma:vanishingbridge} provides a general asymptotic argument to intuit how the probability mass of a node to be at distance $l$ from a given node, in a sparse ensemble average network, scales with network size $n$. Here, we first invoke the approximate closed form of the GLD obtained in Eq. \eqref{eq:spd_analytic_general_eig_uncorrected} of Sec. \ref{sec:general_graphs} to show when this scaling holds, and when it does not, for finite-size networks in the supercritical and subcritical regimes. Then, we describe how finite-size effects manifest in the analytic form of the GLD.

\paragraph*{Asymptotic scaling of geodesic lengths} Lemma \ref{lemma:vanishingbridge} shows that for any \emph{finite} length $l$, the probability of two nodes, in a sparse ensemble average network (SEAN), having a shortest path of length $l$ is either $0$ or scales as $\bigtheta{n^{-1}}$ in the asymptotic limit. This scaling behavior is evident in the expression for the approximate closed form of the survival function of the GLD for a sparse general random network (SGRN), as given by Eq. \eqref{eq:spd_analytic_general_eig_uncorrected}, which is asymptotically tight for finite lengths. That is, for nodes $j,k$ located at $X_j\in V,X_k\in V$:
\begin{equation}
\label{eq:finite_size_effects}
    \begin{split}
        \prob{\lambda_{jk}>l} &= \exp\left(-\frac{1}{n}\sum_{i=1}^NS_l(\tau_i)\varphi_i(X_j)\varphi_i(X_k)\right)\\&\approx 1-\frac{1}{n}\sum_{i=1}^NS_l(\tau_i)\varphi_i(X_j)\varphi_i(X_k)\\
        \implies \prob{\lambda_{jk}=l}&\approx \frac{1}{n}\sum_{i=1}^N\tau_i^l\varphi_i(X_j)\varphi_i(X_k)=\bigtheta{n^{-1}}\textrm{ or }0,
    \end{split}
\end{equation}
where the first approximation uses a first-order expansion of the exponential, and in the second equality we use $\prob{\lambda_{jk}=l}=\prob{\lambda_{jk}>l-1}-\prob{\lambda_{jk}>l}$, and the definition of $S_l(a)\triangleq{a+a^2+\dots+ a^{l}}$ as the geometric sum of $a$ starting at $a$ up to $l$ terms.

\paragraph*{Finite-size effects in the supercritical regime}However, when we begin to consider longer geodesics which scale with the network size $n$, then we can expect deviations between analytics and empirics in the supercritical regime, due to finite-size effects. If length $l=c\log n$ for some $c>0$, then from Eq. \eqref{eq:finite_size_effects} we obtain
\begin{equation}
\label{eq:scaling_logn}
    \prob{\lambda_{jk}=l}=\bigtheta{n^{c\log \tau_1-1}}\textrm{ or }0,
\end{equation}
where $\tau_1$ is the largest eigenvalue of the integral operator (Eq. \eqref{eq:integral_op}). If $\tau_1>1$---which from Theorem 2 in Ref. \onlinecite{loomba2024geodesics2} implies the supercritical regime and vice-versa---then the probability of two nodes having a shortest path of length $l$ does not scale as $\bigtheta{n^{-1}}$, and for $c\ge(\log\tau_1)^{-1}$ can scale as $\bigtheta{n^\epsilon}$ where $\epsilon\ge 0$. Consequently, at such geodesic lengths, we can expect non-vanishing correlations between different geodesics connecting the source and target nodes. To appreciate these finite-size effects, consider the simplest network model of an ER graph, with mean degree $\avgdegree>1$. From Eq. \eqref{eq:spd_er}, we know that the approximate closed form of the survival function of distribution of geodesic lengths $\lambda$ is a discrete version of the Gompertz distribution, i.e. $\prob{\lambda>l} = \exp\left(-a(e^{bl}-1)\right)$, where $a\triangleq \frac{\avgdegree}{n(\avgdegree-1)}$ and $b\triangleq \log \avgdegree$. This distribution has a mode around $l=\frac{\log\frac{n(\avgdegree-1)}{\avgdegree}}{\log \avgdegree}=\bigtheta{\frac{\log n}{\log \avgdegree}}$. Therefore, the finite-size effects should become apparent around the mode of the GLD.

\begin{figure}
    \centering
    \includegraphics[width=\columnwidth]{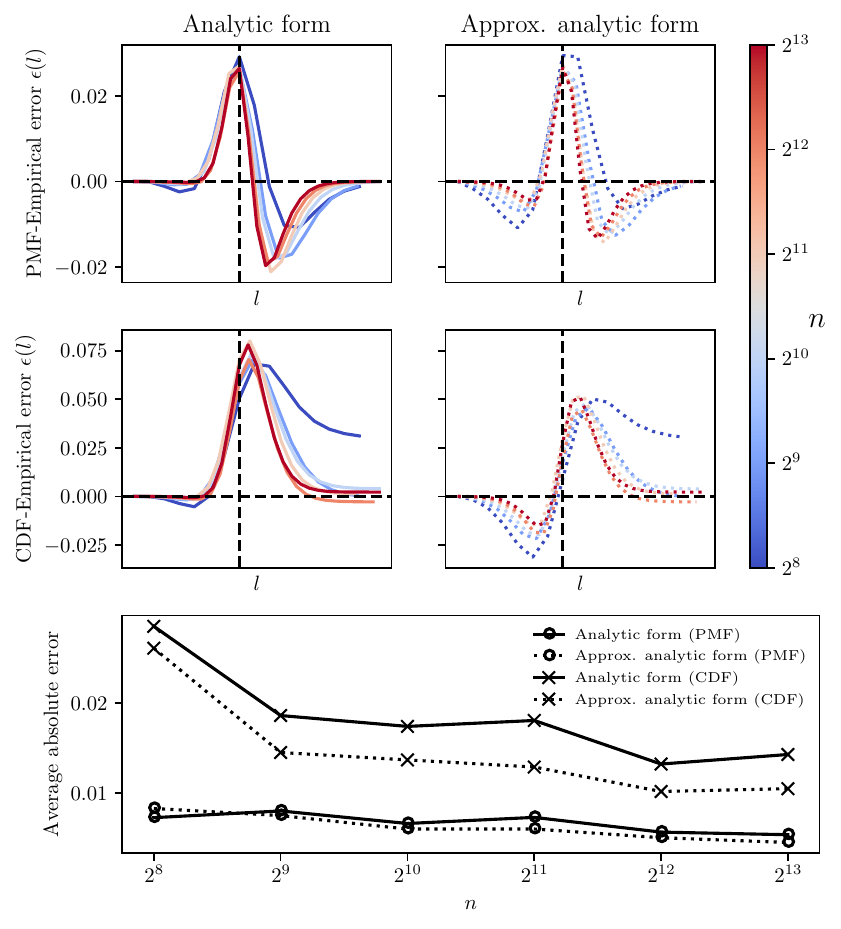}
    \caption{Average error in analytic and approximate analytic forms of the GLD vanishes asymptotically. For an ER graph with fixed mean degree $\avgdegree=2$ and varying network size $n\in\{2^8, 2^9, 2^{10}, 2^{11}, 2^{12}, 2^{13}\}$, the empirical error is computed by subtracting the (top row) empirical PMF and (middle row) CDF over 10 samples from the respective (left column) analytic and (right column) approximate analytic forms of the GLD. We remark that error bars are omitted for clarity, and the support of the distribution is scaled by its analytic mode---marked by the vertical dotted line---to aid comparison across network sizes. Bottom row shows the absolute empirical error averaged over the distribution support. Errors peak around the mode of the distribution but vanish at shorter and longer geodesic lengths, and the average errors decrease for larger networks for both analytic forms of both the PMF and CDF.}
    \label{fig:spd_er_d2_errs}
\end{figure}

We plot the error in the analytic (Eqs. \eqref{eq:spd_main}, \eqref{eq:prob_connect_exact}, and \eqref{eq:gcc_consistency}) and approximate analytic forms (Eqs. \eqref{eq:spd_main} and \eqref{eq:prob_connect_apx}) of the GLD for an ER graph with $\avgdegree=2$ in Fig. \ref{fig:spd_er_d2_errs}, and note that the error in the PMF stays negligible for short geodesics, but eventually rises to its largest at the analytic mode of the GLD i.e. where $l=\bigtheta{\log n}$, regardless of network size---see top row in Fig. \ref{fig:spd_er_d2_errs}. However, non-zero errors become more concentrated around the mode as network size is increased, and consequently the absolute error when averaged over the support of the distribution decreases as network size increases---see bottom row in Fig. \ref{fig:spd_er_d2_errs}. Notably, we observe that the error in the cumulative distribution function (CDF) as $l\to\infty$, which is indicative of the size of the giant component, tends to vanish with network size, for both the analytic and approximate analytic forms of the GLD---see middle row in Fig. \ref{fig:spd_er_d2_errs}, and Fig. \ref{fig:gcc_consistency_spld}. This implies that while the analytic form marginally overestimates the probability mass around the mode, it compensates for it by commensurately underestimating probability mass elsewhere. In particular, while the analytic form shifts mass away from longer geodesic lengths, the approximate analytic form appears to pull it from either sides of the mode---see top row in Fig. \ref{fig:spd_er_d2_errs}.

\begin{figure}[t]
    \centering
    \includegraphics[width=\columnwidth]{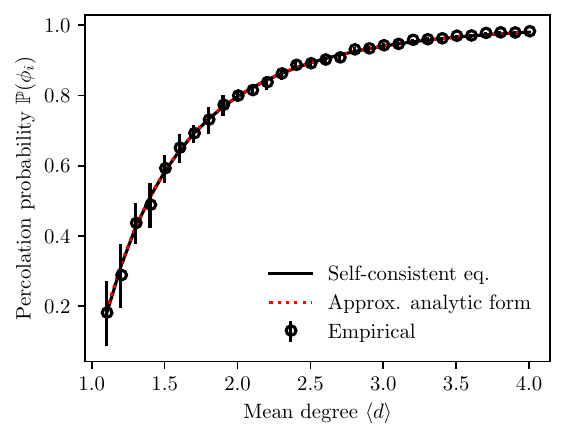}
    \caption{Percolation probabilities derived from the self-consistent equation (solid line) strongly overlap with estimates derived from the limiting value of the approximate analytic form of the GLD (dotted line). Network size is fixed at $n=1024$, while mean degree of the ER graph varies from $\avgdegree\in[1.1, 4]$. Solid and dotted lines indicate analytic solutions derived from Eq. \eqref{eq:gcc_consistency} and Eqs. \eqref{eq:spd_main}, \eqref{eq:perc_prob_limit}, \eqref{eq:perc_prob_pinf} respectively, while symbols ($\circ$) and bars indicate empirical estimates: mean and standard error over 10 network samples. There is good agreement between all three for varying levels of connectivity and percolation probabilities.}
    \label{fig:gcc_consistency_spld}
\end{figure}

\paragraph*{Finite-size effects in the subcritical regime} We now consider the subcritical regime, wherein the shortest path between source and target nodes cannot pass via nodes on a giant component. From Theorem 2 in Ref. \onlinecite{loomba2024geodesics2}, $\tau_1\le 1$ in the subcritical regime, which from Eq. \eqref{eq:scaling_logn} implies that the probability of two nodes having a shortest path of length $l=\bigtheta{\log n}$ is either $0$ or scales as $\bigtheta{n^{-1-\epsilon}}$, where $\epsilon\ge 0$. Thus even when considering geodesic lengths that scale with the network size, the asymptotic correlations between shortest paths should remain vanishingly small in the subcritical regime, and we expect to observe vanishing finite-size effects. This is evident in Fig. \ref{fig:spd_er_smallcomponent} for the GLD in subcritical ER graphs, and in Fig. \ref{fig:bipartite_cdf_smallcomponent} for the GLD in a subcritical bipartite stochastic block model (SBM).

\begin{figure}
    \centering
    \includegraphics[width=\columnwidth]{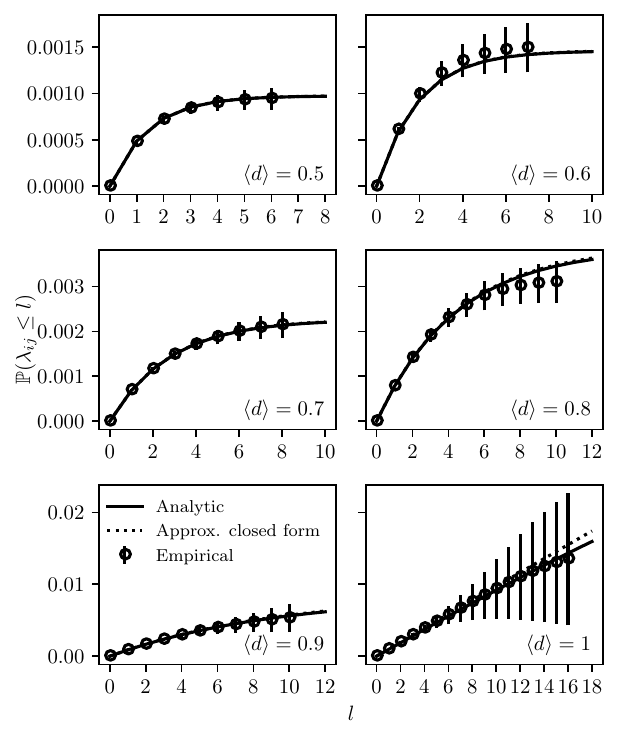}
    \caption{Empirical, analytic and (approximate) closed-form CDF of geodesic lengths on the small components in the subcritical regime agree with each other for an ER graph with varying connectivity. Network size is fixed at $n=1024$, while mean degree varies as $\avgdegree\in\{0.5, 0.6, 0.7, 0.8, 0.9, 1\}$. Solid and dotted lines indicate analytic (Eqs. \eqref{eq:spd_main}, \eqref{eq:prob_connect_apx}) and (approximate) closed-form solutions (Eqs. \eqref{eq:prob_connect_apx}, \eqref{eq:sf_avg} that together yield Eq. \eqref{eq:spd_er}), respectively. Symbols ($\circ$) and bars indicate empirical estimates: mean and standard error over 10 network samples. Both the analytic and (approximate) closed-form GLD are in good agreement for all subcritical connectivities at all geodesic lengths.}
    \label{fig:spd_er_smallcomponent}
\end{figure}

\begin{figure}
    \centering
    \includegraphics[width=\columnwidth]{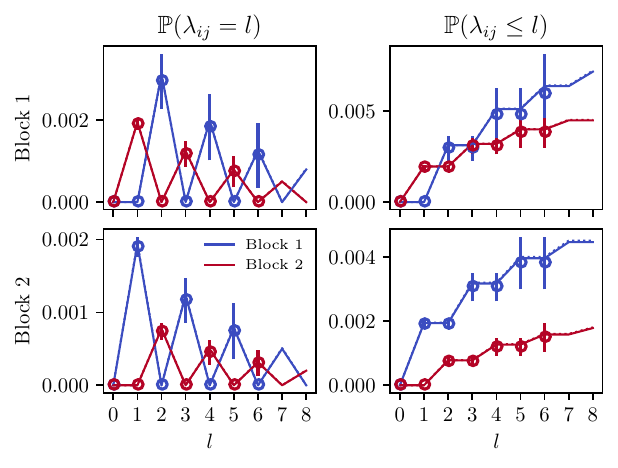}
    \caption{Empirical, analytic, and (approximate) closed-form CDF of geodesic lengths on the small components in the subcritical regime agree with each other for a bipartite stochastic block model (SBM). The block matrix $\mat{B}=\big(\begin{smallmatrix} 0 & 2\\ 2 & 0 \end{smallmatrix}\big)$, distribution vector $\boldsymbol\pi=(0.2, 0.8)$, and network size $n=1024$. Rows correspond to the block membership of source node. Left column depicts the PMF, and the right column depicts the CDF. Solid lines represent analytic form using Eqs. \eqref{eq:spd_main_sbm}, while dotted lines represent the (approximate) closed form using Eq. \eqref{eq:spd_sbm} that strongly overlaps with the analytic form. Symbols ($\circ$) and bars indicate empirical estimates: mean and standard error over 10 network samples.}
    \label{fig:bipartite_cdf_smallcomponent}
\end{figure}

\section{\label{sec:apdx_general}Sparse general random network}


\begin{lemma}\label{lemma:generalultrabottleneck} A sparse general random network (SGRN) is not bottlenecked.
\end{lemma}

\begin{proof}
    Suppose that a sparse general random network is bottlenecked, i.e. there exists a bottlenecking node $i$. It then follows that there exists a node pair $(j,k)$ and their corresponding set of potentially-connecting nodes $S^{jk}_\mathrm{via}\triangleq\{l\in[n]:\avg{A_{jl}}>0,\avg{A_{lk}}>0\}$ such that $0<|S^{jk}_\mathrm{via}|=\littleo{n}$ and $i \in S^{jk}_\mathrm{via}$. That is, the number of nodes that have a positive connection probability from $j$ and to $k$ is strictly positive and scales as $\littleo{n}$.

    Let $\nu^i_\mathrm{in}(x)\triangleq\nu(X_i,x)$ and $\nu^i_\mathrm{out}(x)\triangleq\nu(x,X_i)$ be the respective probabilities of a node to connect from and to node $i$. By continuity of the connectivity kernel, $\nu^i_\mathrm{in}:V\to[0,1)$ and $\nu^i_\mathrm{out}:V\to[0,1)$ are continuous functions, i.e. for every node location $y\in V$ and every neighborhood of the corresponding connectivity functions denoted by $N^{\nu_\mathrm{in}}_{X_iy}\subseteq[0,1)$---that includes an open set containing $\nu(X_i,y)$---and $N^{\nu_\mathrm{out}}_{yX_i}\subseteq[0,1)$---that includes an open set containing $\nu(y,X_i)$---there exist corresponding (unique) neighborhoods in the node space denoted by $N_{X_iy}^{V_\mathrm{in}}\subseteq V$ and $N_{yX_i}^{V_\mathrm{out}}\subseteq V$---that each include some open sets containing $y$---such that $\forall z\in N_{X_iy}^{V_\mathrm{in}}: \nu(X_i,z)\in N^{\nu_\mathrm{in}}_{X_iy}$ and $\forall z\in N_{yX_i}^{V_\mathrm{out}}: \nu(z,X_i)\in N^{\nu_\mathrm{out}}_{yX_i}$.
    
    Applying this to the locations of nodes $i$, $j$ and $k$ with strictly positive neighborhoods $N^{\nu_\mathrm{in}}_{X_jX_i}\subset (0,1)$ and $N^{\nu_\mathrm{out}}_{X_iX_k}\subset(0,1)$---which are valid neighborhoods since the probabilities of $j$ connecting to $i$ and $i$ connecting to $k$ are strictly positive---we get neighborhoods $N^{V_\mathrm{in}}_{X_jX_i}\subseteq V$ and $N^{V_\mathrm{out}}_{X_iX_k}\subseteq V$ that are respectively mapped by the kernel to them. Since these are both neighborhoods of the location $X_i\in V$, it follows that $N^{V_\mathrm{in}}_{X_jX_i} \cap N^{V_\mathrm{out}}_{X_iX_k}$ is also a neighborhood of $X_i$, i.e. $X_i\in N^{V_\mathrm{in}}_{X_jX_i} \cap N^{V_\mathrm{out}}_{X_iX_k}$, and at least some nodes in $S^{jk}_\mathrm{via}$ must be sampled from $N^{V_\mathrm{in}}_{X_jX_i} \cap N^{V_\mathrm{out}}_{X_iX_k}$: Since $i$ has been sampled under $\mu$ (Eq. \eqref{eq:gnodelocation}), this implies that $N^{V_\mathrm{in}}_{X_jX_i} \cap N^{V_\mathrm{out}}_{X_iX_k}$ has a positive measure, and therefore asymptotically the size of $S^{jk}_\mathrm{via}$ cannot scale slower than $\bigtheta{n}$. That is, the number of nodes that have a positive connection probability from $j$ and to $k$ is strictly positive and scales as $\bigtheta{n}$, and we have reached a contradiction. It then follows that the network is not bottlenecked.
\end{proof}

\paragraph*{Equivalent SEAN for an SGRN} Since an SGRN does not contain any bottlenecks (Lemma \ref{lemma:generalultrabottleneck}), we can define an equivalent sparse ensemble average network (SEAN) that can be used as a finite-size representation of the SGRN with $n$ nodes. Using Eq. \eqref{eq:bernoulli_model_general}, we can define the expected adjacency matrix as:
\begin{equation}
    \label{eq:general_to_ensembleavg}
    \avg{A_{ij}}\triangleq\nu(X_i,X_j).
\end{equation}
We emphasize that node locations are themselves random. However asymptotically, by Borel's law of large numbers, the set $\mathcal{X}$ contains a number of nodes from location $x\in V$ in proportion to $\mu(x)$. This yields asymptotic equivalence between the SGRN model, represented by the integral operator $T$ in Eq. \eqref{eq:integral_op}, and its corresponding SEAN model encoded by $\avg{\mat{A}}$ in Eq. \eqref{eq:general_to_ensembleavg}, which allows for asymptotic results on $\avg{\mat{A}}$ that assume no bottlenecks to be extended to those for $T$.

\paragraph*{Approximately equivalent SBM for an SGRN} Since the formalism of SBMs allows us to easily deal with the GLD computationally, by using matrix algebra toolkits, it is useful to show that for any given sparse general random network (SGRN) model, there exist equivalent SBMs up to a desirable level of approximation.
\begin{theorem}[$\epsilon$-equivalent SBMs of SGRN models]\label{lemma:general_sbm_equiv}
Consider a sparse general random network (SGRN) with $n$ nodes in node space $V$ with node density $\mu$ and connectivity kernel $\nu$. Then it has an $\epsilon$-equivalent representation as a stochastic block model (SBM). That is, for any given $\epsilon>0$ there exist:
\begin{enumerate}[label=(\alph*)]
    \item a mapping $f_\epsilon:V\to\{1,2,\dots, k\}$ where
    \begin{equation*}
        k\le
        \begin{cases}
         2^{\lfloor\epsilon^{-1}\rfloor+1}& \textrm{if}\enspace\nu\enspace\textrm{is symmetric,}\\
         4^{\lfloor\epsilon^{-1}\rfloor+1}& \textrm{otherwise,}\\
        \end{cases}
    \end{equation*}
    \item a length-$k$ distribution vector $\boldsymbol\pi_\epsilon$, and
    \item a $k\times k$ block matrix $\mat{B}_\epsilon$  such that for any pair of node locations $(x,y)\in V\times V$, we have: $$n\nu(x,y)-[\mat{B}_\epsilon]_{f(x)f(y)}=\order{\epsilon},$$ almost everywhere.
\end{enumerate}
\end{theorem}
\begin{proof}
The sparsity assumption from Eq. \eqref{eq:sparsity_constraint_general} implies that $\exists m>0, \exists L>0$ such that $\forall n\ge m$ we have $\nu(x,y)\le Ln^{-1}$ almost everywhere on $V\times V$, i.e.
\begin{equation*}
    n\nu(x,y)\le L,
\end{equation*}
where we refer to $n\nu$ as the ``scaled'' connectivity kernel. For the rest of this proof, we ignore zero measure subsets of $V$, i.e. use $V$ to indicate $V\setminus U$ where $\mu(U)=0$.

Let $\epsilon>0$. Since $\nu$ is a non-negative measurable function on $V\times V$, its output (image) can be divided into $p\triangleq\lfloor\epsilon^{-1}\rfloor+1$ intervals of length $\epsilon L$ each, where $\lfloor\cdot\rfloor$ is the floor function. That is, define the half-open intervals:
\begin{equation}\label{eq:sbm_equiv_int}
    I_\epsilon^u\triangleq\left[(u-1)\epsilon L, u\epsilon L\right) \quad \textrm{for } u\in\{1,2,\dots, p\}.
\end{equation}
Then consider the inverse of the scaled connectivity kernel $\nu^{-1}:\realnonneg\to V\times V$. Let $\nu^{-1}(I_\epsilon^u)=\overrightarrow{V}^u_\epsilon\times \overleftarrow{V}^u_\epsilon$ be the measurable set (of positive measure) in $V\times V$ whose connectivity kernel takes values in $I_\epsilon^u$. Since the intervals $I_\epsilon^u$ are disjoint and cover the image of $n\nu$, it must be true that $\{\overrightarrow{V}^u_\epsilon\times \overleftarrow{V}^u_\epsilon\}_{u=1}^p$ forms a partition of $V\times V$. However, the collection $\{\overrightarrow{V}^u_\epsilon\}_{u=1}^p$ will be a \emph{cover} of $V$, and not necessarily a \emph{partition}---i.e. $V=\bigcup_{u=1}^p\overrightarrow{V}^u_\epsilon$ but it need not be the case that for any $u\ne v$ we have $\overrightarrow{V}^u_\epsilon\cap\overrightarrow{V}^v_\epsilon=\emptyset$.

Then consider the following iterative partition refinement process on $V$. Begin with the trivial partition of the node space $\overrightarrow{P}^0_\epsilon\triangleq\{V\}$. Then for $u$ in the ordered set $\{1,2,\dots, p\}$:
\begin{enumerate}
    \item for every node subspace $S_i\in\overrightarrow{P}^{u-1}_\epsilon$,
    \begin{enumerate}
        \item remove $S_i$,
        \item add $S_i\setminus\overrightarrow{V}^u_\epsilon$,
        \item add $S_i\cap\overrightarrow{V}^u_\epsilon$,
    \end{enumerate}
    \item and obtain the new partition refinement $\overrightarrow{P}^u_\epsilon$.
\end{enumerate}
Then the final partition $\overrightarrow{P}^p_\epsilon$---after ignoring empty sets---constitutes a partitioning of the node space into $k$ node subspaces. Evidently, since the refinement is done $p$ times, in the worst case the final size of the partition can be no more than $2^p$, i.e. $k\le 2^{\lfloor\epsilon^{-1}\rfloor+1}$. (We remark that if the kernel $\nu$ is symmetric, then the partition $\overleftarrow{P}^p_\epsilon$ obtained by considering the target node space collection $\{\overleftarrow{V}^u_\epsilon\}_{u=1}^p$ will be identical to $\overrightarrow{P}^p_\epsilon$. If not, we can consider both collections $\{\overrightarrow{V}^u_\epsilon\}_{u=1}^p$ and $\{\overleftarrow{V}^u_\epsilon\}_{u=1}^p$ at the same time in the partition refinement process, in which case the worst case value of $k$ is $2^{2p}$, i.e. $k\le 4^{\lfloor\epsilon^{-1}\rfloor+1}$.) The final collection of node subspaces can be written as an ordered set $P_\epsilon\triangleq\{V_\epsilon^1,V_\epsilon^2,\dots, V_\epsilon^k\}$. Let $\mathbb{I}_X:V\to\{0,1\}$ be the indicator function identifying if $x\in X$, for some $X\subseteq V$. Then define the index-map $f_\epsilon(x)\triangleq \sum_{i=1}^ki\mathbb{I}_{V_\epsilon^i}(x)$. This shows part (a) of the theorem.

Next, define the probability distribution vector $\boldsymbol{\pi}_\epsilon\triangleq\{\mu(V_\epsilon^i)\}_{i=1}^k$, where $\mu(\cdot)$ indicates the measure. Since we have removed any subspaces of zero measure, and since $P_\epsilon$ is a partitioning of $V$, $\boldsymbol{\pi}_\epsilon$ is a valid distribution vector whose entries are positive-valued and sum up to $1$. This shows part (b) of the theorem.

Finally, let $h:\mathcal{P}(V)\times\mathcal{P}(V)\to\{0,1\}$ be the indicator function $h(X,Y)$ identifying if $X\subseteq Y$, for some $X,Y\subseteq V$, where we use $\mathcal{P}(\cdot)$ to denote the power set. Then define a $k\times k$ block matrix $\mat{B}_\epsilon$, such that $$\idx{\mat{B}_\epsilon}{ij}\triangleq\sum_{u=1}^p(u-1)\epsilon L\, h(V_\epsilon^i,\overrightarrow{V}^u_\epsilon)\,h(V_\epsilon^j,\overleftarrow{V}^u_\epsilon).$$Since $\{\overrightarrow{V}^u_\epsilon\times \overleftarrow{V}^u_\epsilon\}_{u=1}^p$ forms a partition of $V\times V$, and by the construction of $P_\epsilon$ we know that $h(V_\epsilon^i,\overrightarrow{V}^u_\epsilon)h(V_\epsilon^j,\overleftarrow{V}^u_\epsilon)=1$ for exactly one such $u$ and is $0$ otherwise, $\idx{\mat{B}_\epsilon}{ij}$ takes the value $(u-1)\epsilon L$. From the definition of the intervals in Eq. \eqref{eq:sbm_equiv_int}, we know that for $x\in V_\epsilon^i$ and $y\in V_\epsilon^j$: $n\nu(x,y)\in I_\epsilon^u$. Consequently, we get for $x,y$:
\begin{equation*}
\begin{split}
     &0\le n\nu(x,y)-[\mat{B}_\epsilon]_{f(x)f(y)}<\epsilon L\\
     \implies& n\nu(x,y)-[\mat{B}_\epsilon]_{f(x)f(y)} = \order{\epsilon} \quad \textrm{almost everywhere}
\end{split}
\end{equation*}
This shows part (c) of the theorem.
\end{proof}
The above theorem provides technical support to the generation of equivalent SBMs by appropriate ``discretization'' of values that the connectivity kernel can take, and consequently of the node space $V$. If the node space is Euclidean, say $V\subseteq\real^k$, and the connectivity kernel is ``sufficiently regular'', then we could consider a ``sufficiently fine'' discretization of $V$ such that Theorem \ref{lemma:general_sbm_equiv}, and therefore a good SBM approximation, holds. In particular, if the connectivity kernel $\nu$ is Darboux (or equivalently Riemann) integrable on $V\times V$, then for any $\epsilon>0$ there would exist a sufficiently fine partition $P_\epsilon$ of $V\times V$ such that the difference in the supremum and infimum of $\nu$ evaluated on each ``cell'' of $P_\epsilon$ is less than $\epsilon$, and the final partition of $V$ can be obtained via the iterative refinement process described in the proof for Theorem \ref{lemma:general_sbm_equiv}. We now outline two methods for continuous models in Euclidean space in one and higher dimensions.

\paragraph*{MLE-SBM approximation in one dimension}
For continuous models on node space $\real$, such as graphons described in Sec. \ref{sec:graphons}, we can define an equivalent $m$-block SBM with equi-proportioned blocks. Let $F:\real\to[0,1]$ be the cumulative distribution function (CDF) for this node space, i.e. $F(x)\triangleq\int_{-\infty}^x\mu(x)dx$. Then using the inverse of $F$, define a size $k$ partition of $\real$ given by:
\begin{equation*}
    I^i_m\triangleq
    \begin{cases}
    \left[F^{-1}\left(\frac{i-1}{m}\right),F^{-1}\left(\frac{i}{m}\right)\right) & \textrm{if } i\in\{1,2,\dots, m-1\},\\
    \left[F^{-1}\left(\frac{i-1}{m}\right),F^{-1}\left(\frac{i}{m}\right)\right]&\textrm{if } i=m.
    \end{cases}
\end{equation*}
Evidently this is an equi-sized partition: $$\mu\left(I_m^i\right)=F\left(F^{-1}\left(\frac{i}{m}\right)\right)-F\left(F^{-1}\left(\frac{i-1}{m}\right)\right)=m^{-1},$$
which generates the index-map:
\begin{equation*}
    f_m(x)\triangleq\sum_{i=1}^mi\mathbb{I}_{I^i_m}(x),
\end{equation*}
where $\mathbb{I}_{X}$ is the indicator function, and the distribution vector:
\begin{equation*}
    \boldsymbol{\pi}_m=\frac{1}{m}\ones{m},
\end{equation*}
where $\ones{m}$ is the all-ones vector of size $m$. We can also define the block matrix $\mat{B}_m$, by invoking a likelihood-maximisation argument. Consider networks sampled from the original general random graph model, and nodes labelled according to $f_m$, then the maximum likelihood estimate (MLE) of the block matrix---see Eq. \eqref{eq:sbm_mle} and Appendix \ref{sec:apdx_gcsbm}---would asymptotically approach the average affinities between blocks defined by the above partition. Consequently, define the block matrix $\mat{B}_m$ as: 
\begin{align*}
    [\mat{B}_m]_{ij}&\triangleq n\frac{\int_{I^i_m}\int_{I^j_m}\nu(x,y)d\mu(y)d\mu(x)}{\int_{I^i_m}\int_{I^j_m}d\mu(y)d\mu(x)}\\&=nm^2\int_{I^i_m}\int_{I^j_m}\nu(x,y)d\mu(y)d\mu(x).
\end{align*}
Then we refer to $(f_m, \boldsymbol{\pi}_m, \mat{B}_m)$ as the ``$m$-block MLE-SBM approximation'' for the model in $\real$. As $m$ grows larger, the approximation is expected to become better, but it remains consistent with the MLE for any value of $m$. In Fig. \ref{fig:spl_grgg_sbm_rank1}, we show for a one-dimensional Gaussian RGG its $32$-block MLE-SBM approximation, and its consequent ``rank-1 approximation'' (that retains only the leading eigenpair of its integral operator in Eq. \eqref{eq:integral_op}).
\begin{figure}
    \centering
    \subfloat[Gaussian RGG ($R=0.1$)]{\label{fig:spl_grgg_grgg} \includegraphics[width=\columnwidth]{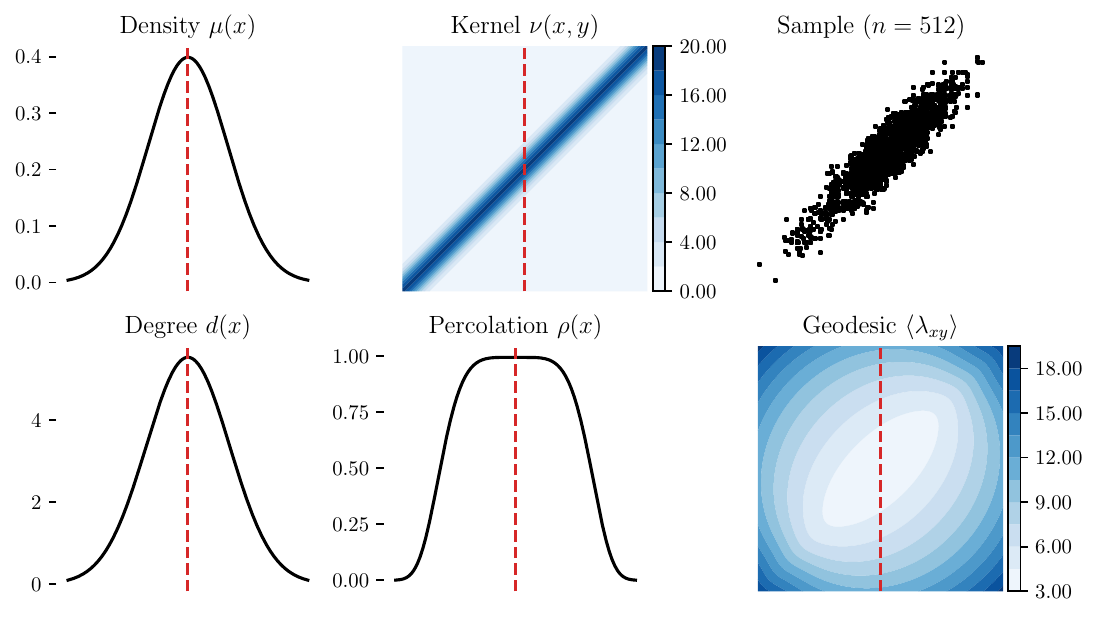}}\\
    \subfloat[Gaussian RGG to SBM ($k=32$)]{\label{fig:spl_grgg_sbm} \includegraphics[width=\columnwidth]{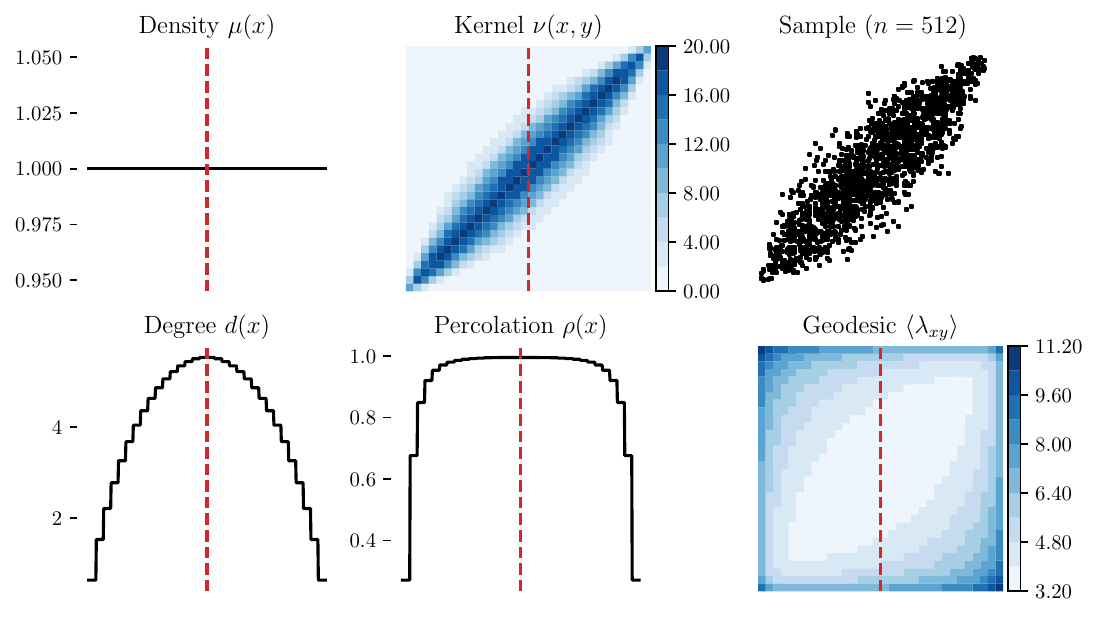}}\\
    \subfloat[SBM to rank-1 graphon]{\label{fig:spl_grgg_rank1} \includegraphics[width=\columnwidth]{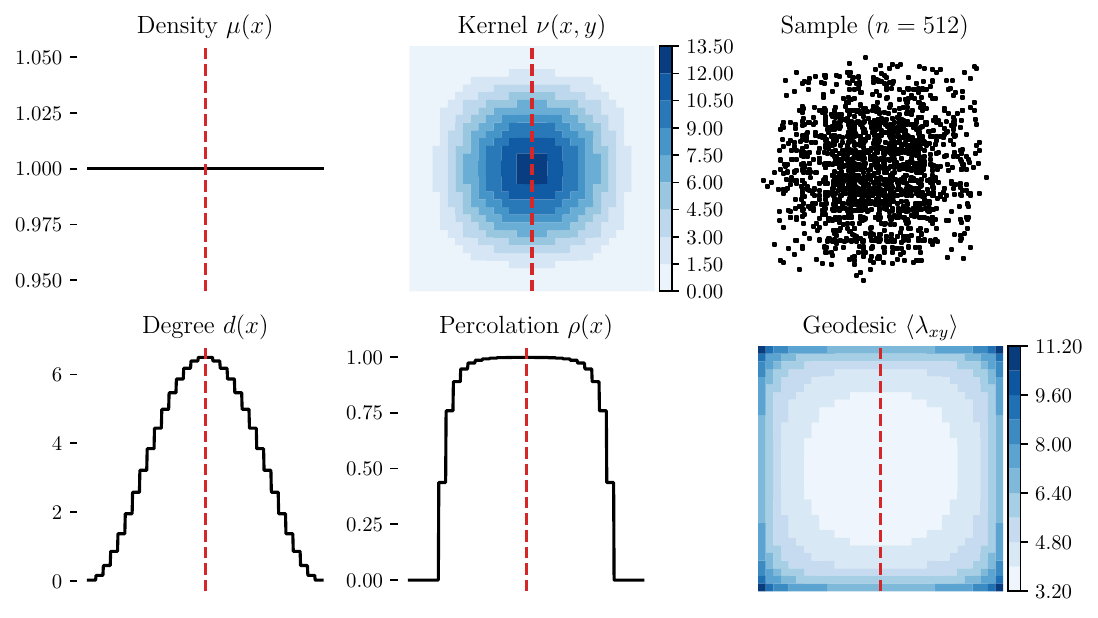}}
    \caption{Node and node pair functions for a (a) $1$-dimensional Gaussian RGG, and its (b) discretization into a $32$-block SBM, which can be further converted into (c) an (approximate) rank-1 graphon by retaining only the largest eigenpair of the integral operator (Eq. \eqref{eq:integral_op}). The functions for Gaussian RGG are viewed on the node space $[-3,3]$. For description of each node and node pair function, refer to Fig. \ref{fig:spl_models}.}
    \label{fig:spl_grgg_sbm_rank1}
\end{figure}

\paragraph*{SBM approximation in higher dimensions}
For continuous models in a higher-dimensional node space $\real^k$, it may be computationally infeasible to integrate the node density $\mu$. However, if we consider a sufficiently fine uniform partition, then assuming smoothness of the node density and connectivity kernel we can consider their point estimates instead. Let $\mat{Z}_{k\to m}$ be a $k\times m$ matrix indicating a column-wise collection of $m$ points that defines a uniform partition of a bounded subspace of $\real^k$, such that negligible probability mass lies outside of it. Define the distribution vector $\boldsymbol{\pi}_{k\to m}$ such that $[\boldsymbol{\pi}_{k\to m}]_i\triangleq\frac{\mu([\mat{Z}_{k\to m}]_{:i})}{\sum_i\mu([\mat{Z}_{k\to m}]_{:i})}$, and kernel matrix $\mat{B}_{k\to m}$ such that $[\mat{B}_{k\to m}]_{ij}\triangleq n\nu\left([\mat{Z}_{k\to m}]_{:i}, [\mat{Z}_{k\to m}]_{:j}\right)$. Then we refer to $(\mat{Z}_{k\to m}, \boldsymbol{\pi}_{k\to m}, \mat{B}_{k\to m})$ as the ``$m$-block SBM approximation'' for the model in $\real^k$. Evidently, as $m$ grows larger, the approximation becomes better. In this article, we make use of the $m$-block SBM approximation only for computing average global statistics---like estimating the mean percolation vector for random dot product graphs (RDPGs), and fitting parameters of the percolation curve for Gaussian random geometric graphs (RGGs)---and thus the errors are expected to be even lower due to averaging.

\paragraph*{Illustrative examples}For an RDPG, consider solving Eq. \eqref{eq:gcc_consistency_rdpg_rho} for the mean percolation vector $\boldsymbol\rho$. Let $(\mat{Z},\boldsymbol\pi,\mat{B})$ be its $m$-block SBM approximation. Then the integral in Eq. \eqref{eq:gcc_consistency_rdpg_rho} can be approximated by:
\begin{equation}
    \label{eq:rho_rdpg}
    \boldsymbol{\rho}\approx\boldsymbol{\phi}-\mat{Z}\,\diag{\boldsymbol\pi}\exp\left(-n\beta \mat{Z}^T\boldsymbol{\rho}\right),
\end{equation}
which is a vector equation that can be numerically solved via function iteration. We remark that for estimating percolation probabilities for the illustrative example of a $3$-dimensional Dirichlet-RDPG, we uniformly partition the 2-simplex---an equilateral triangle---into 4096 equi-triangular partitions.

For a Gaussian RGG, we use the $m$-block SBM approximation to obtain the percolation probability at the points of the uniform partition. Let $(\mat{Z},\boldsymbol\pi,\mat{B})$ be the Gaussian RGG's $m$-block SBM approximation, such that negligible probability mass lies outside of the partition defined by $\mat{Z}$. Define $\boldsymbol{\rho}$ to be a length $m$ vector indicating the percolation probabilities at $m$ locations in $\real^k$. Then from Eq. \eqref{eq:gcc_consistency_general}, we obtain:
\begin{equation*}
    \boldsymbol{\rho} \approx \ones{m}-\exp\left(-n\mat{B}\,\diag{\boldsymbol\pi}\boldsymbol{\rho}\right),
\end{equation*}
where $\ones{m}$ is the all-ones vector of length $m$. This is a vector self-consistent equation in $\boldsymbol{\rho}$, whose non-trivial solution can be obtained via function iteration. Evidently, this equation's form resembles Eq. \eqref{eq:gcc_consistency_sbm} for percolation probabilities in an SBM. Since $\boldsymbol{\rho}$ defines percolation probabilities on a fine partition, assuming smoothness of the percolation probability function in $\real^k$, i.e. of $\rho(\boldsymbol{x})$, we can interpolate for an arbitrary point $\boldsymbol{x}$, or infer the parametric form of a percolation surface in $\real^k$. We remark that, for estimating percolation probabilities for the illustrative example of a $2$-dimensional Gaussian RGG, we restrict to the subspace $[-3,3]\times[-3,3]$, which includes about $99\%$ of the probability mass \cite{wang2015confidence}, divided into 4096 square partitions.

\section{\label{sec:apdx_graph_models}Properties of specific SGRNs}
We elaborate below on the key results on the expected degree function, percolation probability function, and the geodesic length distribution (GLD) of various sparse general random network models considered in Sec. \ref{sec:general_graphs}.

\subsection{\label{sec:apdx_sbm}Stochastic block model (SBM)}
Consider an SBM with $k$ blocks or communities, such that a node belongs to one of the communities $x\in\{1,2,\dots, k\}$ according to a categorical distribution:
\begin{equation}
    \label{eq:sbm_mu}
    \mu(x)=\pi_x,
\end{equation}
where the distribution vector $\boldsymbol\pi$ encodes the probability of being from a community. The connectivity kernel is encoded by the $k\times k$ ``block matrix'' $\mat{B}$:
\begin{equation}
    \label{eq:sbm_nu}
    \nu(x,y)=\frac{B_{xy}}{n}.
\end{equation}
\paragraph*{Expected degree}Functions for expected degree $\meandegree$, percolation probability $\rho$, and the conditional PMF $\omegaaf_l$ and survival function $\psiaf_l$ of the GLD, too possess a block structure. Using Eqs. \eqref{eq:general_degrees_deg}, \eqref{eq:sbm_mu} and \eqref{eq:sbm_nu}, the expected degree of node in block $x$ is given by:
\begin{equation*}
    \begin{split}
        \meandegree(x)&=n\int_{\{1,2,\dots, k\}}\nu(x,y)\,d\mu(y)=\sum_{y=1}^kB_{xy}\pi_y=[\mat{B}\boldsymbol\pi]_x.
    \end{split}
\end{equation*}
Similarly the average degree of the network is given from Eq. \ref{eq:general_degrees_mean} by:
\begin{equation}
    \label{eq:sbm_degree_mean}
    \begin{split}
        \avgdegree&=\int_{\{1,2,\dots, k\}}\meandegree(x)\,d\mu(x)=\sum_{x=1}^k[\mat{B}\boldsymbol{\pi}]_x\pi_x= \boldsymbol{\pi}^T\mat{B}\boldsymbol\pi.
    \end{split}
\end{equation}

\paragraph*{Approximate closed form of the GLD} Next, consider the conditional PMF of the approximate closed form of the GLD from Eq. \eqref{eq:spd_general_omega}; using Eqs. \eqref{eq:sbm_mu} and \eqref{eq:sbm_nu}:
\begin{equation}
    \label{eq:sbm_omega}
    \begin{split}
        \omegaacf_2(x,y)&=n\int_{\{1,2,\dots, k\}}\nu(x,z)\nu(z,y)\,d\mu(z)\\&=\frac{1}{n}\sum_{z=1}^kB_{xz}B_{zy}\pi_z=\frac{[\mat{B\Pi B}]_{xy}}{n}\\
        \implies \omegaacf_l(x,y)&=\frac{\idx{\mat{B}(\mat{\Pi B})^{l-1}}{xy}}{n},
    \end{split}
\end{equation}
for $l>0$, where $\mat{\Pi}\triangleq \diag{\boldsymbol{\pi}}$ and the final implication is obtained via induction. Consequently from Eq. \eqref{eq:spd_general_psi} we can write the survival function of the GLD as a ``block matrix'' containing the survival function of the GLD between block pairs:
\begin{equation}
    \label{eq:sbm_psi}
    \Psiacf_l\triangleq\exp\left(-\frac{\mat{B}}{n}\sum_{i=1}^{l}\left(\mat{\Pi B}\right)^{i-1}\right).
\end{equation}

\paragraph*{Illustrative example: equi-sized planted-partition SBM}

For a non-trivial, yet simple, block structure, consider the $k$-block planted-partition SBM with equi-sized blocks, i.e. $\mat{B}=\delta \eyes{k}+\beta\ones{k}\ones{k}^T$ and $\boldsymbol{\pi}=\ones{k}/k$, where $\ones{k}$ is the all-ones vector of length $k$, and $\eyes{k}$ is the identity matrix of size $k\times k$. Let the mean degree be $\avgdegree$, which from Eq. \eqref{eq:sbm_degree_mean} leads to the constraint:
$$\beta=\avgdegree-\frac{\delta}{k}.$$
Here, $\delta$ models for the level of homophily (heterophily) in the network. If $\delta>0$ ($\delta<0$) then nodes will connect more (less) strongly with nodes of their community than they would to nodes outside of their community. $\mat{B}$ must be a non-negative matrix, which constrains $\delta$ to  $[-\avgdegree k/(k-1),\avgdegree k]$. We now consider the closed-form conditional PMF of the GLD, which from Eq. \eqref{eq:sbm_omega} will involve powers of the matrix $\mat{\Pi B}$. For any $k\times k$ matrix $\mat{M}=p\eyes{k}+q\ones{k}\ones{k}^T$, its matrix powers can be written using a binomial expansion since the matrices $\eyes{k}$ and $\ones{k}\ones{k}^T$ commute:
\begin{equation}
    \label{eq:pp_matrix_powers}
    \begin{split}
        \mat{M}^n&=\left(p\eyes{k}+q\ones{k}\ones{k}^T\right)^n\\&=(p\eyes{k})^n+\sum_{i=1}^n\binom{n}{i}(p\eyes{k})^{n-i}(q\ones{k}\ones{k}^T)^i\\
        &=p^n\eyes{k}+\sum_{i=1}^n\binom{n}{i}p^{n-i}\frac{(qk)^i}{k}\ones{k}\ones{k}^T\\&=p^n\eyes{k}+\frac{(p+qk)^n-p^n}{k}\ones{k}\ones{k}^T\\
        \implies\sum_{j=1}^l\mat{M}^j&=\left(\sum_{j=1}^lp^j\right)\eyes{k}\\&\quad+\left[\sum_{j=1}^l(p+qk)^j-\sum_{j=1}^lp^j\right]\frac{\ones{k}\ones{k}^T}{k}\\
        \implies S_l(\mat{M})&=S_l(p)\eyes{k}+\frac{S_l(p+qk)-S_l(p)}{k}\ones{k}\ones{k}^T,
    \end{split}
\end{equation}
where $S_l(a)\triangleq{a+a^2+\dots+ a^l}=a(a^l-1)(a-1)^{-1}$ represents the geometric sum of $a$ up to $l$ terms. We can write $\mat{\Pi B}=\frac{\delta}{k}\eyes{k}+\frac{1}{k}\left(\avgdegree-\frac{\delta}{k}\right)\ones{k}\ones{k}^T$, then using Eqs. \eqref{eq:sbm_psi} and \eqref{eq:pp_matrix_powers} the survival function block matrix is given by:
\begin{equation*}
    \Psiacf_l = \exp\left(-\frac{1}{n}\left\{S_l(\delta/k)k\eyes{k}+[S_l(\avgdegree)-S_l(\delta/k)]\ones{k}\ones{k}^T\right\}\right).
\end{equation*}

\subsection{\label{sec:apdx_rdpg}Random dot product graph (RDPG)}
Consider a $k$-dimensional non-negative bounded subspace $X\subset\realnonnegk{k}$ in which the nodes are ``embedded'' by distribution $\mu(\boldsymbol{x})$. We focus on the simplest  dot product connectivity kernel here:
\begin{equation}
    \label{eq:rdpg_nu}
    \nu(\boldsymbol{x}, \boldsymbol{y}) = \frac{\beta}{n}\boldsymbol{x}^T\boldsymbol{y},
\end{equation}
where $\beta>0$ controls the mean degree.

\paragraph*{Expected degree} Using Eqs. \eqref{eq:general_degrees_deg} and \eqref{eq:rdpg_nu}, the expected degree at location $\boldsymbol{x}$ is given by:
\begin{equation}
        \label{eq:degree_rdpg}
        \meandegree(\boldsymbol{x})=n\int_{X}\nu(\boldsymbol{x},\boldsymbol{y})\,d\mu(\boldsymbol{y})=\beta\boldsymbol{x}^T\int_{X}\boldsymbol{y}\,d\mu(\boldsymbol{y})= \beta\boldsymbol{x}^T\boldsymbol{\phi},
\end{equation}
where $\boldsymbol{\phi}\triangleq\int_{X}\boldsymbol{x}\,d\mu(\boldsymbol{x})$ is the vector representing mean location of nodes in $X$. Hence, the degree of a node is given by taking the dot product of its location with the mean location of all nodes, scaled by $\beta$. Similarly, using Eq. \eqref{eq:general_degrees_mean}, the average network degree is given by:
\begin{equation}
    \label{eq:rdpg_degree_deg}
        \avgdegree=\int_{X}d(\boldsymbol{x})\,d\mu(\boldsymbol{x})=\beta\boldsymbol{\phi}^T\int_{X}\boldsymbol{x}\,d\mu(\boldsymbol{x})=\beta\boldsymbol{\phi}^T\boldsymbol{\phi}=\beta\left\lVert\boldsymbol{\phi}\right\rVert^2.
\end{equation}

\paragraph*{Approximate closed form of the GLD}Consider the approximate closed form of the conditional PMF of the GLD from Eq. \eqref{eq:spd_general_omega}, then using Eq. \eqref{eq:rdpg_nu}:
\begin{equation}
    \label{eq:rdpg_omega}
    \begin{split}
        \omegaacf_2(\boldsymbol{x},\boldsymbol{y})&=n\int_{X}\nu(\boldsymbol{x},\boldsymbol{z})\nu(\boldsymbol{z},\boldsymbol{y})\,d\mu(\boldsymbol{z})\\&=\beta\boldsymbol{x}^T\left[\beta\int_{X}\boldsymbol{z}\boldsymbol{z}^T\,d\mu(z)\right]\boldsymbol{y}= \beta\boldsymbol{x}^T\mat{\Phi}\boldsymbol{y}\\
        \implies \omegaacf_l(\boldsymbol{x},\boldsymbol{y})&=\beta\boldsymbol{x}^T\mat{\Phi}^{l-1}\boldsymbol{y},
    \end{split}
\end{equation}
for $l>0$, where $\mat{\Phi}\triangleq \beta\int_{X} \boldsymbol{x}\boldsymbol{x}^T\,d\mu(\boldsymbol{x})$ refer to the $k\times k$ matrix of second moments---which essentially encodes the covariance in space $X$ as per the measure $\mu$, and is necessarily positive semi-definite. The last implication can be shown via induction.

\paragraph*{Illustrative example: standard simplex RDPG}
For a concrete example consider $X$ to be restricted to the $(k-1)$-standard simplex in $\real^k$, i.e. $X=\left\{\boldsymbol{x}\in\real^k:\sum_{i=1}^kx_i=1, \enspace \forall i \in[k]:\ x_i\ge 0\right\}$ and $\mu$ corresponding to the Dirichlet distribution on that simplex given by the concentration vector $\boldsymbol{\alpha}\in[0,\infty)^k$. Hence a node at $\boldsymbol{x}\in[0,1]^k$, constrained to $\boldsymbol{x}^T\ones{k}=1$ where $\ones{k}$ is the vector of all-ones of size $k$, has a location in $X$ that follows the Dirichlet distribution:
\begin{equation*}
    \mu(\boldsymbol{x}) = \frac{\prod_{i=1}^k\Gamma(\alpha_i)x_i^{\alpha_i-1}}{\Gamma\left(\sum_{i=1}^k\alpha_i\right)},
\end{equation*}
where $\Gamma(\cdot)$ is the gamma function. This formalism interprets the node's location as the probability of belonging to one of $k$ communities located at the corners of the simplex---a continuous analogue of the SBM. Let $\bar\alpha\triangleq\boldsymbol\alpha^T\ones{k}$, then the mean of this distribution is the usual mean of a Dirichlet distribution \cite{kotz2004dirichlet} i.e. $\boldsymbol{\phi}=\frac{\boldsymbol\alpha}{\bar\alpha}$. Consider the scaled matrix of second moments:
\begin{equation}
    \label{eq:rdpg_delta}
    \begin{split}
        \mat{\Phi}&=\beta\int_X\boldsymbol{x}\boldsymbol{x}^Td\mu(\boldsymbol{x})\\&=n\beta\int_X(\boldsymbol{x}-\boldsymbol{\phi}+\boldsymbol{\phi})(\boldsymbol{x}-\boldsymbol{\phi}+\boldsymbol{\phi})^T\,d\mu(\boldsymbol{x})\\
        &=\beta\bigg[\int_X(\boldsymbol{x}-\boldsymbol{\phi})(\boldsymbol{x}-\boldsymbol{\phi})^T\,d\mu(\boldsymbol{x})\\&\quad+\cancel{\int_X\boldsymbol{\phi}(\boldsymbol{x}-\boldsymbol{\phi})^T\,d\mu(\boldsymbol{x})}+\cancel{\int_X(\boldsymbol{x}-\boldsymbol{\phi})\boldsymbol{\phi}^T\,d\mu(\boldsymbol{x})}+\boldsymbol{\phi}\boldsymbol{\phi}^T\bigg]\\
        &= \beta(\mat{\Sigma} +\boldsymbol{\phi}\boldsymbol{\phi}^T),
    \end{split}
\end{equation}
where $\mat{\Sigma}\triangleq\int_X(\boldsymbol{x}-\boldsymbol{\phi})(\boldsymbol{x}-\boldsymbol{\phi})^T\,d\mu(\boldsymbol{x})$ is the covariance matrix, which for the Dirichlet distribution can be written as $\mat{\Sigma}=\frac{1}{\bar\alpha^2(1+\bar\alpha)}\left[\bar\alpha\ \diag{\boldsymbol{\alpha}}-\boldsymbol{\alpha}\boldsymbol{\alpha}^T\right]$ \cite{kotz2004dirichlet}. Thus using Eq. \eqref{eq:rdpg_delta} and the expression $\boldsymbol{\phi}=\frac{\boldsymbol\alpha}{\bar\alpha}$ leads to $\mat{\Phi}=\frac{\beta}{\bar\alpha(1+\bar\alpha)}[\diag{\boldsymbol{\alpha}}+\boldsymbol{\alpha}\boldsymbol{\alpha}^T]$. Let $\avgdegree$ be the mean degree of the network, then from Eq. \eqref{eq:rdpg_degree_deg} we have $\beta=\frac{\avgdegree\bar\alpha^2}{\left\lVert\boldsymbol{\alpha}\right\rVert^2}$. Altogether, the expression for the closed-form conditional PMF of the GLD for this RDPG from Eq. \eqref{eq:rdpg_omega} is given by: 
\begin{equation*}
     \omegaacf_l(\boldsymbol{x},\boldsymbol{y})=\frac{\avgdegree\bar\alpha^2}{n\left\lVert\boldsymbol{\alpha}\right\rVert^2}\boldsymbol{x}^T\left\{\frac{\avgdegree\bar\alpha}{\left\lVert\boldsymbol{\alpha}\right\rVert^2(1+\bar\alpha)}\left[\diag{\boldsymbol{\alpha}}+\boldsymbol{\alpha}\boldsymbol{\alpha}^T\right]\right\}^{l-1}\boldsymbol{y}.
\end{equation*}
\paragraph*{Illustrative example: non-linear RDPG}
In this article we assume that the RDPG's connectivity kernel $\nu(\boldsymbol{x},\boldsymbol{y})$ is linear in the dot product $\boldsymbol{x}^T\boldsymbol{y}$. As discussed in Sec. \ref{sec:rdpg:}, any positive semi-definite kernel can be written as a dot product in some feature space. In particular, the symmetric positive semi-definite kernel $\nu_{\ge 0}(x,y)$ for node location pairs $(x, y) \in V\times V$ defines a ``feature map'' $h$ from the node space to a Hilbert space $H$, i.e. $h:V\to H$, such that $\nu_{\ge 0}(x,y)=\langle h(x),h(y)\rangle_H$, where $\langle\cdot,\cdot\rangle_H$ indicates the inner product. In other words, $\nu$ is the reproducing kernel for $H$ over the node space: $H$ is a reproducing kernel Hilbert space (RKHS) where node-wise function evaluation is given by a continuous linear functional. For some kernels---like the polynomial kernel---the RKHS is a Euclidean space equipped with the usual vector dot product, while for others---like the radial basis or ``Gaussian'' function kernel---we can derive an explicit map $g:V\to \real^k$ such that $\nu_{\ge 0}(x,y)=\langle h(x),h(y)\rangle_H\approx g(x)^Tg(y)$ \cite{rahimi2007random}. By way of example, consider the ``node2vec'' graph embedding model of Ref. \onlinecite{grover2016node2vec}, where the connectivity kernel is proportional to $\exp(\boldsymbol{x}^T\boldsymbol{y})$. This can be seen as a ``non-linear'' RDPG, i.e. an RDPG where the connectivity kernel is a non-linear function of $\boldsymbol{x}^T\boldsymbol{y}$. We make use of an explicit randomized feature map $g:\real^d\to\real^k$ such that the connectivity kernel in $\real^d$ is well-approximated by the corresponding dot product in $\real^k$ \cite{rahimi2007random}, i.e. $\nu(\boldsymbol{x},\boldsymbol{y})\approx g(\boldsymbol{x})^Tg(\boldsymbol{y})$. This allows us to apply results from the GLD framework of ``linear'' RDPGs to this ``exponential'' RDPG.

\begin{lemma}[Linearization of exponential RDPG] Let $X\subset \real^d$ be a node space with node density $\mu$ and connectivity kernel $\nu(\boldsymbol{x},\boldsymbol{y})=\frac{\beta}{n}\exp(\boldsymbol{x}^T\boldsymbol{y})$, $\beta>0$. Then there exists a map $g:\real^d\to\real^k$ which approximately linearizes this RDPG i.e. $\nu(\boldsymbol{x},\boldsymbol{y})\approx\frac{\beta}{n} g(\boldsymbol{x})^Tg(\boldsymbol{y})$. There also exist corresponding length-$k$ mean vector $\boldsymbol{\phi}_g$, $k\times k$ matrix of second moments $\mat{\Phi}_g$, and length-$k$ mean percolation vector $\boldsymbol{\rho}_g$, which can be used to evaluate the percolation probability function and the conditional PMF and survival function of the GLD for this model in $\real^k$.
\end{lemma}

\begin{proof}
First, note that the exponential kernel can be expressed as a Gaussian kernel in $\real^d$:
\begin{equation*}
    \begin{split}
        \nu(\boldsymbol{x},\boldsymbol{y})&=\frac{\beta}{n}\exp(\boldsymbol{x}^T\boldsymbol{y})=\frac{\beta}{n}\exp\left(\frac{\lVert \boldsymbol{x}\rVert^2+\lVert \boldsymbol{y}\rVert^2-\lVert \boldsymbol{x}-\boldsymbol{y}\rVert^2}{2}\right)\\&=\frac{\beta\,f(\boldsymbol{x}-\boldsymbol{y})}{n\,f(\boldsymbol{x})f(\boldsymbol{y})},
    \end{split}
\end{equation*}
where $\beta>0$, $f(\boldsymbol{z})\triangleq \exp\left(\frac{-\lVert \boldsymbol{z}\rVert^2}{2}\right)$. To obtain an explicit feature map for $\nu(\boldsymbol{x},\boldsymbol{y})$ requires an explicit map for $f(\boldsymbol{x}-\boldsymbol{y})$. For that, we use Algorithm 1 in Ref. \onlinecite{rahimi2007random} to compute ``random Fourier features'' for the Gaussian kernel $f(\boldsymbol{x}-\boldsymbol{y})$. More precisely:
\begin{enumerate}
    \item Draw $k$ i.i.d. samples $\{\boldsymbol{w}_i\}_{i=1}^k$ from the Fourier transform of the Gaussian function, which is given by the standard multivariate Gaussian distribution: $$\prob{\boldsymbol{w}}=(2\pi)^{-\frac{k}{2}}\exp\left(-\frac{\lVert\boldsymbol{w}\rVert^2}{2}\right).$$
    \item Draw $k$ i.i.d. samples $\{b_i\}_{i=1}^k$ from the uniform distribution on $[0,2\pi]$.
\end{enumerate}
Then for $\boldsymbol{x}\in\real^d$, the feature map $g(\boldsymbol{x})$ is given by:
\begin{equation}\label{eq:nonlin_rdpg_map}
    \begin{split}
    g(\boldsymbol{x})\triangleq \frac{1}{f(\boldsymbol{x})}\sqrt{\frac{2}{d}}\big(&\cos\left(\boldsymbol{w}_1^T\boldsymbol{x}+b_1\right),\cos\left(\boldsymbol{w}_2^T\boldsymbol{x}+b_2\right),\dots,\\ &\cos\left(\boldsymbol{w}_k^T\boldsymbol{x}+b_k\right)\big).
    \end{split}
\end{equation}
We can thus define a corresponding ``linear RDPG'' in $\real^k$, i.e. for nodes $\boldsymbol{x},\boldsymbol{y}\in\real^d$:
\begin{equation}
    \label{eq:nonlin_rdpg_nu}
    \nu(\boldsymbol{x},\boldsymbol{y})\approx\nu_g(\boldsymbol{x},\boldsymbol{y})\triangleq\frac{\beta}{n} g(\boldsymbol{x})^Tg(\boldsymbol{y}),
\end{equation}
where the approximation gets better for larger values of $k$ \cite{rahimi2007random}.

Recall from Sec. \ref{sec:rdpg:} that computing the degree or GLD in a linear RDPG only requires the mean vector and matrix of second moments, and not the full node distribution. This can be put to use here, by making an $m$-block SBM approximation of the given non-linear RDPG---see Appendix \ref{sec:apdx_general}---denoted here by $(\mat{Z},\boldsymbol{\pi}, \mat{B})$. Since the random Fourier feature map is smooth with high probability \cite{rahimi2007random} the SBM approximation can be used, and we can estimate the random feature mapping of the grid points $\mat{Z}$ to $\real^k$ given by the $k\times m$ matrix $\mat{Z}_g$ whose $i^\mathrm{th}$ column is given by the vector $g(\idx{\mat{Z}}{:i})$. Consequently, mean vector $\boldsymbol{\phi}_g$ and matrix of second moments $\mat{\Phi}_g$ in $\real^k$, and the mean percolation vector from Eq. \ref{eq:rho_rdpg}, are given by:
\begin{subequations}
\label{eq:nonlin_rdpg_stats}
\begin{align}
    \boldsymbol\phi_g&\approx \mat{Z}_g\boldsymbol{\pi}\\
    \mat{\Phi}_g &\approx \beta\,\mat{Z}_g\,\diag{\boldsymbol{\pi}}\mat{Z}_g^T\\
    \boldsymbol{\rho}_g&\approx\boldsymbol{\phi}_g-\mat{Z}_g\,\diag{\boldsymbol\pi}\exp\left(-\beta\, \mat{Z}_g^T\boldsymbol{\rho}_g\right)
\end{align}
\end{subequations}
Therefore, geodesic statistics for any pair of node locations $\boldsymbol{x}\in\real^d,\boldsymbol{y}\in\real^d$ can be computed by mapping them to $\real^k$ using Eq. \eqref{eq:nonlin_rdpg_map}, and then making use of the connectivity kernel in Eq. \eqref{eq:nonlin_rdpg_nu}, and of Eq. \eqref{eq:nonlin_rdpg_stats} to extract statistics from the corresponding linear RDPG.
\end{proof}
\begin{figure}
    \centering
    \includegraphics[width=\columnwidth]{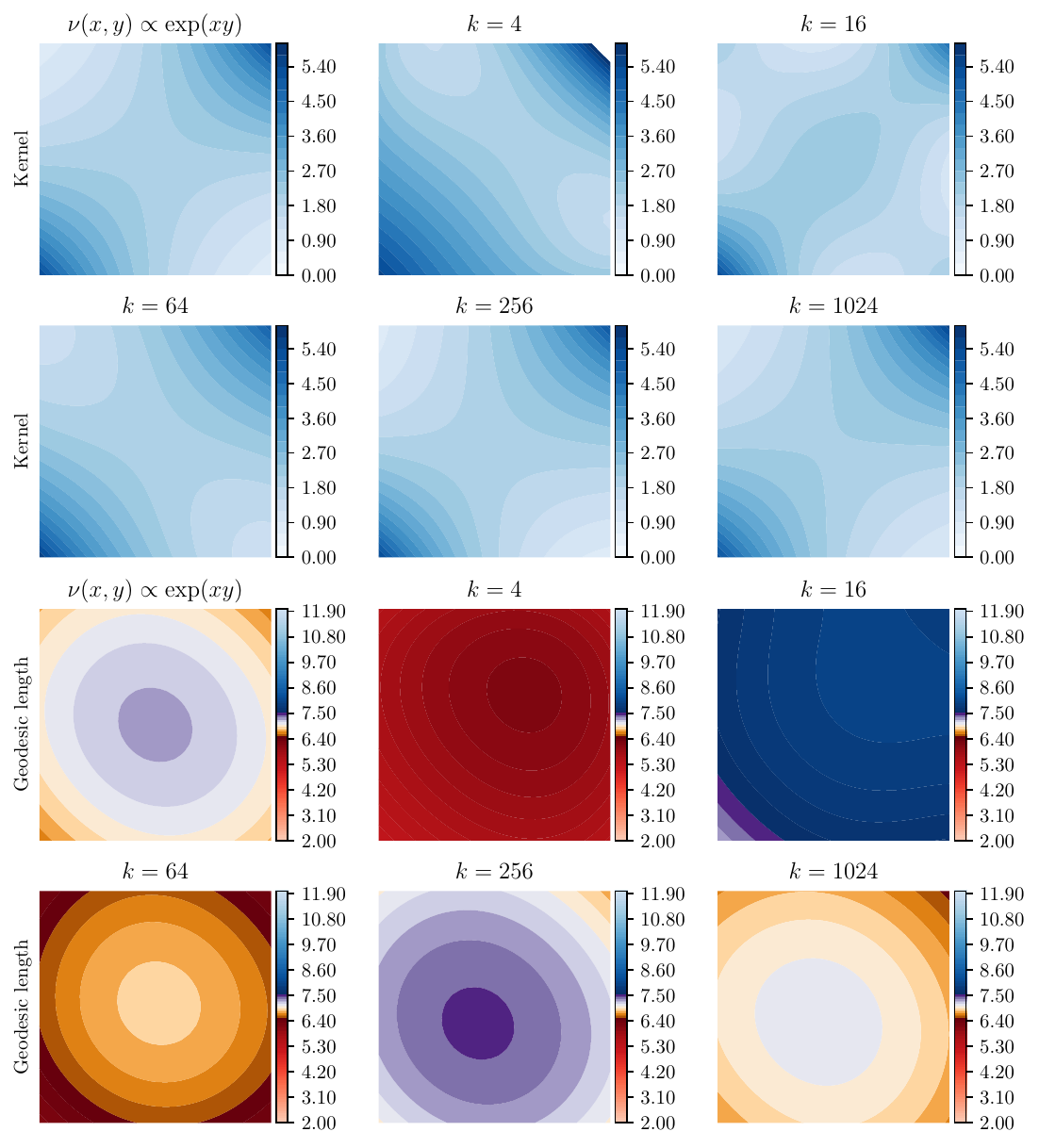}
    \caption{Linearization of an exponential RDPG using random Fourier features \cite{rahimi2007random} with varying number of features $k\in\{4,16,64,256,1024\}$. The exponential kernel $\nu(x,y)=\frac{\beta}{n}\exp(xy)$ is defined on $V=[-1,1]$, with $\beta=2$, $n=512$, and where nodes follow a scaled beta distribution with parameters $(0.2,0.8)$. Upper (blue) subplots correspond to the kernel estimate $\nu_g(x,y)$ on $V\times V$ from Eq. \eqref{eq:nonlin_rdpg_nu} using $k$ features, with the last subplot indicating the original kernel $\nu(x,y)$. Similarly, lower (purple) subplots indicate the expected geodesic length $\avg{\lambda_{xy}}$ computed from the approximate closed form of the geodesic length distribution for a linear RDPG given by Eq. \eqref{eq:spd_rdpg}---making use of Eqs. \eqref{eq:nonlin_rdpg_map}, \eqref{eq:nonlin_rdpg_stats}---with the last subplot using the original kernel $\nu(x,y)$ to obtain the expected geodesic length from the approximate closed form for a general random graph model as given by Eq. \eqref{eq:spd_analytic_general_eig_uncorrected}, obtained via iterative numerical integration on $V$. Despite stochasticity in the sampling of Fourier features, larger values of $k$ evidently improve the approximation, with the linear RDPG having $k=1024$ providing the closest estimates.}
    \label{fig:rdpg_rff}
\end{figure}
In Fig. \ref{fig:rdpg_rff}, we show for an exponential kernel defined on $[-1,1]$ both the kernel estimate $\nu_g(x,y)$ and corresponding expected geodesic length $\avg{\lambda^g_{xy}}$ between node locations $(x,y)\in[-1,1]\times[-1,1]$, for varying number of random Fourier features $k$. Evidently, as the number of features increases, the ``linear RDPG approximation'' significantly improves, asymptotically matching up to the true connectivity kernel, and corresponding expected geodesic lengths estimated via numerical integration. For other non-linear kernels, the same method follows by using their respective Fourier transforms.

\subsection{\label{sec:apdx_grgg}Gaussian random geometric graph (RGG)}
Consider the $k$-dimensional Euclidean space $\real^k$, wherein nodes are distributed according to a multivariate Gaussian distribution with mean vector $\boldsymbol{m}$ and (full-rank) covariance matrix $\mat{\Sigma}$:
\begin{equation*}
    \mu(\boldsymbol{x})=\frac{1}{\sqrt{(2\pi)^k|\mat{\Sigma}|}}\exp\left(-\frac{1}{2}(\boldsymbol{x}-\boldsymbol{m})^T\mat{\Sigma}^{-1}(\boldsymbol{x}-\boldsymbol{m})\right),
\end{equation*}
where $|\mat{\Sigma}|$ is the determinant of $\mat{\Sigma}$, and (undirected) edges are added according to a Gaussian connectivity kernel with the (positive-definite and symmetric) scale matrix $\mat{R}$, as given by Eq. \eqref{eq:grgg_nu}. Succinctly, we refer to this as the $(\boldsymbol{m}, \mat{\Sigma}, \mat{R})$ Gaussian RGG.

\paragraph*{Sufficiency of the standard Gaussian RGG}First, we show that this Gaussian RGG can be uniquely mapped to a ``standard Gaussian RGG'' given by $(\zeros{k}, \eyes{k}, \mat{S})$ where $\zeros{k}$ is the vector of zeros of size $k$ and $\eyes{k}$ is the $k\times k$ identity matrix, i.e. where the node distribution is standard multivariate Gaussian. Since $\mat{\Sigma}$ is a full-rank covariance matrix, it must be real, symmetric, and with positive eigenvalues. Let us express it via its eigendecomposition: $\mat{\Sigma}=\mat{U\Lambda U}^T$ where $\mat{U}$ is an orthonormal matrix with eigenvectors stacked column-wise, and $\mat{\Lambda}$ is the diagonal matrix with corresponding eigenvalues on the diagonal. Then, define the affine whitening transform $f:\real^k\to\real^k$ $$f(\boldsymbol{x})\triangleq\mat{\Lambda}^{-\frac{1}{2}}\mat{U}^T(\boldsymbol{x}-\boldsymbol{m}),$$
which maps to values distributed by $\mathcal{N}(\zeros{k}, \eyes{k})$. We also require a new connectivity kernel $\nu_w(\boldsymbol{x},\boldsymbol{y})$ for the whitened space, which from Eq. \eqref{eq:grgg_nu} will be given by substituting $\boldsymbol{x}$ with $f^{-1}(\boldsymbol{x})$ to obtain 
\begin{equation*}
    \begin{split}
         \nu_w(\boldsymbol{x},\boldsymbol{y})&=\frac{\beta}{n}\exp\bigg(-\frac{1}{2}[f^{-1}(\boldsymbol{x})-f^{-1}(\boldsymbol{y})]^T\mat{R}^{-1}\\&\quad
         \times[f^{-1}(\boldsymbol{x})-f^{-1}(\boldsymbol{y})]\bigg)\\
         &=\frac{\beta}{n}\exp\left(-\frac{1}{2}(\boldsymbol{x}-\boldsymbol{y})^T\mat{\Lambda}^\frac{1}{2}\mat{U}^T\mat{R}^{-1}\mat{U\Lambda}^\frac{1}{2}(\boldsymbol{x}-\boldsymbol{y})\right),
    \end{split}
\end{equation*}
which leads to a standard Gaussian RGG given by $(\boldsymbol{0},\eyes{k},\mat{\Lambda}^{-\frac{1}{2}}\mat{U}^T\mat{RU\Lambda}^{-\frac{1}{2}})$. We remark that if $\mat{R}$ commutes with $\mat{\Sigma}$, then they have the same eigenvectors---i.e. connectivity directions are oriented along the same axes as the node distribution---and we can write $\mat{R}=\mat{U\Gamma U}^T$, i.e. we obtain a diagonal scale matrix $\mat{\Lambda\Gamma}$. This implies that it is sufficient to study a standard Gaussian RGG to obtain geodesic statistics for general Gaussian RGGs, as we do in this article. While we maintain the connectivity kernel in Eq. \eqref{eq:grgg_nu}, we use the standard normal node distribution from Eq. \eqref{eq:grgg_mu}. It will be useful to state the following multivariate Gaussian integral:
\begin{equation}
    \label{eq:gauss_integral}
        \int_{\real^k}\exp\left(-\frac{\boldsymbol{x}^T\mat{C}\boldsymbol{x}}{2}+\boldsymbol{v}^T\boldsymbol{x}\right)d\boldsymbol{x}=\frac{(2\pi)^\frac{k}{2}}{\sqrt{|\mat{C}|}}\exp\left(\frac{\boldsymbol{v}^T\mat{C}^{-1}\boldsymbol{v}}{2}\right).
\end{equation}

\paragraph*{Expected degree} Using Eqs. \eqref{eq:general_degrees_deg}, \eqref{eq:grgg_mu}, \eqref{eq:grgg_nu} the expected degree at location $\boldsymbol{x}$ is given by the following Gaussian curve:
\begin{equation}
    \label{eq:grgg_degree}
    \begin{split}
        \meandegree(\boldsymbol{x})&=n\int_{\real^k}\nu(\boldsymbol{x},\boldsymbol{y})\,d\mu(\boldsymbol{y})
        \\&=\frac{\beta}{(2\pi)^{\frac{k}{2}}}\int_{\real^k}\exp\bigg(-\frac{1}{2}\big[(\boldsymbol{x}-\boldsymbol{y})^T\mat{R}^{-1}(\boldsymbol{x}-\boldsymbol{y})\\&\quad+\boldsymbol{y}^T\boldsymbol{y}\big]\bigg)d\boldsymbol{y}\\
        &=\frac{\beta}{(2\pi)^{\frac{k}{2}}}\exp\left(-\frac{\boldsymbol{x}^T\mat{R}^{-1}\boldsymbol{x}}{2}\right)\\&\quad\times\int_{\real^k}\exp\left(-\frac{\boldsymbol{y}^T(\eyes{k}+\mat{R}^{-1})\boldsymbol{y}}{2}+\boldsymbol{x}^T\mat{R}^{-1}\boldsymbol{y}\right)d\boldsymbol{y}\\
        &=\frac{\beta}{\sqrt{|\eyes{k}+\mat{R}^{-1}|}}\\&\quad\times\exp\left(-\frac{\boldsymbol{x}^T[\mat{R}^{-1}-\mat{R}^{-1}(\eyes{k}+\mat{R}^{-1})^{-1}\mat{R}^{-1}]\boldsymbol{x}}{2}\right)
        \\&=\frac{\beta}{\sqrt{|\eyes{k}+\mat{R}^{-1}|}}\exp\left(-\frac{\boldsymbol{x}^T\mat{R}^{-1}(\eyes{k}+\mat{R}^{-1})^{-1}\boldsymbol{x}}{2}\right)\\
        &=\frac{\beta}{\sqrt{|\eyes{k}+\mat{R}^{-1}|}}\exp\left(-\frac{\boldsymbol{x}^T(\eyes{k}+\mat{R})^{-1}\boldsymbol{x}}{2}\right),
    \end{split}
\end{equation}
where in the fourth equality we use Eq. \eqref{eq:gauss_integral}, and in the fifth equality we use a special case of the Woodbury matrix identity (WMI): $(\eyes{k}+\mat{M})^{-1}=\eyes{k}-(\eyes{k}+\mat{M})^{-1}\mat{M}$. Similarly, using Eqs. \eqref{eq:general_degrees_mean}, \eqref{eq:grgg_degree}, \eqref{eq:grgg_mu} the average degree of the whole network is given by
\begin{equation}
    \label{eq:grgg_degree_deg}
        \begin{split}
            \avgdegree&=\int_{\real^k}d(\boldsymbol{x})\,d\mu(\boldsymbol{x})=\frac{\beta}{(2\pi)^{\frac{k}{2}}\sqrt{|\eyes{k}+\mat{R}^{-1}|}}\\&\quad\times\int_{\real^k}\exp\left(-\frac{\boldsymbol{x}^T[\eyes{k}+(\eyes{k}+\mat{R})^{-1}]\boldsymbol{x}}{2}\right)d\boldsymbol{x}
        \\&=\frac{\beta}{\sqrt{|\eyes{k}+\mat{R}^{-1}||\eyes{k}+(\eyes{k}+\mat{R})^{-1}|}}.
        \end{split}
\end{equation}
If we assume infinitely large connection scales, i.e. $\mat{R}^{-1}\to \mat{0}$, this results in $\avgdegree=\beta$, which can be seen as the usual ER graph where all spatial structure is lost, since nodes connect to any other node with the same probability $\frac{\avgdegree}{n}$.

\paragraph*{Approximate closed form of the GLD} Next, consider the conditional PMF of the GLD from Eq. \eqref{eq:spd_general_omega}. We apply an ansatz: by symmetry of this function, it will take a form wherein terms of $\left\lVert\boldsymbol{x}-\boldsymbol{y}\right\rVert^2$ are weighted by some functions of $l$ as:
\begin{equation}
    \label{eq:grgg_omega_ansatz}
    \begin{split}
    \omegaacf_l(\boldsymbol{x},\boldsymbol{y})=&c_l\exp\bigg(-\frac{1}{2}\big(\boldsymbol{x}^T\mat{R}^{-1}\mat{U}_l\boldsymbol{x}+\boldsymbol{y}^T\mat{R}^{-1}\mat{U}_l\boldsymbol{y}\\&-2\boldsymbol{x}^T\mat{R}^{-1}\mat{W}_l\boldsymbol{y}\big)\bigg),
    \end{split}
\end{equation}where $\mat{U}_l,\mat{W}_l$ are symmetric matrices. Then, using Eqs. \eqref{eq:spd_general_omega}, \eqref{eq:grgg_mu}, \eqref{eq:grgg_nu}, we get 
\begin{equation*}
    \begin{split}
        \omegaacf_{l+1}(\boldsymbol{x},\boldsymbol{y})=\ & n\int_{\real^k}\omegaacf_l(\boldsymbol{x},\boldsymbol{z})\nu(\boldsymbol{z},\boldsymbol{y})\,d\mu(\boldsymbol{z})\\
        =\ &\frac{c_l\beta}{(2\pi)^{\frac{k}{2}}}\int_{\real^k}\exp\bigg(-\frac{1}{2}\times\big(\boldsymbol{x}^T\mat{R}^{-1}\mat{U}_l\boldsymbol{x}\\&+\boldsymbol{z}^T\mat{R}^{-1}\mat{U}_l\boldsymbol{z}-2\boldsymbol{x}^T\mat{R}^{-1}\mat{W}_l\boldsymbol{z}+\boldsymbol{z}^T\mat{R}^{-1}\boldsymbol{z}\\&+\boldsymbol{y}^T\mat{R}^{-1}\boldsymbol{y}-2\boldsymbol{y}^T\mat{R}^{-1}\boldsymbol{z}+\boldsymbol{z}^T\boldsymbol{z}\big)\bigg)d\boldsymbol{z}\\
        =\ &\frac{c_l\beta}{(2\pi)^{\frac{k}{2}}}\exp\left(-\frac{1}{2}\left(\boldsymbol{x}^T\mat{R}^{-1}\mat{U}_l\boldsymbol{x}+\boldsymbol{y}^T\mat{R}^{-1}\boldsymbol{y}\right)\right)\\
        &\times\int_{\real^k}\exp\bigg(-\frac{1}{2}\left[\boldsymbol{z}^T(\eyes{k}+\mat{R}^{-1}\mat{U}_l+\mat{R}^{-1})\boldsymbol{z}\right]\\&-\left(\boldsymbol{x}^T\mat{R}^{-1}\mat{W}_l+\boldsymbol{y}^T\mat{R}^{-1}\right)\boldsymbol{z}\bigg)d\boldsymbol{z}\\
        =\ &\frac{c_l\beta}{\sqrt{|\mat{V}_l|}}\exp\bigg(-\frac{1}{2}\Big[\boldsymbol{x}^T\mat{R}^{-1}\mat{U}_l\boldsymbol{x}+\boldsymbol{y}^T\mat{R}^{-1}\boldsymbol{y}\\&-\left(\boldsymbol{x}^T\mat{R}^{-1}\mat{W}_l+\boldsymbol{y}^T\mat{R}^{-1}\right)\\&\times\mat{V}_l^{-1}\left(\mat{W}_l\mat{R}^{-1}\boldsymbol{x}+\mat{R}^{-1}\boldsymbol{y}\right)\Big]\bigg)\\
        &\textrm{(where }\mat{V}_l\triangleq \eyes{k}+\mat{R}^{-1}+\mat{R}^{-1}\mat{U}_l\textrm{)}\\
        =\ &\frac{c_l\beta}{\sqrt{|\mat{V}_l|}}\exp\bigg(-\frac{1}{2}\\&\times\Big[\boldsymbol{x}^T\mat{R}^{-1}\left(\mat{U}_l-\mat{W}_l\mat{V}_l^{-1}\mat{W}_l\mat{R}^{-1}\right)\boldsymbol{x}\\&+\boldsymbol{y}^T\mat{R}^{-1}\left(\eyes{k}-\mat{V}_l^{-1}\mat{R}^{-1}\right)\boldsymbol{y}\\&-2\boldsymbol{x}^T\mat{R}^{-1}\mat{W}_l\mat{V}_l^{-1}\mat{R}^{-1}\boldsymbol{y}\Big]\bigg)\\
        \triangleq\ &c_{l+1}\exp\bigg(-\frac{1}{2}\big(\boldsymbol{x}^T\mat{R}^{-1}\mat{U}_{l+1}\boldsymbol{x}\\&+\boldsymbol{y}^T\mat{R}^{-1}\mat{U}_{l+1}\boldsymbol{y}-2\boldsymbol{x}^T\mat{R}^{-1}\mat{W}_{l+1}\boldsymbol{y}\big)\bigg),
    \end{split}
\end{equation*}
where we have applied Eq. \eqref{eq:gauss_integral} to obtain the fourth equality. Given the ansatz in Eq. \eqref{eq:grgg_omega_ansatz} to be satisfied for $l+1$, and using $\mat{R}\mat{V}_l=\eyes{k}+\mat{R}+\mat{U}_l$, we can write the following set of recursive equations:
\begin{subequations}
\label{eq:grgg_omega_recursion}
\begin{align}
    \mat{U}_{l+1} &= \eyes{k}-\ (\eyes{k}+\mat{R}+\mat{U}_l)^{-1},\\
    \label{eq:grgg_omega_recursion_v}
    \mat{W}_{l+1} &= \mat{W}_l(\eyes{k}+\mat{R}+\mat{U}_l)^{-1},\\
    \label{eq:grgg_omega_recursion_c}
    c_{l+1} &= c_l \beta|\eyes{k}+\mat{R}^{-1}+\mat{R}^{-1}\mat{U}_l|^{-\frac{1}{2}}.
\end{align}
\end{subequations}
For the approximate closed-form GLD: $\omegaaf_1(\boldsymbol{x},\boldsymbol{y})=\nu(\boldsymbol{x},\boldsymbol{y})$. Combining with the ansatz in Eq. \eqref{eq:grgg_omega_ansatz} gives the base cases $\mat{U}_l=\mat{W}_l=\eyes{k}$ and $c_1=\beta$, which can then be used to solve Eq. \eqref{eq:grgg_omega_recursion}, and thus provide the expression for conditional PMF from Eq. \eqref{eq:grgg_omega_ansatz}.

\paragraph*{Highly spatial Gaussian RGG} Consider a highly ``spatial'' Gaussian RGG, wherein the connection scales are much smaller than the variance, i.e. $\mat{R}\to 0$. This simplifies solving Eq. \eqref{eq:grgg_omega_recursion}, since as $\mat{R} \to 0$, from Eq. \eqref{eq:grgg_omega_recursion} we get:
\begin{equation*}
    \begin{split}
        \mat{U}_{l+1}&\to \eyes{k}-(\eyes{k}+\mat{U}_l)^{-1}=(\eyes{k}+\mat{U}_l)^{-1}\mat{U}_l,\\
        \mat{W}_{l+1}&\to \mat{W}_l(\eyes{k}+\mat{U}_l)^{-1},\\
        c_{l+1}&\to c_l\beta\sqrt{\frac{|\mat{R}|}{|\eyes{k}+\mat{U}_l|}}.
    \end{split}
\end{equation*}
where in the first limit we make use of the WMI. Putting in the base cases, we obtain closed forms:
\begin{equation*}
    \begin{split}
        \mat{U}_{l}&\to \frac{1}{l}\eyes{k},\\
        \mat{W}_{l}&\to \frac{1}{l}\eyes{k},\\
        c_{l}&\to \beta\left(\beta\sqrt{|\mat{R}|}\right)^{l-1}l^{-\frac{k}{2}},
    \end{split}
\end{equation*}
Note from Eq. \eqref{eq:grgg_degree_deg} that the mean degree approaches $$\avgdegree\to \beta\sqrt{|\mat{R}|}2^{-\frac{k}{2}}.$$ Putting in Eq. \eqref{eq:grgg_omega_ansatz} yields the approximate closed form of the conditional PMF of the GLD in Eq. \eqref{eq:spd_grgg_homophily}.

\subsection{\label{sec:apdx_graphons}Multiplicative graphon}

\paragraph*{Expected degree} From Eqs. \eqref{eq:general_degrees_deg} and \eqref{eq:general_degrees_mean}, and using the definition in Eq. \eqref{eq:mult_graphon_stats_zeta}, the expected degree at $x$ and expected degree are respectively given by
\begin{subequations}
    \label{eq:degree_r1g_overall}
    \begin{align}
    \label{eq:degree_r1g}
    \meandegree(x) &= nf(x)\int_0^1f(y)\,dy=\sqrt{n}\zeta f(x),\\
     \label{eq:degree_net_r1g}
     \avgdegree&=\sqrt{n}\zeta\int_0^1f(x)\,dx=\zeta^2.
    \end{align}
\end{subequations}
From Eq. \eqref{eq:degree_r1g}, it is evident that $f(x)\propto \meandegree(x)$ encodes the expected degree at $x$, rendering multiplicative graphons as equivalent to canonical degree-configuration models that lack any modularity structure \cite{klimm2021modularity}. However, they can still capture degree-related properties of real-world graphs, like ``scale-free'' networks showcasing power law degree distributions. Using Eq. \eqref{eq:degree_r1g_overall}, we can rewrite $f(x)$, $\zeta$ and $\eta$ from Eq. \eqref{eq:mult_graphon_stats} in terms of degree statistics as:
\begin{subequations}
\label{eq:mult_graphon_stats_deg}
\begin{align}
     \label{eq:mult_graphon_foo_deg}
     f(x)&=\frac{\meandegree(x)}{\sqrt{n\avgdegree}},\\
    \label{eq:mult_graphon_stats_zeta_deg}
    \zeta&=\sqrt{\avgdegree},\\
    \label{eq:mult_graphon_stats_eta_deg_2}
    \eta&=\frac{\expectwrt{\mu}{\meandegree(x)^2}}{\avgdegree}.
\end{align}
\end{subequations}
Recall from Eq. \eqref{eq:degree_distribution} that the degree at $x$ is asymptotically Poisson distributed with rate $\meandegree(x)$. Then the second moment of the degree distribution is given by the law of total expectation:
\begin{equation}
    \label{eq:degree_distribution_second_moment}
    \begin{split}&\avg{\degree^2}=\expectwrt{\mu}{\avg{\degree(x)^2}}\approx\expectwrt{\mu}{\meandegree(x)^2+\meandegree(x)}
    \\\implies&\expectwrt{\mu}{\meandegree(x)^2}=\avg{\degree^2}-\avgdegree,
    \end{split}
\end{equation}
where we apply the definition of mean degree from Eq. \eqref{eq:general_degrees_mean}. Using Eqs. \eqref{eq:degree_distribution_second_moment}, \eqref{eq:mult_graphon_stats_eta_deg_2} yields $\eta$ in terms of the first and second moments of the degree distribution:
\begin{equation*}
    \eta=\frac{\avg{\degree^2}}{\avgdegree}-1,
\end{equation*}
which is the RHS in Eq. \eqref{eq:mult_graphon_stats_eta}.

\paragraph*{Scale-free graphon} Consider the real-interval $[h,1]$ where $0<h\ll 1$ on which nodes are distributed uniformly. (Conventionally, nodes are distributed uniformly on $[0,1]$ in a graphon, but we curtail the interval to prevent unbounded node degrees in this ``scale-free'' graphon.)
\begin{equation*}
    \mu(x)=
    \begin{cases}
    \frac{1}{1-h} &\mbox{if }h\le x \le 1,\\
    0 &\mbox{otherwise.}
    \end{cases}
\end{equation*}
Next, (undirected) edges are added according to the connectivity kernel:
\begin{equation*}
    \nu(x,y)=\frac{\beta}{n}\left(\frac{xy}{h^2}\right)^{-\alpha},
\end{equation*}
where $0<\alpha\le 1$ and $\beta>0$. In terms of the multiplicative graphons definition from Sec. \ref{sec:graphons}:$$\nu(x,y)=W^\times_n(x,y)\triangleq f(x)f(y),$$ where: 
\begin{equation}
    \label{eq:sfg_foo}
    f(x)=\sqrt{\frac{\beta}{n}}\left(\frac{x}{h}\right)^{-\alpha}.   
\end{equation}
Then we know that the expected degrees and GLD are given by making use of $\zeta\triangleq\frac{\sqrt{n}}{1-h}\int_h^1f(x)\,dx$:
\begin{equation}
    \label{eq:sfg_zeta}
    \begin{split}
        \zeta&=\frac{\sqrt{n}}{1-h}\int_h^1f(x)\,dx=\frac{\sqrt{\beta}h^\alpha}{1-h}\int_h^1x^{-\alpha}\,dx\\&=\begin{cases}
            \frac{\sqrt{\beta}h}{1-h}\left.\log x\right\vert_h^1 &\mbox{if }\alpha=1,\\
            \frac{\sqrt{\beta}h^\alpha}{1-h}\left.\frac{x^{1-\alpha}}{1-\alpha}\right\vert_h^1 &\mbox{otherwise}
        \end{cases}\\
        &= \begin{cases}
            \sqrt{\beta}\frac{h\log h^{-1}}{1-h} &\mbox{if }\alpha=1,\\
            \sqrt{\beta}\frac{h^\alpha-h}{(1-h)(1-\alpha)} &\mbox{otherwise,}
        \end{cases}
    \end{split}
\end{equation}
and $\eta\triangleq \frac{n}{1-h}\int_h^1f(x)^2\,dx$:
\begin{equation}
    \label{eq:sfg_eta_apdx}
    \begin{split}
        \eta&=\frac{n}{1-h}\int_h^1f(x)^2\,dx=\frac{\beta h^{2\alpha}}{1-h}\int_h^1x^{-2\alpha}\,dx\\&=\begin{cases}
            \frac{\beta h}{1-h}\left.\log x\right\vert_h^1 &\mbox{if }\alpha=1/2,\\
            \frac{\beta h^{2\alpha}}{1-h}\left.\frac{x^{1-2\alpha}}{1-2\alpha}\right\vert_h^1 &\mbox{otherwise}
        \end{cases}\\
        &= \begin{cases}
            \beta\frac{h\log h^{-1}}{1-h} &\mbox{if }\alpha=1/2,\\
            \beta\frac{h^{2\alpha}-h}{(1-h)(1-2\alpha)} &\mbox{otherwise.}
        \end{cases}
    \end{split}
\end{equation}

\paragraph*{Scale-free graphon: expected degree distribution}Using Eqs. \eqref{eq:degree_r1g}, \eqref{eq:sfg_foo}, \eqref{eq:sfg_zeta} the expected degree at location $x$ is given by:
\begin{equation}
    \label{eq:sfg_degree}
    \begin{split}
        \meandegree(x)&=\sqrt{n}\zeta f(x)= \begin{cases}
            x^{-1}\beta\frac{h^2\log h^{-1}}{1-h} &\mbox{if }\alpha=1,\\
            x^{-\alpha}\beta\frac{h^{2\alpha}-h^{1+\alpha}}{(1-h)(1-\alpha)} &\mbox{otherwise}
        \end{cases}\\
        &\triangleq c(\alpha)x^{-\alpha},
    \end{split}
\end{equation}
where $c(\alpha)$ is a constant independent of $x$. To see how expected degrees are distributed by a power law, consider Eq. \eqref{eq:sfg_degree} as a monotonically decreasing transformation of random variable $x$ to $\meandegree(x)$.
\begin{equation}
    \label{eq:sfg_degree_powerlaw}
    \begin{split}
        \prob{\meandegree(x)\le k}&=\prob{x\ge \meandegree^{-1}(k)}=1-\prob{x\le \meandegree^{-1}(k)}\\
        &=\begin{cases}
            1 &\mbox{if }\meandegree^{-1}(k)<h,\\
            \frac{1-\meandegree^{-1}(k)}{1-h}&\mbox{if }\meandegree^{-1}(k)\in[h,1],\\
            0 &\mbox{if }\meandegree^{-1}(k)>1
        \end{cases}\\
        &=\begin{cases}
            1 &\mbox{if }k>c(\alpha)h^{-\alpha},\\
            \frac{1-\left(\frac{k}{c(\alpha)}\right)^{-\frac{1}{\alpha}}}{1-h}&\mbox{if }k\in[c(\alpha),c(\alpha)h^{-\alpha}],\\
            0 &\mbox{if }k<c(\alpha)
        \end{cases}\\
        \implies\prob{\meandegree(x)=k}&=\begin{cases}
            \frac{\theta-1}{1-h}c(\alpha)^{\theta-1}k^{-\theta}&k\in[c(\alpha),c(\alpha)h^{-\alpha}],\\
            0 &\mbox{otherwise},
        \end{cases}
    \end{split}
\end{equation}
i.e. $\meandegree\sim \meandegree^{-\theta}$ where $\theta\triangleq\left(1+\frac{1}{\alpha}\right)$ is the power law exponent which can take values $\theta\in[2,\infty)$. This power law distribution is over a bounded interval $[c(\alpha),c(\alpha)h^{-\alpha}]$ unlike more conventional ones which are unbounded from above. We emphasize that this is the distribution of the \emph{expectation} of degrees in the network and not, strictly speaking, the empirical degree distribution. The network generation process of adding conditionally independent edges will asymptotically impose a Poisson degree distribution at a given location $x$, whose expectation is given by $\meandegree(x)$, which follows the power law as shown in Eq. \eqref{eq:sfg_degree_powerlaw}. But it can be shown that the empirical degree distribution will also follow this power law behavior.

\paragraph*{Scale free graphon: empirical degree distribution}Let $\degree(x)$ be the degree of node at $x$, then $\degree(x)\sim\poisson{\meandegree(x)}$. If $\degree$ represent the degree of any node in the network, then its distribution is given by:
\begin{equation}
    \label{eq:sfg_degree_true}
    \begin{split}
        \prob{\degree=k}&=\int_h^1\prob{\degree(x)=k}\,d\mu(x)\\&=\frac{1}{1-h}\int_h^1\frac{\left[c(\alpha)x^{-\alpha}\right]^k\exp\left(-c(\alpha)x^{-\alpha}\right)}{k!}\,dx\\
        &=\frac{c(\alpha)^{\frac{1}{\alpha}}}{\alpha(1-h)k!}\bigg[\Gamma\left(k-\frac{1}{\alpha},c(\alpha)\right)\\&\quad-\Gamma\left(k-\frac{1}{\alpha},c(\alpha)h^{-\alpha}\right)\bigg]\quad\qquad(\mbox{if }k>\alpha^{-1})\\
        &=\frac{\theta-1}{1-h}c(\alpha)^{\theta-1}\\&\quad\times\frac{\Gamma\left(k+1-\theta,c(\alpha)\right)-\Gamma\left(k+1-\theta,c(\alpha)h^{-\alpha}\right)}{\Gamma(k+1)}.
    \end{split}
\end{equation}
Now, assuming that $c(\alpha)\to 0$ and $c(\alpha)h^{-\alpha}\to \infty$ i.e. the interval bounding the expected degree $\meandegree(x)$ is very wide:
\begin{equation}
    \label{eq:sfg_degree_true_apx}
    \begin{split}
        \prob{\degree=k}&\approx\frac{\theta-1}{1-h}c(\alpha)^{\theta-1}\frac{\Gamma\left(k+1-\theta\right)}{\Gamma(k+1)}\\&\approx\frac{\theta-1}{1-h}c(\alpha)^{\theta-1}\frac{(k-\theta)^{k+\frac{1}{2}-\theta}}{k^{k+\frac{1}{2}}}e^\theta\\
        &=\frac{\theta-1}{1-h}c(\alpha)^{\theta-1}\left(1-\frac{\theta}{k}\right)^{k+\frac{1}{2}-\theta}e^\theta k^{-\theta}\\&\approx\frac{\theta-1}{1-h}c(\alpha)^{\theta-1}e^{\frac{\theta}{k}\left(\theta-\frac{1}{2}\right)} k^{-\theta}\\
        &\approx\frac{\theta-1}{1-h}c(\alpha)^{\theta-1}k^{-\theta}.
    \end{split}
\end{equation}
i.e. $\degree\sim \degree^{-\theta}$ where in the second approximation we apply Stirling's formula for the Gamma function i.e. $\Gamma(z+1)\approx\sqrt{2\pi z}\left(\frac{z}{e}\right)^z$; in the last two approximations we make a stronger use of $k\gg\theta$. These assumptions make it apparent that the power law holds for the empirical degree distribution in this ``scale-free'' graphon for large degrees when a broad degree distribution is possible, and the form obtained in Eq. \eqref{eq:sfg_degree_true_apx} matches that of the expected degree distribution in Eq. \eqref{eq:sfg_degree_powerlaw}, as depicted in Fig. \ref{fig:sfg_degree_distribution}. Also, using Eqs. \eqref{eq:degree_net_r1g}, \eqref{eq:sfg_zeta} the average degree of the whole network is given by
\begin{equation}
    \label{eq:sfg_degree_deg}
        \avgdegree = \zeta^2
        =\begin{cases}
            \beta\left(\frac{h\log h^{-1}}{1-h}\right)^2 &\mbox{if }\alpha=1,\\
            \beta\left[\frac{h^{\alpha}-h}{(1-h)(1-\alpha)}\right]^2 &\mbox{otherwise}.
        \end{cases}
\end{equation}
\begin{figure}
    \centering
    \includegraphics[width=\columnwidth]{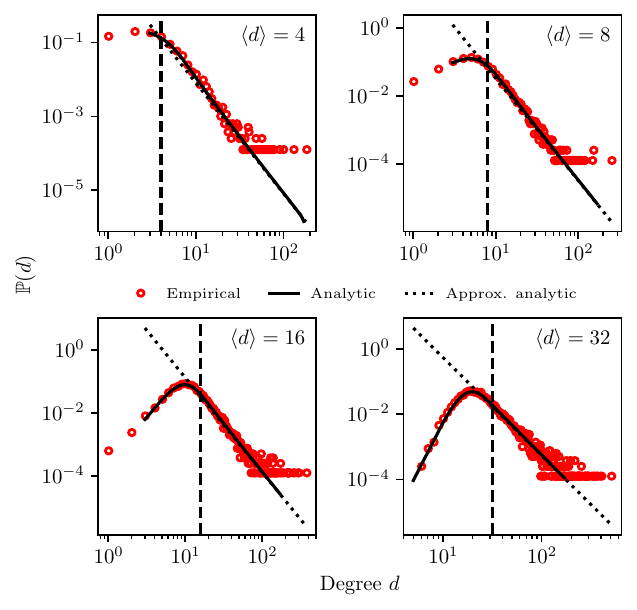}
    \caption{Degree distribution for ``scale-free'' graphons with power law exponent $\theta=3$ (i.e. $\alpha=\frac{1}{2}$), corresponding to BA graphs, for $n=8192$ nodes, free-parameter $\beta=1$, with varying mean degree $\avgdegree$. Markers indicate empirical distribution from a single network sample, while the solid line indicates the analytic estimate from Eq. \eqref{eq:sfg_degree_true}, and the dotted line is the approximate analytic estimate from Eq. \eqref{eq:sfg_degree_true_apx}. In each subplot, the power law is apparent for degrees above a certain cut-off---evidently above $\avgdegree$ indicated by the dashed line---but these are not \emph{pure} power law behaviors.} 
    \label{fig:sfg_degree_distribution}
\end{figure}

\paragraph*{Special cases: BA and highly scale-free graphons}
Here, we consider some special cases of the ``scale-free'' graphon, as described in Sec. \ref{sec:graphons}. It will be useful to define the solution to the following equation:
\begin{equation}\label{eq:sfg_lambertw_eq}
    x\log x=a.
\end{equation}
We know that for real numbers $y, a$ the solution to
\begin{equation}\label{eq:sfg_lambertw_maineq}
    ye^y=a,
\end{equation}
is given by:
\begin{equation*}
    y = \begin{cases}
    W_0(a)&\textrm{if } a\ge 0,\\
    W_{-1}(a)&\textrm{if } -\frac{1}{e}\le a<0,
    \end{cases}
\end{equation*}
where $W_0$ and $W_{-1}$ are the principal and non-principal branches of the Lambert $W$ function, and have the following asymptotic expansions \cite{corless1996lambertw}:
\begin{subequations}
\label{eq:sfg_lambertw_asym}
\begin{align}
    &W_0(x)\approx \log x - \log\log x &\textrm{ (for large }x>0\textrm{)}, \\
    &W_{-1}(-x)\approx \log (-x) - \log(-\log (-x)) &\textrm{(for small }x<0\textrm{)}.
\end{align}
\end{subequations}
Consequently, the solution for Eq. \eqref{eq:sfg_lambertw_eq} is given by substituting $y$ with $\log x$ in Eq. \eqref{eq:sfg_lambertw_maineq} to obtain:
\begin{equation}
    \label{eq:sfg_lambertw}
    \begin{split}
    x &= \begin{cases}
    e^{W_0(a)}&\textrm{if } a\ge 0,\\
    e^{W_{-1}(a)}&\textrm{if } -\frac{1}{e}\le a<0
    \end{cases}
    \\&=\begin{cases}
    \frac{a}{W_0(a)}&\textrm{if } a\ge 0,\\
    \frac{a}{W_{-1}(a)}&\textrm{if } -\frac{1}{e}\le a<0
    \end{cases}
    \\&\approx\begin{cases}
    \frac{a}{\log a - \log\log a}&\textrm{for large positive }a,\\
    \frac{-a}{\log(-\log (-a))-\log(- a)}&\textrm{for small negative }a,
    \end{cases}
    \end{split}
\end{equation}
where in the first equality we use the identity satisfied by the Lambert $W$ function in Eq. \eqref{eq:sfg_lambertw_maineq}, and in the second one we use the asymptotic expression from Eq. \eqref{eq:sfg_lambertw_asym}. Given that formulations in terms of the network's mean degree are easier to interpret, we consider the special cases of scale-free graphons and use Eqs. \eqref{eq:sfg_degree_deg}, \eqref{eq:sfg_eta_apdx} along with the constraint $h\ll 1$, to express $h$ and $\eta$ in terms of $\avgdegree,\beta$ as follows:
\begin{enumerate}
    \item  For ER graphs, we have $\theta\to\infty\implies\alpha\to 0$:
    \begin{equation}
        \label{eq:sfg_dphi_er}
        \begin{split}
            &\avgdegree \to \beta,\enspace\eta\to \beta,\\
            \implies&\eta = v. 
        \end{split}
    \end{equation}
    \item For BA graphs the power law exponent is $\theta=3$ \cite{barabasi1999bagraph}, and therefore $\alpha=1/2$
    \begin{equation}
        \label{eq:sfg_dphi_ba}
        \begin{split}
            &\avgdegree =\beta\left(2\frac{\sqrt{h}}{1+\sqrt{h}}\right)^2\approx 4\beta h,\enspace
            \eta\approx \beta h\log h^{-1},\\
            \implies& h \approx \frac{\avgdegree}{4\beta},\\
            \implies&\eta \approx \frac{\avgdegree}{4}\log\frac{4\beta}{\avgdegree}. 
        \end{split}
    \end{equation}
    \item For ``highly scale-free'' networks where power law exponent $\theta\to 2\implies \alpha\to 1$:
    \begin{equation}
        \label{eq:sfg_dphi_hsf}
        \begin{split}
            &\avgdegree\approx \beta(h\log h^{-1})^2,\enspace\eta= \beta h,\\
            \implies& h\log h = -\sqrt{\frac{\avgdegree}{\beta}}
            \\\implies& h\approx \frac{\sqrt{\frac{\avgdegree}{\beta}}}{\log\log\sqrt{\frac{\beta}{\avgdegree}}+\log\sqrt{\frac{\beta}{\avgdegree}}}\approx \frac{\sqrt{\frac{\avgdegree}{\beta}}}{\log\sqrt{\frac{\beta}{\avgdegree}}},\\
            \implies&\eta \approx \frac{\sqrt{\beta \avgdegree}}{\log\sqrt{\frac{\beta}{\avgdegree}}},
        \end{split}
    \end{equation}
   where we use the result in Eq. \eqref{eq:sfg_lambertw}, and asymptotically $\log\log n\ll \log n$.
\end{enumerate}

\section{\label{sec:apdx_gcsbm}MLE of labeled networks as SBMs}

Suppose we have a given network represented by the adjacency matrix $\mat{A}$ with known node labels given by the assignment matrix $\mat{Z}$. The most trivial label assignment would be to assume a single block to which every node belongs, i.e. $k=1$, which assumes an underlying ER graph model. For many empirical networks, labels can be retrieved from domain knowledge---like nodes in a social network will have associated socio-demographic characteristics. Assuming that the network is generated by an SBM with parameters $\boldsymbol{\pi}, \mat{B}$, we can use the model definition to obtain a maximum likelihood estimate (MLE) of the parameters. For ease of notation let $\kappa:\{1,2,\dots, n\}\to\{1,2,\dots, k\}$ define the label assignment for a given network $\mat{A}$. Then the likelihood function is given by:
\begin{equation}
    \label{eq:sbm_likelihood}
    \begin{split}
    L(\boldsymbol{\pi},\mat{B}) = &\prod_{i\ne j}\condprob{A_{ij}}{Z_{i\kappa(i)}=1,Z_{j\kappa(j)}=1,B_{\kappa(i)\kappa(j)}}\\&\times\prod_i\condprob{Z_{i\kappa(i)}=1}{\boldsymbol{\pi}} \\
    =&\prod_{i\ne j}\left(\frac{B_{\kappa(i)\kappa(j)}}{n}\right)^{A_{ij}}\left(1-\frac{B_{\kappa(i)\kappa(j)}}{n}\right)^{1-A_{ij}}\prod_i\pi_{\kappa(i)}.
    \end{split}
\end{equation}
To maximize the log-likelihood, subject to the constraint $\sum_x\pi_x=1$, apply the method of Lagrange multipliers. Define the Lagrangian $\mathcal{L}(\boldsymbol{\pi},\mat{B},\zeta)\triangleq\log L(\boldsymbol{\pi},\mat{B})+\zeta\left(1-\sum_x\pi_x\right)$, where $\zeta$ is the Lagrange multiplier, which from Eq. \eqref{eq:sbm_likelihood} is given by:
\begin{equation}
    \label{eq:sbm_mle_likelihood}
    \begin{split}
    \mathcal{L}(\boldsymbol{\pi},\mat{B},\zeta) = &\sum_{x,y}\Bigg[ \idx{\mat{Z}^T\mat{AZ}}{xy}\log\left(\frac{B_{xy}}{n}\right) \\&+\idx{\mat{Z}^T\left(\ones{n}\ones{n}^T-\mat{A}-\eyes{n}\right)\mat{Z}}{xy}\log\left(1-\frac{B_{xy}}{n}\right)\Bigg]\\&+\sum_{x} \left(\idx{\mat{Z}^T\ones{n}}{x} \log\pi_x-\zeta\pi_x\right)+\zeta\\
    \end{split}
\end{equation}
For finding the stationary points, consider the first-order partial derivatives of Eq. \eqref{eq:sbm_mle_likelihood}:
\begin{subequations}
    \label{eq:sbm_mle_lagrange_1}
    \begin{align}
        \frac{\partial\mathcal{L}}{\partial B_{xy}} = & \frac{\idx{\mat{Z}^T\mat{AZ}}{xy}}{B_{xy}}-\frac{\idx{\mat{Z}^T\left(\ones{n}\ones{n}^T-\mat{A}-\eyes{n}\right)\mat{Z}}{xy}}{n-B_{xy}}\\
        \frac{\partial\mathcal{L}}{\partial \pi_x} =&\frac{\idx{\mat{Z}^T\ones{n}}{x}}{\pi_x}-\zeta\\
         \frac{\partial\mathcal{L}}{\partial \zeta}=&1-\sum_x\pi_x
    \end{align}
\end{subequations}
Setting the partial derivatives with respect to the parameters to zero, i.e. RHS of Eq. \eqref{eq:sbm_mle_lagrange_1} to $0$, along with looking at the Hessian condition, yields:
\begin{subequations}
    \label{eq:sbm_mle}
    \begin{align}
        \label{eq:sbm_mle_pi}
        \hat{\boldsymbol\pi} &= \frac{\mat{Z}^T\ones{n}}{n},\\
        \label{eq:sbm_mle_b}
        \hat{\mat{B}} &= n\mat{Z}^T\mat{AZ}\oslash\left\{\mat{Z}^T\left(\ones{n}\ones{n}^T-\eyes{n}\right)\mat{Z}\right\},
    \end{align}
\end{subequations}
where $\ones{n}$ is an all-ones vector of length $n$, $\eyes{n}$ is the identity matrix of size $n\times n$, and $\oslash$ is element-wise division. Evidently, $\hat{\boldsymbol{\pi}}$ corresponds to the normalized total node-count of every block, and $\hat{\mat{B}}$ to the normalized total edge-count of every block pair, which is useful especially when we may not have access to the full adjacency structure. This formalism allows us to define an SBM for any given $(\mat{A},\mat{Z})$ pair, regardless of how $\mat{Z}$ itself was inferred, with no computational cost except for summation of edge counts.

\section{\label{sec:apdx_asymmetric}Asymmetric connectivity kernel}

As discussed in Sec. \ref{sec:spd_supercritical}, there can be non-trivial differences between the undirected and directed setting as the latter allows for asymmetry in the connectivity kernel. Consequently, the concept of a ``giant component'' is more subtle for directed networks: it is not necessary for a (directed) path to exist from $i$ to $j$, even if one exists from $j$ to $i$. This yields two generalizations of the giant component: the set of nodes of size $\bigtheta{n}$ that are reachable from a node constitute its giant out-component, and those from which a node can be reached constitute its giant in-component \cite{broder2011graphweb, newman2001random}. For directed networks, we say that the node pair $(i,j)$ is supercritical if asymptotically there can exist a giant in-component of $j$ such that $i$ is on it, or equivalently if asymptotically there can exist a giant out-component of $i$ such that $j$ is on it. The definition of a subcritical node pair arises via negation. We emphasize that supercriticality is symmetric and transitive in an undirected setting, whereas it is only transitive in a directed setting. Consequently, it is possible for $(i,j)$ and $(k,j)$ to be supercritical, but $(i,k)$ and/or $(k,i)$ to be subcritical.


We consider the GLD on the in-/out-components of a directed network, first in the supercritical regime for node pair $(i,j)$ in a sparse ensemble average network (SEAN). By the same argument as in Sec. \ref{sec:spd_supercritical}, without loss of generality, we assume that the expected adjacency matrix $\avg{\mat{A}}$ is not permutation similar to a block diagonal matrix. (In the directed setting, this is a weaker condition than $\avg{\mat{A}}$ being irreducible, i.e. permutation similar to a block triangular matrix.) Let $\phi_{ij}^-$ be the event that node $i$ is on the putative giant in-component of $j$ (note the explicit dependence on both $i,j$), and $\phi_{ji}^+$ be the event that node $j$ is on the putative giant out-component of $i$. We remark that $\phi_{jj}^-\equiv\phi_{jj}^+$ is the event that $j$ is on a giant strongly connected component of the network, which is a connected component of size $\bigtheta{n}$ such that every node is reachable from every other node on it. Analogous to Eq. \eqref{eq:def_omega_psi}, we can define the conditional PMF and survival function matrices, but when conditioning on $\phi_{ij}^-$ instead of $\phi_i$. Assuming conditionally independent edges, sparsity, and no bottlenecks, the recursive setup of Sec. \ref{sec:spd_supercritical}, and the result in Lemma \ref{lemma:1}, can be applied verbatim to directed networks. However, the initial condition for the conditional PMF matrix $[\Omegaaf_1]_{ik}=\condprob{A_{ik}=1}{\phi_{ij}^{-}}$ will be different from Eq. \eqref{eq:prob_connect_exact}. First, we extend the definitions of percolating, non-percolating, and dangling nodes to the directed setting, by considering the probability of nodes being on the giant in-component of a given node $j$ in network $G$. We remark that Lemma \ref{lemma:perccontribs} and Corollary \ref{lemma:perc} analogously follow in this setting. Second, analogous to Lemma \ref{lemma:deg_gcc} in the undirected setting, we derive connection probabilities that condition on the source node being on the giant in-component of a given node $j$. For the following lemma, we use $G^{\setminus ij}$ to indicate a subgraph with the edge $A_{ij}$ unconsidered, and explicate the dependence of the percolation event $\phi_{ij}^-$ on a network $G$ by writing it as $\phi_{ij}^-(G)$.

\begin{lemma}[Connection probability to giant in-component(s)]\label{lemma:deg_gcc_directed}
    For nodes $i,j,k$ in a sparse ensemble average network $G$ that is directed, if $i$ is a percolating node with respect to $j$ ($\prob{\phi_{ij}^-(G)}=\bigomega{1}$) and $k$ is a percolating or non-percolating node with respect to $j$ ($\prob{\phi_{kj}^-(G)}=\bigomega{1}$ or $0$), then asymptotically:
    \begin{equation}\label{eq:prob_exact_connect_directed}
    \begin{split}
        \condprob{A_{ik}=1}{\phi_{ij}^-(G)}&\approx \prob{A_{ik}=1}\times\\&\left\lbrace 1+\left[\frac{1}{\prob{\phi_{ij}^-(G)}}-1\right]\prob{\phi_{kj}^-(G)}\right\rbrace.
    \end{split}
    \end{equation}
\end{lemma}
\begin{proof}
    The proof is very similar to that of Lemma \ref{lemma:deg_gcc}. Without loss of generality, generate the network $G$ such that the edge from node $i$ to $k$ is generated in the end, which will be independent of anything else (Eq. \eqref{eq:ciem}), and consider the probability 
    \begin{equation*}
        \begin{split}
        &\prob{A_{ik}=1,\phi_{ij}^-(G),\phi_{kj}^-(G)}\\&=\prob{A_{ik}=1,\phi_{ij}^-(G^{\setminus{ik}}),\phi_{kj}^-(G^{\setminus{ik}})} \\&\quad+ \prob{A_{ik}=1,\neg\phi_{ij}^-(G^{\setminus{ik}}),\phi_{kj}^-(G^{\setminus{ik}})}\\
            &=\prob{A_{ik}=1}\prob{\phi_{kj}^-(G^{\setminus{ik}})}.
        \end{split}
    \end{equation*}
    Similarly consider the probability 
    \begin{align*}
        &\prob{A_{ik}=1,\phi_{ij}^-(G),\neg\phi_{kj}^-(G)}\\&=\prob{A_{ik}=1,\phi_{ij}^-(G^{\setminus{ik}}),\neg\phi_{kj}^-(G^{\setminus{ik}})}\\&=\prob{A_{ik}=1}\prob{\phi_{ij}^-(G^{\setminus{ik}}),\neg\phi_{kj}^-(G^{\setminus{ik}})}
    \end{align*}
    Analogously to Lemma \ref{lemma:perccontribs}, the percolation events of two percolating nodes with respect to $j$ in $G$ are asymptotically independent, that is, $\prob{\phi_{ij}^-(G^{\setminus{ik}}),\neg\phi_{kj}^-(G^{\setminus{ik}})}\approx\prob{\phi_{ij}^-(G^{\setminus{ik}})}\prob{\neg\phi_{kj}^-(G^{\setminus{ik}})}$. Substituting in the previous equation and adding the last two equations yields:
    \begin{align*}
         \prob{A_{ik}=1,\phi_{ij}^-(G)}&\approx \prob{A_{ik}=1}\\&\times\bigg[\prob{\phi_{ij}^-(G^{\setminus{ik}})}+\prob{\phi_{kj}^-(G^{\setminus{ik}})}\\&-\prob{\phi_{ij}^-(G^{\setminus{ik}})}\prob{\phi_{kj}^-(G^{\setminus{ik}})}\bigg].
    \end{align*}
    Since $i$ and $k$ are, by assumption, non-dangling nodes with respect to $j$ in $G$, we know that removal of any single node has a vanishing effect on their percolation probabilities (analogously to Eq. \eqref{eq:percsubfullgraph}), therefore the removal of a single edge $A_{ik}$ has a vanishing effect on their percolation probabilities: $\prob{\phi_{ij}^-(G^{\setminus ik})}\approx\prob{\phi_{ij}^-(G)}$, $\prob{\phi_{kj}^-(G^{\setminus ik})}\approx\prob{\phi_{kj}^-(G)}$. This yields:
    \begin{align*}
    \prob{A_{ik}=1,\phi_{ij}^-(G)}&\approx\prob{A_{ik}=1}\\&\times\bigg[\prob{\phi_{ij}^-(G)}+\prob{\phi_{kj}^-(G)}\\&-\prob{\phi_{ij}^-(G)}\prob{\phi_{kj}^-(G)}\bigg],
    \end{align*}
    which divided through by $\prob{\phi_{ij}^-(G)}$ gives us the desired expression on the RHS of Eq. \eqref{eq:prob_exact_connect_directed}.
\end{proof}

\paragraph*{Directed percolation probability via the GLD} Finding the RHS of Eq. \eqref{eq:prob_exact_connect_directed} requires an estimate of the ``directed'' percolation probability $\prob{\phi_{ij}^-}$. Using arguments of Sec. \ref{sec:perc_prob}, i.e. by continuity of the GLD, the amount of probability mass at $\lambda_{ij}=\infty$ when $i$ is on the putative giant in-component of $j$ is given by the steady state of Eq. \eqref{eq:spd_main_psi} in the directed setting:
\begin{equation}\label{eq:perc_prob_limit_directed}
    \condprob{\lambda_{ij}=\infty}{\phi_{ij}^-} = \lim_{l\to\infty} \condprob{\lambda_{ij}>l}{\phi_{ij}^-}.
\end{equation}
Since $i$ is already on a putative giant in-component of $j$, the event $\lambda_{ij}=\infty\ |\ \phi_{ij}^-$ corresponds to $j$ not being on this putative giant out-component of $i$:
\begin{equation}\label{eq:perc_prob_pinf_directed_out}
    \prob{\phi_{ji}^+}=1-\condprob{\lambda_{ij}=\infty}{\phi_{ij}^-}.
\end{equation}
We note that the probability of $i$ being on the putative giant in-component of $j$ is the same as the probability of $i$ being on the putative giant out-component of $j$ when the direction of all edges in the network is reversed. That is, with the conditioning on the average adjacency matrix made explicit:
\begin{equation}\label{eq:perc_prob_directed_out_in}
    \condprob{\phi_{ij}^-}{\avg{\mat{A}}}=\condprob{\phi_{ij}^+}{\avg{\mat{A}}^T}.
\end{equation}
Therefore, Eqs. \eqref{eq:spd_main}, \eqref{eq:prob_exact_connect_directed}, \eqref{eq:perc_prob_limit_directed}, \eqref{eq:perc_prob_pinf_directed_out} and \eqref{eq:perc_prob_directed_out_in} provide us with the full analytic form of the GLD for directed networks. Since precision in the initial condition is not important for the steady state of Eq. \eqref{eq:spd_main}, we can use the na\"{i}ve initial condition
\begin{equation}
    \label{eq:naive_init_directed}
    \condprob{A_{ij}=1}{\phi_{ij}^-}\approxnaive\prob{A_{ij}=1}
\end{equation}
for the recursive setup given the transposed average adjacency matrix $\avg{\mat{A}}^T$ to first obtain the directed percolation probability of every supercritical node pair, and then run the recursive setup a second time given $\avg{\mat{A}}$ and the exact initial condition in Eq. \eqref{eq:prob_exact_connect_directed}.

\paragraph*{Asymmetric supercriticality} As described above, a directed setting that permits asymmetric kernels can result in non-trivial deviations from an undirected setting, particularly when supercriticality is asymmetric. For instance, consider a directed ``chain'' stochastic block model (SBM) with $k$ blocks, whose block matrix has non-zero entries on the diagonal---allowing for nodes to connect to other nodes within the block---and on the upper off-diagonal---allowing for nodes to connect to nodes in the next adjacent block:
\begin{equation*}
    \mat{B} = 
    \begin{bmatrix}
    a&b&0&\cdots&0\\
    0&a&b&\ddots&\vdots\\
    \vdots&\ddots&\ddots&\ddots&0\\
    \vdots&&\ddots&a&b\\
    0&\cdots&\cdots&0&a
    \end{bmatrix}.
\end{equation*}
Evidently for $i<j$, (directed) paths can exist from nodes of block $i$ to block $i$, and block $i$ to block $j$, but they are entirely prohibited from nodes of block $j$ to block $i$. Given this directed chain structure we can expect $\prob{\phi_{ij}^-}<\prob{\phi_{ik}^-}$ for block indices $i<j<k$. Fig. \ref{fig:percolation_pipesbm} shows the variation in directed percolation probabilities in one such chain SBM, using Eqs. \eqref{eq:spd_main}, \eqref{eq:naive_init_directed}, \eqref{eq:perc_prob_limit_directed}, \eqref{eq:perc_prob_pinf_directed_out} and \eqref{eq:perc_prob_directed_out_in}.

\begin{figure}[t]
    \centering
    \includegraphics[width=\columnwidth]{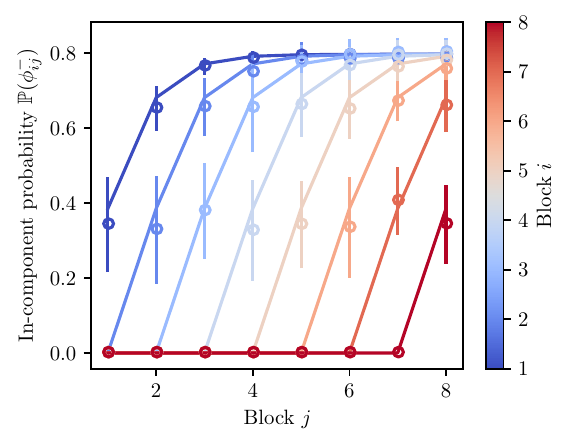}
    \caption{Empirical and analytic estimates of directed percolation probabilities agree with each other for an asymmetric ``chain'' SBM (with $8$ equi-sized blocks, values $10$ and $6$ on the diagonal and upper off-diagonal entries of the block matrix, respectively). Colors correspond to the block membership of source node, while $x$-axis corresponds to the block membership of target node. Solid lines represent analytic form using Eqs. \eqref{eq:spd_main}, \eqref{eq:naive_init_directed}, \eqref{eq:perc_prob_limit_directed}, \eqref{eq:perc_prob_pinf_directed_out} and \eqref{eq:perc_prob_directed_out_in}, while markers and bars indicate empirical estimates: mean and standard error over 20 network samples.}
    \label{fig:percolation_pipesbm}
\end{figure}

\paragraph*{Symmetric supercriticality} Consider the scenario where supercriticality is symmetric for all node pairs $(i,j)$: asymptotically, if there can exist a giant in-component of $j$ containing $i$, then there can exist a giant in-component of $i$ containing $j$, and vice-versa. This implies that the putative giant in- and out-components of all nodes coalesce into a single giant \emph{strongly connected} component, on which every node can reach every other node. Then we can asymptotically drop the dependence on the ``target'' node from the directed percolation events and write $\forall j:\phi_{ij}^-=\phi_i^-,\phi_{ij}^+=\phi_i^+$, where $\phi_i^-$ ($\phi_i^+$) is the event that node $i$ has an edge to (from) the giant strongly connected component. This renders the initial condition for directed networks to be exactly the same as that for undirected networks, Eq. \eqref{eq:prob_connect_exact}, except where $\phi_i$ is substituted by $\phi_i^-$. If we further assume that the average adjacency matrix is symmetric, then Eqs. \eqref{eq:perc_prob_directed_out_in} and \eqref{eq:perc_prob_pinf_directed_out} yield $\prob{\phi_j^-}=1-\condprob{\lambda_{ij}=\infty}{\phi_i^-}$, which is identical to Eq. \eqref{eq:perc_prob_pinf} for undirected networks, except where $\phi_i$ is substituted by $\phi_i^-$. It then follows that the GLD for a network generated using a symmetric $\avg{\mat{A}}$ does not depend on the network's directedness. This is evident in Fig. \ref{fig:spd_sbm_asym}, where we plot the empirical and analytic GLD for a directed network generated with an asymmetric kernel with symmetric supercriticality, by using an SBM with an asymmetric block matrix. We emphasize that if $\avg{\mat{A}}$ is symmetric, (or equivalently the connectivity kernel $\nu$ in the sparse general random networks framework of Sec. \ref{sec:general_graphs} is symmetric,) then supercriticality is symmetric for all node pairs, but the converse need not be true.

\paragraph*{Percolation probability via self-consistent equation} We can also derive percolation probabilities without invoking the GLD. For $i$ to not be on the giant in-component of $j$, it must not have a directed edge to any node $k$ such that $k$ is on the giant in-component of $j$. (This exemplifies the transitivity of supercriticality so-defined.) Therefore, the self-consistent Eq. \eqref{eq:gcc_consistency} yields percolation probabilities in the directed setting as well, except where $\rho_i= \prob{\phi_{ij}^-}$. Given the dependence on $j$---and unlike the undirected setting---in the worst case there can be as many non-trivial solutions to the self-consistent equation as the number of nodes, all of which need to be found. However, the GLD approach to derive percolation probabilities extracts all solutions in one sweep, encoded in the steady state of the survival function matrix, which is a significant advantage. If supercriticality is symmetric for all node pairs, there can only be a unique giant strongly connected component, and there can exist only a unique non-trivial solution to the self-consistent equation---as in the undirected setting. 

\paragraph*{Subcritical regime} We now consider a node pair $(i,j)$ that is subcritical, wherein asymptotically there cannot exist a giant in-component of $j$ containing $i$, or equivalently a giant out-component of $i$ containing $j$. In the asymptotic limit this yields that either, trivially, there exist no paths from $i$ to $j$, or that $i$ can reach $j$ through paths only on a small in-/out-component. This results in the subcriticality condition:
\begin{equation*}
    \prob{\phi_{ij}^-}=\prob{\phi_{ji}^+}=0,
\end{equation*}
in which case we consider the GLD on the small in-/out-component containing $i,j$. By arguments of Sec. \ref{sec:spd_subcritical}, the recursive setup remains identical, with the difference arising in the initial condition for the conditional PMF, now given by Eq. \eqref{eq:init_omega_subcritical}. That is, we obtain the same expressions for the GLD---Eqs. \eqref{eq:spd_main} and \eqref{eq:init_omega_subcritical}---under the subcritical regime in directed and undirected networks.
\begin{figure*}
    \centering
    \subfloat[Directed graph using $\mat{B}$]{\label{fig:spd_sbm_asym_directed} \includegraphics[width=\columnwidth]{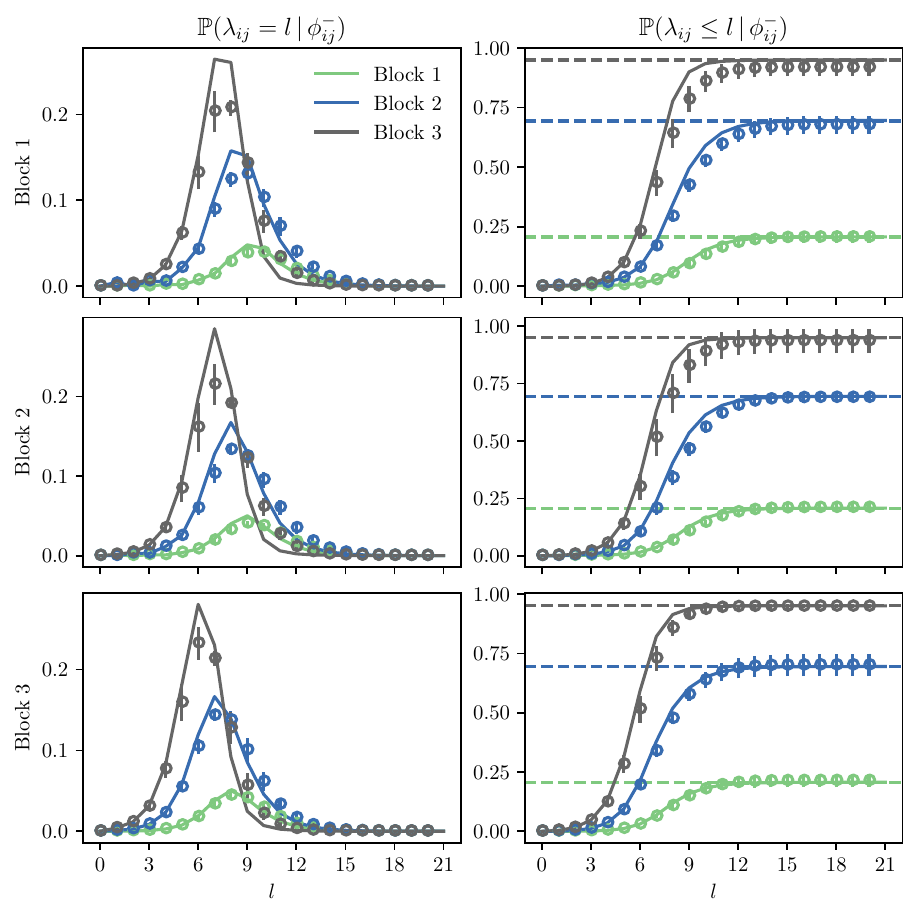}}
    \subfloat[Directed graph using $\mat{B}^T$]{\label{fig:spd_sbm_asym_directed_T} \includegraphics[width=\columnwidth]{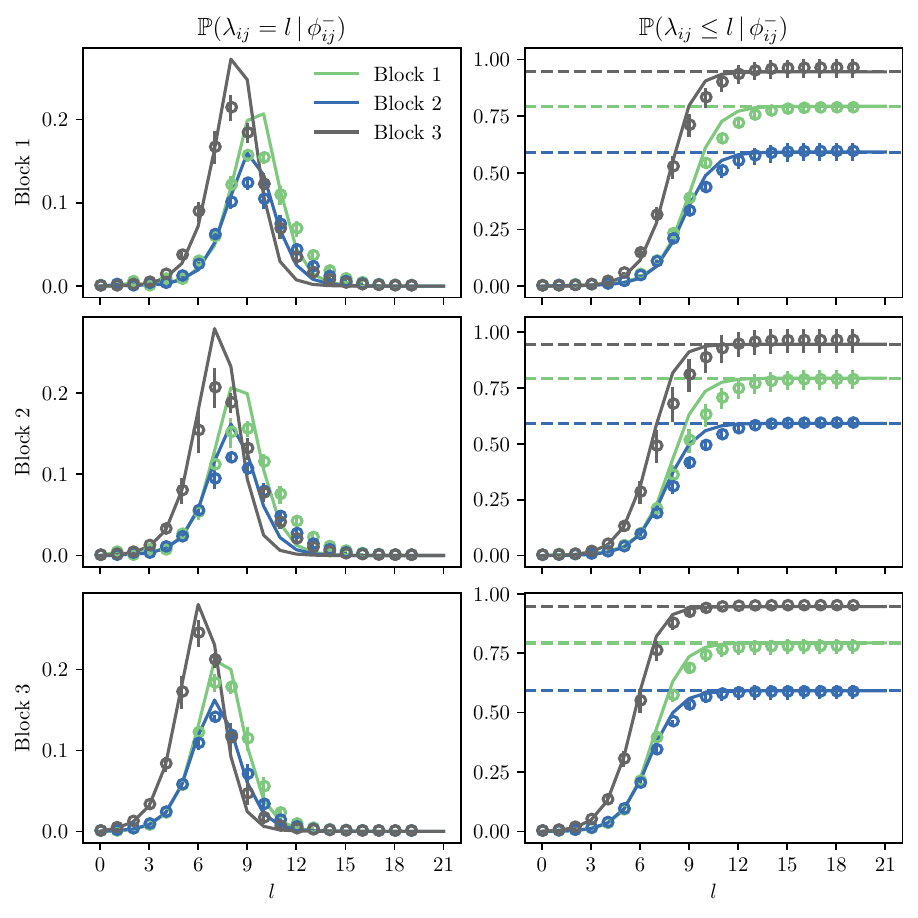}}\\
    \subfloat[Directed graph using $\frac{\mat{B}+\mat{B}^T}{2}$]{\label{fig:spd_sbm_asym_sym_directed} \includegraphics[width=\columnwidth]{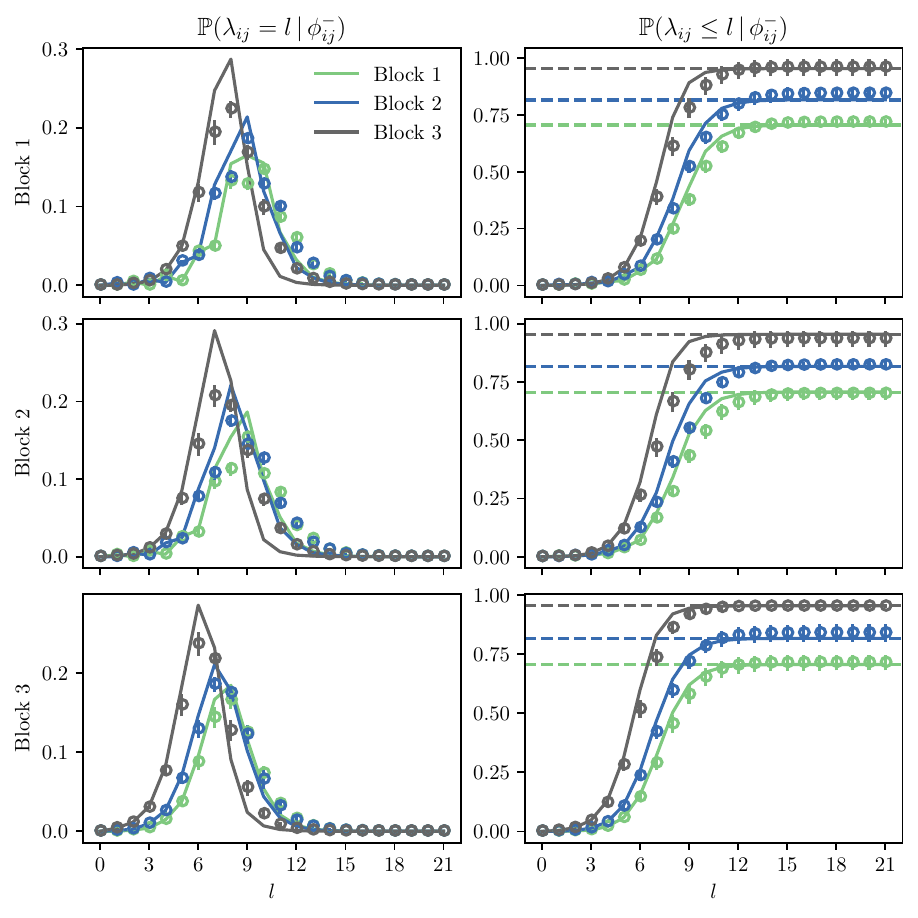}}
    \subfloat[Undirected graph using $\frac{\mat{B}+\mat{B}^T}{2}$]{\label{fig:spd_sbm_asym_sym_undirected} \includegraphics[width=\columnwidth]{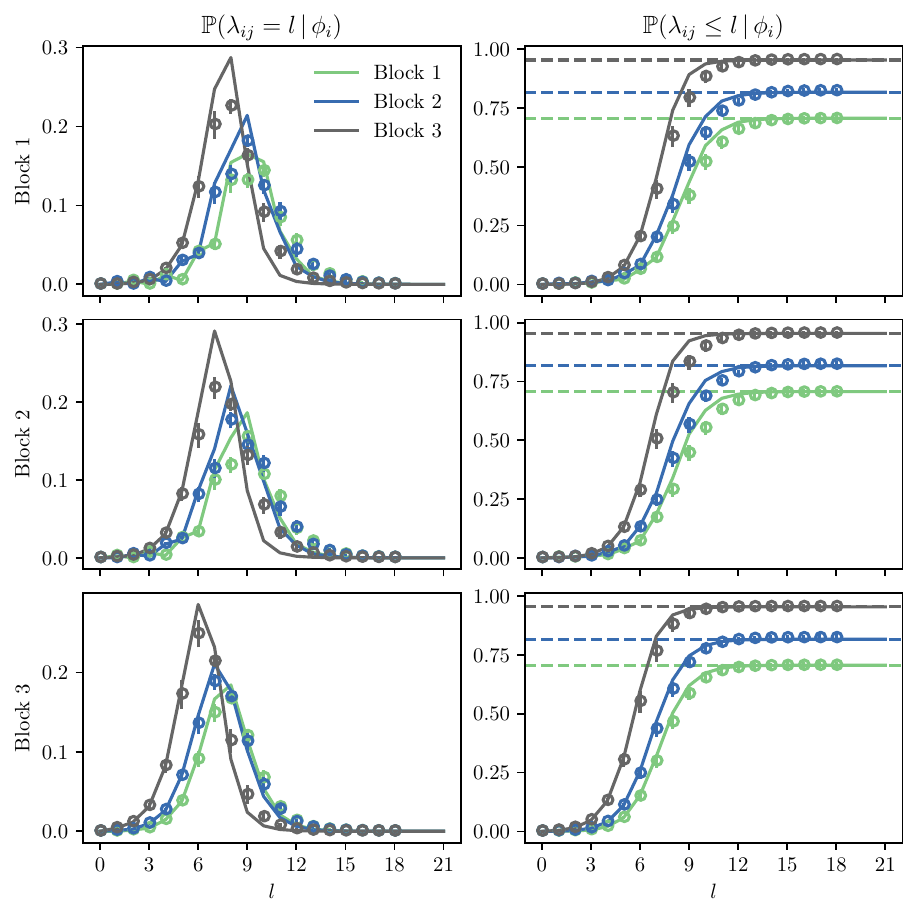}}
    \caption{Empirical and analytic distribution of geodesic lengths for directed (and undirected) graphs generated via an SBM with asymmetric (and symmetrized) block matrix $\mat{B}=\Big(\begin{smallmatrix} 0 & 8 & 0\\ 1 & 0 & 2\\ 0 & 2 & 8 \end{smallmatrix}\Big)$, distribution vector $\boldsymbol\pi=\left(\frac{1}{3}, \frac{1}{3}, \frac{1}{3}\right)$, and $n=2048$ are in good agreement. In each subplot, rows correspond to the block membership of source node, left column depicts the PMF, and the right column depicts the CDF. Solid lines represent analytic form using Eqs. \eqref{eq:spd_main}, \eqref{eq:prob_exact_connect_directed}, \eqref{eq:perc_prob_limit_directed}, \eqref{eq:perc_prob_pinf_directed_out} and \eqref{eq:perc_prob_directed_out_in}. Symbols ($\circ$) and bars indicate empirical estimates: mean and standard error over 10 network samples. Dashed asymptotes indicate percolation probabilities from Eq. \eqref{eq:gcc_consistency}.}
    \label{fig:spd_sbm_asym}
\end{figure*}

\paragraph*{Closed form of the GLD}
From Eq. \eqref{eq:sf_avg}, the closed form of the survival function matrix for the ensemble average model can be written entirely in terms of powers of $\avg{\mat{A}}$. Since the kernel $\nu$ is possibly asymmetric, $\avg{\mat{A}}$ may not be a symmetric matrix. Consequently, $\avg{\mat{A}}$ may not necessarily be a normal matrix, and therefore may not be diagonalis=zed by a unitary matrix. Analogously, in terms of the integral operator $T$ defined in Eq. \eqref{eq:integral_op}, $T$ may not be self-adjoint and therefore the spectral theorem cannot be applied, which prevents us from writing closed-form expressions for the GLD as in Eqs. \eqref{eq:spd_analytic_general_eig}, \eqref{eq:spd_analytic_general_eig_uncorrected}. However, by making use of the Jordan normal form, the closed-form survival function of the GLD can still be expressed as a weighted sum of powers of eigenvalues over a generalized eigenvector basis. In particular, let $\avg{\mat{A}}=\mat{PJP}^{-1}$, where $\mat{P}$ is an invertible matrix and $\mat{J}$ is a block diagonal matrix with $k$ blocks---not to be confused with the ``blocks'' of the stochastic block model---such that the $i^\mathrm{th}$ block denoted by $\mat{J}_i$, of size $n_i\times n_i$, corresponds to the $i^\mathrm{th}$ eigenvalue of $\avg{\mat{A}}$, denoted by $\tau_i$, and can be written as the sum of a scaled identity matrix and the canonical nilpotent matrix:
\begin{equation*}
\begin{split}
    \mat{J} &= 
    \begin{bmatrix}
    \mat{J}_{1}&0&\cdots&0\\
    0&\mat{J}_{2}&\ddots&\vdots\\
    \vdots&\ddots&\ddots&0\\
    0&\cdots&0&\mat{J}_{k}
    \end{bmatrix},\\
    \mat{J}_i &= 
    \begin{bmatrix}
    \tau_i&1&0&\cdots&0\\
    0&\tau_i&1&\ddots&\vdots\\
    \vdots&\ddots&\ddots&\ddots&0\\
    \vdots&&\ddots&\tau_i&1\\
    0&\cdots&\cdots&0&\tau_i
    \end{bmatrix}.
\end{split}
\end{equation*}
We remark that for an eigenvalue $\tau_i$, the number of Jordan blocks corresponding to it gives its geometric multiplicity, while the sum of sizes of Jordan blocks corresponding to it gives the algebraic multiplicity \cite{johnson1985matrix}. If $\mat{J}$ is diagonal, then this implies that $\avg{\mat{A}}$ is diagonalizable, and therefore $\mat{P}$ simply corresponds to the right eigenvectors of $\avg{\mat{A}}$ stacked as column vectors. In that case, the Jordan normal form is simply the eigendecomposition of $\avg{\mat{A}}$. However, since we consider a scenario where $\avg{\mat{A}}$ is not a symmetric matrix, it need not be diagonalized in such a manner. However, the Jordan normal form \emph{almost} diagonalizes it, and in this case $\mat{P}$ corresponds to a generalized set of eigenvectors of $\avg{\mat{A}}$ \cite{johnson1985matrix}, and yields: $$\avg{\mat{A}}^m = \mat{P}\mat{J}^m\mat{P}^{-1}.$$ Due to the diagonal block structure of $\mat{J}$, its powers can be written as the corresponding diagonal-block matrix of the powers of its blocks. Then for the $i^\mathrm{th}$ block, due to its structure described above, its powers for $m\in\integernonneg$ can be written in closed form as an upper-triangular matrix such that $\idx{\mat{J}_i^m}{pq}=\binom{m}{q-p}\tau_i^{m+p-q}$ when $m\ge q-p\ge 0$, and $0$ otherwise \cite{higham2008matrixfunctions}. Here, $\binom{m}{q-p}=\frac{m!}{(q-p)!(m+p-q)!}$ is the usual binomial coefficient. We can also express $\avg{\mat{A}}^m$, $\mat{P}$ and $\mat{P}^{-1}$ as block matrices:
\begin{equation*}
\begin{split}
    \avg{\mat{A}}^m &= 
    \begin{bmatrix}
    \idxp{\avg{\mat{A}}^m}{11}&\idxp{\avg{\mat{A}}^m}{12}&\cdots&\idxp{\avg{\mat{A}}^m}{1k}\\
    \idxp{\avg{\mat{A}}^m}{21}&\idxp{\avg{\mat{A}}^m}{22}&\cdots&\idxp{\avg{\mat{A}}^m}{2k}\\
    \vdots&\vdots&\ddots&\vdots\\
    \idxp{\avg{\mat{A}}^m}{k1}&\idxp{\avg{\mat{A}}^m}{k2}&\cdots&\idxp{\avg{\mat{A}}^m}{kk}
    \end{bmatrix},\\
    \mat{P} &= 
    \begin{bmatrix}
    \mat{P}_{11}&\mat{P}_{12}&\cdots&\mat{P}_{1k}\\
    \mat{P}_{21}&\mat{P}_{22}&\cdots&\mat{P}_{2k}\\
    \vdots&\vdots&\ddots&\vdots\\
    \mat{P}_{k1}&\mat{P}_{k2}&\cdots&\mat{P}_{kk}
    \end{bmatrix},\\
    \mat{P}^{-1} &= 
    \begin{bmatrix}
    \idxp{\mat{P}^{-1}}{11}&\idxp{\mat{P}^{-1}}{12}&\cdots&\idxp{\mat{P}^{-1}}{1k}\\
    \idxp{\mat{P}^{-1}}{21}&\idxp{\mat{P}^{-1}}{22}&\cdots&\idxp{\mat{P}^{-1}}{2k}\\
    \vdots&\vdots&\ddots&\vdots\\
    \idxp{\mat{P}^{-1}}{k1}&\idxp{\mat{P}^{-1}}{k2}&\cdots&\idxp{\mat{P}^{-1}}{kk}
    \end{bmatrix},
\end{split}
\end{equation*}
where we use the notation $\idxp{\mat{X}}{rs}$ to refer to the $(r,s)^\mathrm{th}$ block of a block/partition matrix $\mat{X}$ to avoid any confusion where it might arise (for instance, when $\mat{X}$ has an associated super/subscript or is a product or power of matrices). Therefore, we can express powers of $\avg{\mat{A}}$ in a block-wise manner. For any two nodes $a,b$ we can define their block indices $r,s\in\{1,2,\dots, k\}$, and within-block indices $u\in\{1,2,\dots, n_r\}$, $v\in\{1,2,\dots, n_s\}$, such that:
\begin{equation}
\label{eq:power_ensemble_avg_asym}
    \begin{split}
        &\idxp{\avg{\mat{A}}^m}{rs} = \sum_{i=1}^k\mat{P}_{ri}\mat{J}_i^m\idxp{\mat{P}^{-1}}{is}\\
        &\implies \idx{\idxp{\avg{\mat{A}}^m}{rs}}{uv} = \sum_{i=1}^k\sum_{p=1,q=1}^{n_i}\idx{\mat{P}_{ri}}{up}\idx{\mat{J}_i^m}{pq}\idx{\idxp{\mat{P}^{-1}}{is}}{qv}\\
        &=\sum_{i=1}^k\sum_{p=1}^{n_i}\sum_{q=p}^{\min(p+m,n_i)}\binom{m}{q-p}\tau_i^{m+p-q}\idx{\mat{P}_{ri}}{up}\idx{\idxp{\mat{P}^{-1}}{is}}{qv}
    \end{split}
\end{equation}
Putting Eq. \eqref{eq:power_ensemble_avg_asym} for the powers of $\avg{\mat{A}}$ in Eqs. \eqref{eq:sf_avg}, \eqref{eq:sf_avg_uncorrected} for the sparse ensemble average network, we have shown that the closed form of the survival function of the GLD of an asymmetric kernel can be expressed as a sum of powers of eigenvalues over a generalized eigenvector basis. In particular, the approximate closed form of the conditional PMF and the survival function for the geodesic length between $a,b$ using Eqs. \eqref{eq:power_ensemble_avg_asym}, \eqref{eq:spd_analytic_omega} and \eqref{eq:sf_avg_uncorrected}, is given block-wise by:
\begin{subequations}
\label{eq:asym_1}
\begin{align}
    \label{eq:asym_1_omega}
    \idx{\idxp{\Omegaacf_l}{rs}}{uv} =
    \begin{split}
    &\sum_{i=1}^k\sum_{p=1}^{n_i}\sum_{q=p}^{\min(p+l,n_i)}{\binom{l}{q-p}}\tau_i^{l+p-q}\\&\hspace{1.5cm}\times\idx{\mat{P}_{ri}}{up}[\idxp{\mat{P}^{-1}}{is}]_{qv},
    \end{split}\\
    \label{eq:asym_1_psi}
    \idx{\idxp{\Psiacf_l}{rs}}{uv} =& \exp\left(-\sum_{m=1}^l\idx{\idxp{\Omegaacf_m}{rs}}{uv}\right).
\end{align}
\end{subequations}

\bibliography{apssamp}

\end{document}